\documentclass[a4paper,onecolumn,10pt,accepted=2019-09-30]{quantumarticle}
\pdfoutput=1
\usepackage[numbers, sort&compress, merge]{natbib}
\usepackage{graphicx}
\usepackage{amsthm}
\usepackage{amsmath}
\usepackage{amssymb}
\usepackage{placeins}
\usepackage{enumerate}
\usepackage{dsfont}

\newcommand{\eq}[1]{Eq.~\hyperref[eq:#1]{(\ref*{eq:#1})}}
\renewcommand{\sec}[1]{\hyperref[sec:#1]{Section~\ref*{sec:#1}}}
\newcommand{\app}[1]{\hyperref[app:#1]{Appendix~\ref*{app:#1}}}
\newcommand{\tab}[1]{\hyperref[tab:#1]{Table~\ref*{tab:#1}}}
\newcommand{\fig}[1]{\hyperref[fig:#1]{Figure~\ref*{fig:#1}}}
\newcommand{\figx}[2]{\hyperref[fig:#1]{Figure~\ref*{fig:#1}(#2)}}
\newcommand{\micro}{RWSWT}
\newcommand{\caltech}{LLDUC}
\newcommand{\step}{k}

\def\ket#1{\mathinner{|{#1}\rangle}}
\newcommand{\nn}{\nonumber\\}
\newcommand{\wid}{M}
\newcommand{\dm}{d}
\newcommand{\chunk}{k}
\newtheorem{theorem}{Theorem}

\input{Qcircuit}

\begin{document}

\title{Qubitization of Arbitrary Basis Quantum Chemistry Leveraging Sparsity and Low Rank Factorization}

\date{\today}
\author{Dominic W. Berry}
\email{dominic.berry@mq.edu.au}
\affiliation{Department of Physics and Astronomy, Macquarie University, Sydney, NSW 2109, Australia}
\orcid{0000-0003-3446-1449}
\author{Craig Gidney}
\affiliation{Google Research, Venice, CA 90291, United States}
\author{Mario Motta}
\affiliation{Division of Chemistry,
California Institute of Technology, Pasadena, CA 91125, United States}
\author{Jarrod R. McClean}
\affiliation{Google Research, Venice, CA 90291, United States}
\orcid{0000-0002-2809-0509}
\author{Ryan Babbush}
\email{babbush@google.com}
\affiliation{Google Research, Venice, CA 90291, United States}
\orcid{0000-0001-6979-9533}
\date{2019-11-27}

\begin{abstract}
Recent work has dramatically reduced the gate complexity required to quantum simulate chemistry by using linear combinations of unitaries based methods to exploit structure in the plane wave basis Coulomb operator. Here, we show that one can achieve similar scaling even for arbitrary basis sets (which can be hundreds of times more compact than plane waves) by using qubitized quantum walks in a fashion that takes advantage of structure in the Coulomb operator, either by directly exploiting sparseness, or via a low rank tensor factorization. We provide circuits for several variants of our algorithm (which all improve over the scaling of prior methods) including one with $\widetilde{\cal O}(N^{3/2} \lambda)$ T complexity, where $N$ is number of orbitals and $\lambda$ is the 1-norm of the chemistry Hamiltonian. We deploy our algorithms to simulate the FeMoco molecule (relevant to Nitrogen fixation) and obtain circuits requiring about seven hundred times less surface code spacetime volume than prior quantum algorithms for this system, despite us using a larger and more accurate active space.
\end{abstract}

\maketitle

\section{Introduction}

Quantum computers were originally proposed as special purpose tools for efficiently modeling physical quantum mechanical systems \cite{Feynman1982}. Ever since then quantum simulation has been central to the study of quantum computing \cite{Lloyd1996} while also regarded as one of its most promising applications. In recent years, progress in quantum hardware has led to great optimism for the field. However, a large gap remains between expectations for the technology and the expected value of the relatively few known applications that appear viable on even a small fault-tolerant quantum computer \cite{Mohseni2017}. This disparity has underscored the importance of estimating and reducing the resources required to implement quantum algorithms within a fault-tolerant cost model.

The most widely studied and anticipated application of quantum simulation is chemistry \cite{Aspuru-Guzik2005}. Most work on this topic has focused on providing solutions to the electronic structure problem by using phase estimation to sample molecular eigenstates and estimate eigenvalues \cite{Kitaev1995,Abrams1997}. Even at small problem sizes of around a hundred qubits, efficient and accurate solutions to this problem could prove transformative for various fields of study and the development of technologies such as better batteries, pharmaceuticals and industrial catalysts.

In order to represent molecular systems on a quantum computer one usually discretizes the many-body wavefunction using a basis of single-particle functions referred to as orbitals. The vast majority of quantum chemistry calculations use either plane wave orbitals, or more elaborate orbitals that are commonly composed of linear combinations of Gaussians.
Plane waves are often chosen in calculations of periodic materials and lead to highly structured Hamiltonians. The work of \cite{BabbushLow} showed that exploiting this structure leads to asymptotic advantages for quantum algorithms. Today, the best-scaling quantum algorithms for chemistry in second quantization use plane waves; with either ${\cal O}(N^3)$ gate complexity (with small constant factors) \cite{BabbushSpectra,Kivlichan2019} or  ${\cal O}(N^2 \log N)$ gate complexity (with large constant factors and more spatial complexity) \cite{Low2018}.

A major limitation to using plane waves in second quantization is that one needs a very large number of spin-orbitals to represent many molecular systems to chemical accuracy. The work of \cite{BabbushContinuum} suggests resolving this problem by simulating the plane wave Hamiltonian in first quantization to achieve $\widetilde{\cal O}(N^{1/3}\eta^{8/3})$ gate complexity, where $\eta$ is the number of electrons. With such low scaling in $N$, one might be able to use an extremely large plane wave basis. Unfortunately, the practicality of that algorithm is unclear because it has not been compiled to explicit circuits, and it is unclear how large the basis would need to be \cite{Low2018}.

The more obvious remedy to the low resolution of plane waves is to use a more compact basis. Indeed, the majority of proposals for the quantum simulation of chemistry focus on using very compact molecular orbitals. However, using molecular orbitals leads to complex Hamiltonians with coefficients defined in terms of integrals and ${\cal O}(N^4)$ distinct terms. As a consequence, the first quantum algorithms in this representation had gate complexity ${\cal O}(N^{11})$ \cite{Whitfield2010,Wecker2014}. Since then, a large community of researchers has worked to significantly reduce the cost of simulation in this representation through tighter bounds \cite{Wecker2014,BabbushTrotter,Poulin2014}, better mappings between fermions and qubits \cite{Seeley2012,Setia2017,Bravyi2017,Steudtner2017,Jiang2018}, improved state preparation techniques \cite{Veis2014,Berry2018,Poulin2017,Tubman2018}, application of new time-evolution strategies \cite{Low2016,BabbushSparse1,Campbell2019}, considerations of fault-tolerant overheads \cite{Jones2012,Reiher2017,Litinski2018} and other representational and algorithmic insights \cite{Kassal2008,Toloui2013,Hastings2015,Sugisaki2016,Motzoi2017,Motta2018}.

The lowest rigorous complexity of prior work on second quantized arbitrary basis chemistry simulation is either the $\widetilde{\cal O}(N^5)$ scaling of \cite{BabbushSparse1}, or the $\widetilde{\cal O}(\lambda^2)$ scaling of \cite{Campbell2019}, where $\lambda$ is the 1-norm of the Hamiltonian. However, the \cite{BabbushSparse1} algorithm suffers from large constant factors in the scaling, and the approach of \cite{Campbell2019} scales quadratically worse than post-Trotter methods with respect to the evolution time. In practice, we expect the most competitive prior method would be Lie-Trotter product formulae \cite{Motta2018}, but the step size for that approach has not been studied.
These results are challenging to compare directly because the scaling of $\lambda$ with respect to $N$ depends on whether $N$ is growing towards the thermodynamic (large system) or continuum (large basis) limit. Here we provide an algorithm with $\widetilde{\cal O}(N^{3/2} \lambda)$ T complexity, which appears better than all prior work so long as $\lambda = \Omega(N^{3/2})$, which is usually the case.

Prior papers to compile a quantum chemistry algorithm to the level of Clifford + T gates and estimate the resources required within an error-correcting code are \cite{Reiher2017,BabbushSpectra,Kivlichan2019}. These papers focus on minimizing T complexity or Toffoli complexity because these gates cannot be transversely implemented within practical codes \cite{Litinski2018,Gidney2018pub}. To implement these gates one must distill magic states or Toffoli states, which takes orders of magnitude more spacetime volume (qubitseconds) than executing Clifford gates and also consumes a very large number of physical qubits \cite{Fowler2012,Fowler2018}.

The work of \cite{Reiher2017} focused on the simulation of an active space of the FeMo cofactor of the Nitrogenase enzyme, aka ``FeMoco'' (stoichiometry ${\rm Fe}_7 {\rm Mo} {\rm S}_9 {\rm C}$). FeMoco is the active site for the catalytic conversion of Nitrogen gas into ammonia (fertilizer) in biological processes \cite{Beinert1997}. This reaction is of great importance; while the mechanism is not understood due to complex electronic structure, biological Nitrogen fixation is significantly more efficient than the industrial alternative. The paper by {\micro} \cite{Reiher2017} focused on a 108 qubit active space,  and determined that roughly $10^{14}$ T gates would be required. If implemented in the surface code using gates with $10^{-3}$ error rates, the most efficient protocols for magic state distillation in this context require roughly 14 qubitseconds \cite{Litinski2018,Gidney2018pub} of spacetime volume. At those rates, just distilling the magic states needed for \cite{Reiher2017} would require over four million qubitdecades (e.g., four million qubits running for a decade or a billion qubits running for two weeks), which is not practical.

The works of \cite{BabbushSpectra,Kivlichan2019} show that one can perform similar sized chemistry simulations with roughly $10^8$ T gates, but in a plane wave rather than Gaussian basis. By application of techniques from \cite{Gidney2018pub} such calculations could be implemented in the surface code at $10^{-3}$ physical error rates with fewer than a million physical qubits in just hours. However, one would require far more plane waves to treat FeMoco, so these algorithms are not appropriate. In this paper, we develop an approach that has T counts somewhere in between those discussed in \cite{BabbushSpectra,Kivlichan2019} and \cite{Reiher2017} and is compatible with compact molecular orbitals appropriate for a system like FeMoco.

Our approach is to perform phase estimation directly  on a quantum walk \cite{Szegedy2004} generated using qubitization oracles \cite{Low2016}, designed to simulate Hamiltonians in the linear combinations of unitaries query model \cite{Childs2012}. Our analysis of the phase estimation algorithm is nearly identical to that in \cite{BabbushSpectra}, which realizes a proposal suggested in \cite{Berry2018,Poulin2017} based on qubitization \cite{Low2016}. We make heavy use of the unary iteration technique introduced in \cite{BabbushSpectra} (see also a similar idea in \cite{Childs2017}) as well as the QROM based state preparation and coherent alias sampling technique that was originally developed in \cite{BabbushSpectra} and then improved to lower T gate complexity in \cite{Lowpreparation}. Finally, a key aspect of our algorithm is to leverage the sparse nature of the Coulomb operator, using a low rank representation recently discussed in \cite{Motta2018}.

In the case where we limit the number of ancilla qubits used, mostly using the system qubits as ``dirty'' ancilla, our algorithm can obtain chemical accuracy for FeMoco with about $2\times 10^{13}$ Toffoli gates, using the active spaces of either Reiher, Wiebe, Svore, Wecker, and Troyer (\micro) \cite{Reiher2017} or Li, Li, Dattani, Umrigar, and Chan (\caltech) \cite{Li2019}.
If we allow a large number of ancilla then the number of Toffoli gates achieved with our most efficient approach is about $2\times 10^{11}$ for the {\micro} orbitals, or $8 \times 10^{10}$ for the {\caltech} orbitals.
Throughout we focus on complexities in terms of Toffoli counts, because the non-Clifford gates we use are exclusively Toffolis.
The cost in terms of T gates will be four times as large but since we are bottlenecked by Toffolis we can directly distill Toffoli states, which is possible with roughly the same cost as distilling two magic states for T gates \cite{Gidney2018pub}. Although we improve upon the distillation spacetime volume required by \cite{Reiher2017}, at $10^{-3}$ error rates we still require about three million qubitweeks of state distillation, which improves over previous results by roughly a factor of seven hundred, but is still substantial.

The paper is organized as follows.
In \sec{low_rank} we review how it is possible to truncate the Coulomb operator to low rank, and establish notation.
In \sec{LCU} we describe the Hamiltonian as a linear combination of unitaries, and how to perform the state preparation and controlled unitary operations.
We give calculations of the complexities obtained with the low rank truncation for FeMoco in \sec{complex}.
In \sec{sparsity} and \sec{sparse_complex} we discuss techniques that can be used to further lower the cost of qubitization based quantum chemistry simulations by leveraging unstructured sparsity that may exist in the Coulomb operator. We conclude in \sec{conc}.
In \app{lookup},  \app{clean-lookup}, and \app{unlookup} we discuss technical details relating to how the qubitization state preparation oracle is implemented. In \app{lambda_general} we discuss the scaling of the $\lambda$ parameter for more general chemical systems.
In \app{mincost} we give the details for minor contributions to the cost, and in \app{arith} we give circuits and exact costings for arithmetic.

\section{Low Rank Tensor Factorization of the Coulomb Operator}
\label{sec:low_rank}

In this section we review representations of the Coulomb operator based on low rank tensor decompositions. These ideas have existed in some form in the classical electronic structure literature for decades \cite{Martinez_2013,DF1,DF2,mCD1,mCD2,mCD3}, and were recently discussed in the context of Trotter based electronic structure simulations in \cite{Motta2018}.

We first define the electronic structure Hamiltonian in an arbitrary second quantized basis as
\begin{align}
	H
    & = \sum_{\sigma \in \{\uparrow, \downarrow\}} \sum_{p,q=1}^{N/2} h_{pq} a^\dagger_{p,\sigma} a_{q,\sigma} + \frac{1}{2}\sum_{\alpha,\beta\in \{\uparrow, \downarrow\}} \sum_{p,q,r,s=1}^{N/2} h_{pqrs} a^\dagger_{p, \alpha} a_{q,\beta}^\dagger a_{r,\beta} a_{s,\alpha} \\
    & = \sum_{\sigma \in \{\uparrow, \downarrow\}}\sum_{p,q=1}^{N/2} T_{pq} a^\dagger_{p,\sigma} a_{q,\sigma} + \sum_{\alpha,\beta\in \{\uparrow, \downarrow\}} \sum_{p,q,r,s=1}^{N/2} V_{pqrs}a^\dagger_{p,\alpha} a_{q,\alpha} a^\dagger_{r,\beta} a_{s,\beta}
    \label{eq:full_hamiltonian}
\end{align}
where $a^\dagger_p$ and $a_p$ are fermionic creation and annihilation operators for spin-orbital $\phi_p(r)$. 
The scalar coefficients $h_{pq}$ and $h_{pqrs}$ are the one- and two-electron integrals over the basis functions,
\begin{align}
h_{pq} &= \int dr_1 \, \phi_p^* (r_1) \left(-\frac{\nabla^2}{2} + U(r_1)\right) \phi_q (r_1), \\
h_{pqrs} &= \int dr_1 \, dr_2 \,  \phi_p^*(r_1) \phi_q^* (r_2) V(r_1, r_2) \phi_r (r_2)  \phi_s (r_1)
\label{eq:two_body_ints},
\end{align}
where $U(r_1)$ and $V(r_1,r_2)$ are the nuclear and electron-electron potentials, respectively.
In \eq{full_hamiltonian} we have implicitly defined $T_{pq}$ and $V_{pqrs}$ with respect to the usual integrals by rearranging the Coulomb operator in so-called ``chemist notation'' with $a^\dagger a\, a^\dagger a$ instead of $a^\dagger  a^\dagger a \, a$, and absorbing the factor of 1/2 into the coefficients. Reordering operators according to the fermionic anticommutation relations $\{a^\dagger_p, a_q\} = \delta_{pq}$ also slightly changes the one-body coefficients.

Assuming real basis functions (such as molecular orbitals), $T_{pq}$ and $V_{pqrs}$ are real and have the symmetries \cite{Helgaker2002},
\begin{equation}
\label{eq:symmetries}
T_{pq}=T_{qp}, \qquad  \qquad V_{pqrs}=V_{srqp}=V_{pqsr}=V_{qprs}=V_{qpsr}=V_{rsqp}=V_{rspq}=V_{srpq},
\end{equation}
which are important for properties of the tensor decompositions we will discuss. The rank-4 tensor $V$ has $N/2$ elements along each axis and we can reshape $V$ (e.g.\ using ``numpy.reshape'') into an $N^2/4 \times N^2/4$ matrix $W$ which has composite indices $pq$ (representing the first electron) and $rs$ (representing the second electron). This procedure is commonly referred to in the applied math literature as the matricization of a tensor.

The $W$ matrix is symmetric and positive semidefinite. It is important to emphasize here that we have focused on the spatial orbital representation of the two-electron integrals in the chemist ordering.  In the physicist ordering, the resulting matrix has full rank, and no reduction of cost is possible.  Similarly, introduction of fermionic symmetries induced in the full spin-orbital Hamiltonian into the coefficients (e.g.\ removing coefficients of terms like $a_i^\dagger a_j^\dagger a_k a_k$) will destroy the required structure for efficient simulation.  This is because we are exploiting the low rank nature of the underlying spatial Coulomb interaction through matrix factorization, and it is easy to lose this structure if one is not careful.  We will diagonalize it as
\begin{equation}
\label{eq:diag}
W g^{(\ell)} = \omega_\ell \, g^{(\ell)},
\qquad \qquad
W = \sum_{\ell=1}^{L} \omega_\ell \, g^{(\ell)} \left(g^{(\ell)}\right)^{\rm T},
\end{equation}
where $g^{(\ell)}$ is the $\ell^{\rm th}$ eigenvector of $W$ having size $N^2/4$ and $\omega_\ell \geq 0$ is its associated eigenvalue. 
Since we are taking $V$ and hence $W$ to be real, $g^{(\ell)}$ will also be real.
The rank of $W$ is denoted $L$.

If $W$ were of full rank then it would be the case that $L = N^2/4$; however, the integrals that one encounters in molecular electronic structure applications contain considerable structure. As a consequence of this structure, it turns out that $W$ is not full rank, and instead $L \in {\cal O}(N)$. The physical basis for this result is the pairwise nature of the Hamiltonian interactions, arising from the Coulomb kernel in a real-space representation. This property is regularly exploited in classical approaches to electronic structure in techniques such as ``density fitting'' \cite{DF1,DF2} which is commonly performed using a Cholesky decomposition \cite{mCD1,mCD2,mCD3} (which is similar to the diagonalization in \eq{diag} but is numerically more efficient and permits different left and right eigenvectors).

We use the notation $g^{(\ell)}_{pq}$ to denote the entry of $g^{(\ell)}$ indexed by the composite index $pq$ (the same composite index we used to flatten $V$ into $W$).
The eigenvectors $g^{(\ell)}$ inherit certain symmetry properties of $V$, with the result that $g^{(\ell)}_{pq}$ is symmetric between $p$ and $q$.
We can express the two-electron operator in terms of $g^{(\ell)}_{pq}$ as
\begin{equation}
\label{eq:diag1}
\sum_{\alpha,\beta \in \{\uparrow,\downarrow\}} \sum_{p,q,r,s=1}^{N/2} V_{pqrs}a^\dagger_{p,\alpha} a_{q,\alpha} a^\dagger_{r,\beta} a_{s,\beta} =  \sum_{\ell=1}^{L} \omega_\ell \left(\sum_{\sigma \in \{\uparrow,\downarrow\}}\sum_{p,q=1}^{N/2} g^{(\ell)}_{pq} a^\dagger_{p,\sigma} a_{q,\sigma}\right)^2.
\end{equation}
While common in electronic structure, this representation was first proposed in a quantum computing context in \cite{Poulin2014}; however, that work did not appear to appreciate the low rank aspect, which was first exploited for advantage in quantum computing in \cite{Motta2018}. Whereas there are ${\cal O}(N^4)$ seemingly distinct coefficients on the left side of this equation, there are ${\cal O}(N^2 L) = {\cal O}(N^3)$ distinct coefficients on the right side of the equation. Due to symmetry, the number of independent coefficients for each $\ell$ is $N^2/8+N/4$, giving a total number of $L(N^2/8+N/4)$.

The work of \cite{Motta2018} also discussed further factorizations of the Hamiltonian based on the results of \cite{Peng_2017}. There, they showed that one can also diagonalize and truncate the matrix with elements $g_{pq}^{(\ell)}$ which turns out to have only ${\cal O}(\log N)$ significant eigenvalues in certain asymptotic limits. Furthermore, one can rotate into the basis where the operator is diagonal using ${\cal O}(N \log N)$ operations. One might think this would be useful for linear combinations of unitaries based quantum simulation, and with sufficient cleverness, that might be the case. However, we do not focus on that second factorization in this work because it adds many intricacies, is less well understood than the first factorization, and might not offer an asymptotic advantage in T complexity for technical reasons related to the improved scaling that comes from using the improved ``QROAM'' discussed later in this work.

\section{LCU based simulation}
\label{sec:LCU}
A number of techniques for simulating Hamiltonian evolution are based upon the linear combination of unitaries approach \cite{Childs2012,Berry2013}.
This approach enables one to achieve a sum of unitaries that yields another unitary operation.
Say the operation to perform is given in the form $U=\sum_j w_j U_j$, where $w_j$ are real and positive.
First a control register is prepared in the state $\sum_j \sqrt{w_j/\lambda} \ket{j}$, where $\lambda=\sum_j w_j$ is needed for normalization.
We call this preparation operation ``\textsc{prepare}''.
Then a controlled $U_j$ operation is performed on the target system, an operation we will call ``\textsc{select}''.
The inverse \textsc{prepare} operation is performed, then if the control system is measured in the state $\ket{0}$, then the operation $U$ will have been applied to the target system.
This operation only has probability $1/\lambda^2$ of success, so to achieve $U$ with unit success probability one can use $\sim\lambda$ steps of oblivious amplitude amplification.

This formalism was generalized by the block encoding, or ``qubitization'', formalism of \cite{Low2016}, where one can take the Hamiltonian to be a linear combination of unitaries, and use quantum signal processing \cite{Low2017} to obtain Hamiltonian evolution.
For quantum chemistry, we are typically interested in the eigenvalues of the Hamiltonian.
In that case, instead of performing the Hamiltonian evolution, one can instead consider performing phase estimation on a single step from the qubitization formalism of \cite{Low2016}.
This step corresponds to expressing the Hamiltonian as a linear combination of unitaries using two \textsc{prepare} operations and one \textsc{select} operation, as well as a reflection as one would do for oblivious amplitude amplification.
The eigenvalues of this step will then be $e^{\pm i \arccos (E_k/\lambda)}$, where $E_k$ are the eigenvalues of the Hamiltonian.
The complexity is fundamentally dependent on the quantity $\lambda=\sum_j w_j$.
The overall complexity will be proportional to $\lambda$ multiplied by the complexity of the \textsc{select} and \textsc{prepare} operations.

\subsection{The Hamiltonian as a linear combination of unitaries}

We can map $a^\dagger_p a_q + a^\dagger_q a_p$ and the number operator $n_p=a^\dagger_p a_p$ to qubits using the Jordan-Wigner transformation as
\begin{equation}\label{eq:jw}
a^\dagger_{p,\sigma} a_{q,\sigma} + a^\dagger_{q,\sigma} a_{p, \sigma} \mapsto \left(X_{p,\sigma} \vec Z X_{q,\sigma} + Y_{p,\sigma} \vec Z Y_{q,\sigma}\right), \qquad
a^\dagger_{p,\sigma} a_{p,\sigma} \mapsto \frac{1}{2} \left(\openone - Z_{p,\sigma}\right),
\end{equation}
where $X$, $Y$ and $Z$ are the Pauli operators, the subscripts indicate the qubits these operators act on, and $A_p \vec Z A_q$ is shorthand for $A_p Z_{p+1} \cdots Z_{q-1} A_q$. Thus, our Hamiltonian can be represented as the linear combination of unitaries:
\begin{align}\label{eq:firstham}
H &\mapsto \frac 12 \sum_{p\ne q,\sigma} T_{pq}\left(X_{p,\sigma} \vec Z X_{q,\sigma} + Y_{p,\sigma} \vec Z Y_{q,\sigma}\right)
+ \frac 12 \sum_{p,\sigma} T_{pp} \left(\openone - Z_{p,\sigma}\right) \nonumber \\ & \quad
+ \frac 14 \sum_{\ell} \omega_\ell \left(\sum_{p\ne q,\sigma} g^{(\ell)}_{pq} \left(X_{p,\sigma} \vec Z X_{q,\sigma} + Y_{p,\sigma} \vec Z Y_{q,\sigma}\right) + \sum_{p,\sigma} g^{(\ell)}_{pp}\left(\openone - Z_{p,\sigma}\right) \right)^2.
\end{align}
The terms in the first line in this expression correspond to the one-body operator $T$, and the terms in the second line correspond to the factorized two-body operator.
Here the ranges of the summations are the same as for $H$.

Given the Hamiltonian expressed as a linear combination of unitaries, we can now give the expression for $\lambda$.
In the following we will use $\lambda_T$ to refer to the sum of weights for the one-body term as in \eq{full_hamiltonian}, and use $\lambda_W$ to refer to the sum of weights for Coulomb operator in its factorized form, as on the right-hand side of \eq{diag1}.
We have $\lambda = \lambda_T + \lambda_W$, with
\begin{align}
    \label{eq:lambdas}
    \lambda_T = 2 \sum_{p,q=1}^{N/2} \left | T_{pq} \right |,
    \qquad \qquad
    \lambda_W = 4 \sum_{\ell=1}^L \omega_\ell \left(\sum_{p,q=1}^{N/2} \left |g_{pq}^{(\ell)} \right |\right)^2.
\end{align}
Here the factors of 2 and 4 in front of the sums are due to summation over the up and down spins. By simulating the factorized Hamiltonian we are slightly increasing $\lambda$ over what it would be if we used the Hamiltonian in its original form from \eq{full_hamiltonian}. Denoting by $\lambda_V$ the sum of weights for the potential term in the form \eq{full_hamiltonian}, one has
\begin{equation}
\lambda_V = 4 \sum_{p,q,r,s=1}^{N/2} \left | V_{pqrs} \right|
    = 4 \sum_{p,q,r,s=1}^{N/2} \left| \sum_{\ell=1}^L \omega_\ell \, g^{(\ell)}_{pq} g^{(\ell)}_{rs} \right| \leq  4 \sum_{\ell=1}^L \omega_\ell \!\! \sum_{p,q,r,s=1}^{N/2}  \left|g^{(\ell)}_{pq}\right| \left|g^{(\ell)}_{rs} \right| = \lambda_W.
    \label{eq:lambda_V}
\end{equation}
However, we do not expect any difference between the asymptotic scaling of $\lambda_V$ and $\lambda_W$ with respect to $N$.

Because these quantities directly scale the complexity of our approach, techniques for reducing the effective value of $\lambda$ are potentially important. Perhaps the simplest idea to reduce $\lambda$ might be to optimize it under rotations of the single-particle basis (and to accordingly rotate the initial state using the Givens rotation technique in \cite{Kivlichan2017}). Another example of such a technique is to modify the Hamiltonian by adding to it a linear combination of $n$-representability \cite{Mazziotti2012} equality constraints that have provably zero expectation value. This strategy was introduced and shown to be effective in Section V of \cite{Rubin2018}. Other techniques that might be useful for this include mean-field background subtraction \cite{Al-Saidi2006} and the use of soft pseudopotentials \cite{Vanderbilt1990}. However, one should make sure that these methods are applied in a way that does not increase the rank of the Coulomb operator, which would be counterproductive for the overall complexity. Note finally that interaction picture techniques \cite{Low2018} do not appear to be helpful here because while $\lambda_V \gg \lambda_T$, the $V$ operator cannot be fast-forwarded.

In order to perform phase estimation via the qubitization/LCU approach, we need to be able to perform the state preparation \textsc{prepare} and controlled unitaries \textsc{select}.
The techniques to achieve these operations are described in the following subsections.

\subsection{State preparation}
The state we would need to prepare is
\begin{align}
\ket{\psi}&=\ket{0}\ket{+}\ket{0}\sum_{p,q,\sigma} \sqrt{\frac{|T_{pq}|}{\lambda}} \ket{\theta_{pq}^{(0)}}\ket{0}\ket{p,q,\sigma}\ket{0}\nn
    &\quad +\sum_{\ell} \sqrt{\frac{\omega_\ell}{\lambda}} \ket{\ell}\ket{+}\ket{+} \!\!\!\! \sum_{p,q,r,s,\alpha,\beta} \!\!\!\! \sqrt{|g_{pq}^{(\ell)}g_{rs}^{(\ell)}|} \ket{\theta_{pq}^{(\ell)}}\ket{\theta_{rs}^{(\ell)}}\ket{p,q,\alpha}\ket{r,s,\beta},
\end{align}
where $\lambda$ is as defined in \eq{lambdas}, $\ket{+}=(\ket{0}+\ket{1})/\sqrt{2}$, $\alpha$ and $\beta$ are spin labels, and $\theta_{pq}^{(\ell)}$ are used to obtain the correct signs on the terms.
They are given by
\begin{equation}
\theta_{pq}^{(0)} = \begin{cases}
    0, & T_{pq}>0, \\
    1, & T_{pq}<0,
\end{cases} \qquad 
\theta_{pq}^{(\ell)} = \begin{cases}
    0, & g_{pq}^{(\ell)}>0, \\
    1, & g_{pq}^{(\ell)}<0.
\end{cases}
\end{equation}
The first register gives $\ell$, with $\ket{0}$ reserved for the first term.
The first term gives the first two terms in \eq{firstham}, with the first term in \eq{firstham} obtained for $p\ne q$ and the second term obtained for $p=q$. The $\ket{+}$ state on the second qubit selects between $\openone$ and $Z_{p,\sigma}$ for $p=q$.
We use $p<q$ and $p>q$ to select between $X_{p,\sigma} \vec Z X_{q,\sigma}$ and $Y_{p,\sigma} \vec Z Y_{q,\sigma}$ for $p\ne q$, in a similar way as in \cite{BabbushSpectra}.
The second term gives the third term in \eq{firstham}, with the sums over $p,q$ and $r,s$ yielding the square.
The two $\ket{+}$ states in registers 2 and 3 select between $\openone$ and $Z_{p,\alpha}$ for $p=q$ and between $\openone$ and $Z_{r,\beta}$ for $r=s$.

There are $(L+1)(N^2/8+N/4)={\mathcal O}(N^3)$ unique coefficients, which indicates the state preparation can be performed with similar complexity.
Using the QROM and subsampling techniques from \cite{BabbushSpectra}, the T complexity can be expected to be ${\cal O}(N^3+\log(1/\epsilon))$, where $\epsilon$ is the required precision.
By using a more sophisticated preparation scheme it will be possible to reduce the number of T gates, as will be described below.

To perform the state preparation, we start with all registers in the $\ket{0}$ state, then perform the following steps.
\begin{enumerate}
    \item Prepare a superposition over the first register, to produce the state
\begin{equation}
\left(\ket{0}\sqrt{\sum_{p,q} \frac{2|T_{pq}|}{\lambda}}
    +2\sum_{\ell} \sqrt{\frac{\omega_\ell}{\lambda}} \ket{\ell} \sum_{p,q} |g_{pq}^{(\ell)}| \right)\ket{0}\ket{0}\ket{0}\ket{0}\ket{0}\ket{0} .
\end{equation}
\item Perform a Hadamard on the second register, and a Hadamard on the third register controlled on the state of the first register being $\ket{\ell}$ for $\ell>0$, giving
\begin{equation}
\left(\ket{0}\ket{+}\ket{0}\sqrt{\sum_{p,q} \frac{2|T_{pq}|}{\lambda}}
    +2 \sum_{\ell} \sqrt{\frac{\omega_\ell}{\lambda}} \ket{\ell}\ket{+}\ket{+} \sum_{p,q} |g_{pq}^{(\ell)}| \right)\ket{0}\ket{0}\ket{0}\ket{0}.
\end{equation}
\item Prepare a superposition over register six controlled on the first register.
For the first register in the state $\ket{0}$, we prepare weights $\sqrt{|T_{pq}|}$, and for $\ket{\ell}$ with $\ell>0$ we prepare weights proportional to $\sqrt{|g_{pq}^{(\ell)}|}$.
The state is then
\begin{align}
&\ket{0}\ket{+}\ket{0}\sum_{p,q,\sigma} \sqrt{\frac{|T_{pq}|}{\lambda}} \ket{0}\ket{0}\ket{p,q,\sigma}\ket{0}\nn
    &\quad +\sqrt{2} \sum_{\ell} \sqrt{\frac{\omega_\ell}{\lambda}} \ket{\ell}\ket{+}\ket{+} \sum_{p,q,\alpha} \sqrt{|g_{pq}^{(\ell)}|}\sqrt{\sum_{r,s} |g_{rs}^{(\ell)}|} \ket{0}\ket{0}\ket{p,q,\alpha}\ket{0}.
\end{align}
\item Perform another state preparation on register seven, controlled on register one.
For register one in the state $\ket{\ell}$ with $\ell>0$ we prepare weights proportional to $\sqrt{|g_{rs}^{(\ell)}|}$, giving
\begin{align}
&\ket{0}\ket{+}\ket{0}\sum_{p,q,\sigma} \sqrt{\frac{|T_{pq}|}{\lambda}} \ket{0}\ket{0}\ket{p,q,\sigma}\ket{0}\nn
    &\quad +\sum_{\ell} \sqrt{\frac{\omega_\ell}{\lambda}} \ket{\ell}\ket{+}\ket{+} \sum_{p,q,r,s,\alpha,\beta} \sqrt{|g_{pq}^{(\ell)}g_{rs}^{(\ell)}|} \ket{0}\ket{0}\ket{p,q,\alpha}\ket{r,s,\beta}.
\end{align}
\item
Use QROM to output $\ket{\theta_{pq}^{(\ell)}}$ in register four and
$\ket{\theta_{rs}^{(\ell)}}$ in register five.
\end{enumerate}

We will allow total error $\epsilon$.
Because there are a number of steps, each step will have an allowable error some fraction of $\epsilon$.
Here we aim to estimate the leading-order term in the complexity, and the allowable error will only appear in logarithms, so we will simply give $\log(1/\epsilon)$, rather than subdividing the allowable error between the different steps.
Throughout we will use $\log$ to indicate logarithms to base 2.

For the state preparation in step 1, the approach in \cite{BabbushSpectra} gives complexity in terms of Toffolis $L+{\mathcal O}(\log(1/\epsilon))$.
The complexity in step 1 will be negligible compared to the complexity of later steps.
The second step is just controlled operations on two qubits, and has negligible complexity compared to the other steps.

Steps 3 and 4 have the greatest complexity.
A simple method is to use the unary iteration procedure as described in \cite{BabbushSpectra} (Section IIIA) combined with the state preparation procedure in \cite{BabbushSpectra} (Section IIID).
The unary iteration procedure allows us to progressively perform an operation controlled on a register being $\ket{0}$, then $\ket{1}$, and so forth, with the overall complexity of the control in terms of the number of Toffoli gates being $L$ (since there are here $L+1$ possible values).
That unary iteration procedure is performed on the first register, and for each value of this register, the state preparation is performed on the sixth register (for step 3) or the seventh register (for step 4).

There is a subtlety in that the values of $p$ and $q$ range over $N/2$ values, which need not be a power of $2$.
The result from \cite{BabbushSpectra} is for contiguous binary values.
In the case where there are two subregisters, then we can iterate through the register for $p$, and use its output qubit as the control for the unary iteration over the register for $q$.
The complexity for iterating over $p$ is $N/2-1$, and for each of the $N/2$ values of $p$ the complexity of iterating over $q$ is $N/2-1$.
That gives a total complexity for iterating over $p$ and $q$ that is $N^2/4-1$.

As a result the complexity of the state preparation is $N^2/4+{\cal O}(\log(1/\epsilon))$ for each $\ell$ in both steps 3 and 4.
In each case the complexity is $N/2+{\cal O}(\log(1/\epsilon))$ for zero on the first register for step 3 only.
The total complexity is then $N/2+LN^2/2 +{\cal O}( L\log(1/\epsilon))$.
We can significantly reduce the complexity using three techniques.
\begin{enumerate}[A)]
    \item Take advantage of the symmetry of $g_{pq}^{(\ell)}$ and $T_{pq}$.
    \item Combine the preparation for all values of $\ell$ together.
    \item Use the QROM of \cite{Lowpreparation} which allows one to further reduce Toffoli complexity at the cost of extra ancilla.
\end{enumerate}
To take advantage of the symmetry, we can initially prepare a state proportional to
\begin{equation}
\sqrt{2}\sum_{p>q} \sqrt{|g_{pq}^{(\ell)}|}\ket{p,q,\alpha}+\sum_{p} \sqrt{|g_{pp}^{(\ell)}|}\ket{p,p,\alpha}.
\end{equation}
Then we have this state tensored with a register in a $\ket{+}$ state.
For preparation on register six, the $\ket{+}$ state is on register two (step 3), or for preparation on register seven the $\ket{+}$ state is on register three (step 4).

We can then perform a swap between the registers storing $p$ and $q$ controlled by this register, giving state
\begin{equation}
\sum_{p>q} \sqrt{|g_{pq}^{(\ell)}|}\ket{0}\ket{p,q,\alpha}
+\sum_{p<q} \sqrt{|g_{pq}^{(\ell)}|}\ket{1}\ket{p,q,\alpha}
+\sum_{p} \sqrt{|g_{pp}^{(\ell)}|}\ket{+}\ket{p,p,\alpha}.
\end{equation}
This gives the correct weighting for each of the terms in the superposition.
As always with these state preparations for LCU, the prepared state is permitted to be entangled with junk registers.
For $p\ne q$ the additional ancilla may be regarded as a junk register, whereas for $p=q$ this register will be used to distinguish between $\openone$ and $Z_{p,\alpha}$ operations.
The controlled swap costs ${\cal O}(\log N)$ Toffolis, which is negligible compared to other steps.
As a result of this simplification, there are $(L+1)(N^2/8+N/4)$ distinct values required in step 3, and $L(N^2/8+N/4)$ distinct values in step 4.

To explain technique B for reducing the complexity, the state preparation is performed in the following way.
\begin{enumerate}[(i)]
\item Create an equal superposition over $j$ for the register where we are performing the state preparation.
\item Output alternate indices ($\ket{{\rm alt}_j}$ in \cite{BabbushSpectra}) and probabilities ($\ket{{\rm keep}_j}$ in \cite{BabbushSpectra}) using a QROM.
\item Perform an inequality test between the probability register and an ancilla in an equal superposition state.
\item Perform a controlled swap between the register where we are performing the state preparation and the alternate index register based on the result of the inequality test.
\end{enumerate}
We also need to create superpositions over the spin registers, but that can be done trivially with Hadamards.
If we were to iterate through $\ell$ and perform state preparation for each value of $\ell$, we would be performing the entire procedure for each value of $\ell$.
The insight here is to note that we can call the QROM for all $\ell$, then perform the inequality test and controlled swap.
That means we only need to perform the inequality test and controlled swap once, instead of $L$ times.

Technique C for reducing the complexity is the most significant.
The dominant cost in the procedure is that of the QROM, which has cost of $(2L+1)(N^2/8+N/4)$ Toffolis if we use the procedure of \cite{BabbushSpectra}.
That is, we need to output $(L+1)(N^2/8+N/4)$ or $L(N^2/8+N/4)$ numbers in QROM, with outputs of size $\log(N^2)+\log(1/\epsilon)+{\mathcal O}(1)$.
Here $\log(N^2)$ is the size of the register for the alternate values,
and the size of the register for the probability is $\log(1/\epsilon)+{\mathcal O}(1)$.

The complexity in terms of Toffolis can be reduced by using a more advanced QROM based on that of \cite{Lowpreparation}.
This QROM uses a combination of the QROM of \cite{BabbushSpectra} and a technique for trading between spatial complexity and gate complexity in a fashion that accomplishes something  reminiscent of what authors aspired to demonstrate with the original concept of ``QRAM'' \cite{Giovannetti2008}. Thus, here we will refer to the more advanced QROM of \cite{Lowpreparation} as ``QROAM''.
In the following we will use $d$ for the number of entries that we must look up using the QROM (here we have $\dm\approx L(N^2/8+N/4)$), $\chunk$ for an arbitrary power of two, and $\wid$ for the size of the output in qubits (here we have $\wid=\log(N^2/\epsilon)+{\mathcal O}(1)$).
Then the complexity for computing the QROM is $\lceil\dm/\chunk\rceil+\wid(\chunk-1)$, and for uncomputing the QROM is $\lceil\dm/\chunk\rceil+\chunk$ (see \app{unlookup}).
Moreover, it is possible to choose the $k$ for the uncompute to be different from that for the compute step.
The number of additional ancillae needed is $(\chunk-1)\wid$.
It is also possible to use ancillae that are already being used for some other purpose, called ``dirty'' ancillae.
Using these dirty ancillae, the compute cost is over twice as much, $2\lceil\dm/\chunk\rceil + 4\wid(\chunk-1)$ (see \app{lookup}), and the uncompute cost is changed to $2\lceil d/\chunk\rceil+4\chunk$.
The compute and uncompute are used for the state preparation and inverse state preparation, so the combined cost is what needs to be considered.

The results we use here are improved slightly over those in \cite{Lowpreparation}.
The Toffoli count achieved in \cite{Lowpreparation} is $2\dm/\chunk + 8 \wid \chunk$ (from the last column and last row of Table II, after dividing by 4 to account for the fact they are counting T gates and also after substituting $\dm$, $\chunk$ and $\wid$ for $N$, $\lambda$ and $b$ respectively).
Our corresponding Toffoli count is $2\dm/\chunk + 4 \wid(\chunk - 1)$.
The factor of two improvement in the $\wid \chunk$ term is because we use a linear depth swapping network instead of a logarithmic depth swapping network.
A logarithmic depth network requires spreading control qubits for parallel CSWAPs over many ancillae, but because the ancillae are dirty each CSWAP must be toggled-controlled which involves repeating the operation twice.
The small improvement from $\wid\chunk$ to $\wid (\chunk - 1)$ in the ancilla count is due to using $\ket{+}$ states instead of a spare register in order to ensure the output is only toggled once.
There is also an improvement from $\wid\chunk$ to $\wid (\chunk - 1)$ in the Toffoli count, but that is due to a more careful accounting of the number of controlled swaps needed.

The most significant improvement we make is the application of measurement based uncomputation, as described in \app{unlookup}, which removes the dependence on $\wid$ in the complexity when uncomputing a lookup.
The principle is similar to that used to reduce the Toffoli complexity of addition in \cite{GidneyAdder}.  Instead of just reversing the circuit for a table lookup, you can perform $X$ measurements on the output qubits.  Based on the measurement result you can perform a classically-conditioned phase fixup.
This procedure also means that the ancillae used by the forward QROAM only need to be used temporarily, and can be erased after the QROAM and reused.

There is a subtlety in using these results in that the QROAM is for a \emph{single} control register which can take a contiguous set of values.
In contrast, here we have three registers with $\ell$, $p$, and $q$.
In this case we can simply convert to a single contiguous register for the iteration.
We can compute a value for a single contiguous register $s$ from $\ell$, $p$, and $q$ as
\begin{equation}\label{eq:qromreg}
s = \ell (N^2/8+N/4) + p(p+1)/2 + q .
\end{equation}
The $p(p+1)/2$ term takes account of the fact that we are preparing $p$ and $q$ only for $p\ge q$.
We can use QROAM directly on $s$ with just an additional logarithmic overhead for the arithmetic.

We will consider two cases.
First is that where we attempt to minimize the cost in terms of Toffolis, but use a large number of ancillae.
In that case, for the compute we can take $\chunk\approx\sqrt{d/\wid}$, in which case the cost of the compute step is approximately $2\sqrt{\dm\wid}$.
For the uncompute, we can take $\chunk\approx \sqrt{d}$, which gives an uncompute cost of approximately $2\sqrt{d}$, for a total cost of the compute and uncompute of $2\sqrt{d}(\sqrt{M}+1)$.
For our $\dm\approx LN^2/8$ and $\wid\approx \log(N^2/\epsilon)$, we get a combined cost of approximately
\begin{equation}
\sqrt{LN^2\log(N^2/\epsilon)/2}\, .
\end{equation}
We find that this approach needs a number of extra ancillae
\begin{equation}
\sqrt{LN^2\log(N^2/\epsilon)/8}\, .
\end{equation}
As $L={\mathcal O}(N)$ the complexity in terms of Toffolis is ${\mathcal O}(N^{3/2}\sqrt{\log(N/\epsilon)})$, and a similar number of ancillae are needed.

Alternatively, if we are attempting to minimize the number of additional ancillae needed, we can use ``dirty'' ancillae instead (ones that are not initialized to zero).
Fortunately, we happen to have $N$ dirty ancilla lying around because the system register is not acted upon while implementing the state preparation operation.
Moreover, there are multiple steps of state preparation that are performed, and there are qubits that will be used in some steps of state preparation that can be used as dirty qubits for the other steps of state preparation.
We will find that we can take the number of dirty qubits to be somewhat larger than $N$, but not a lot larger.
Assuming the number of dirty qubits is about $N$, we can take $\chunk\approx N/\wid$.
Then we would have compute cost $2\dm\wid /N - 4\wid + 4N\approx \frac 14 LN\log(N^2/\epsilon)$.
For the uncompute step we can take $\chunk\approx N$, giving cost approximately $LN/4$.
In both cases the costs need to be multiplied by 2 to account for steps 3 and 4.

Finally we consider the cost of outputting $\ket{\theta_{pq}^{(\ell)}}$ in register four and
$\ket{\theta_{rs}^{(\ell)}}$ in register five.
This use of QROM can simply be combined with that in steps 3 and 4.
For example, for step 3, when calling the QROM for the state preparation, output the value of $\theta_{pq}^{(\ell)}$, as well as that for the alternate values of $p$ and $q$.
Then, when doing the controlled swap, also swap these registers.
There is a net increase in the size of the output of 2 qubits, and one extra Toffoli for the controlled swaps.
This cost is negligible compared to the overall cost in steps 3 and 4.

\subsection{Controlled unitaries}
For the controlled unitaries (the \textsc{select} circuit) in the case of using only the first diagonalization, we need to implement the terms in the Hamiltonian as in \eq{firstham}.
The general principle is that we do a pair of operations, each of which has $X_p \vec Z X_q$ and $Y_p \vec Z Y_q$ for $p\ne q$, with the term selected by whether $p$ or $q$ is larger.
For $p=q$ we use an ancilla qubit to select between $\openone$ and $Z_p$.
The way the state preparation is chosen, this can be performed in the same way for $\ell=0$ and $\ell>0$, because the ancillae will only select the identity operation.
The operations we need are
\begin{align}
\textsc{select}_1&\ket{q_1,q_2,\theta_1,\theta_2,\{p,q,\alpha\},\{r,s,\beta\}}\ket{\psi}\nonumber \\
&=(-1)^{\theta_1} \ket{q_1,q_2,\theta_1,\theta_2,\{p,q,\alpha\},\{r,s,\beta\}}\otimes \begin{cases}
X_{p,\alpha}\vec Z X_{q,\alpha}\ket{\psi}, & p<q, \\
Y_{p,\alpha}\vec Z Y_{q,\alpha}\ket{\psi}, & p>q, \\
\ket{\psi}, & p=q \wedge q_1=1, \\
-Z_{p,\alpha}\ket{\psi}, & p=q \wedge q_1=0,
\end{cases}
\end{align}
\begin{align}
\textsc{select}_2&\ket{q_1,q_2,\theta_1,\theta_2,\{p,q,\alpha\},\{r,s,\beta\}}\ket{\psi}\nonumber \\
&=(-1)^{\theta_2}\ket{q_1,q_2,\theta_1,\theta_2,\{p,q,\alpha\},\{r,s,\beta\}}\otimes \begin{cases}
X_{r,\beta}\vec Z X_{s,\beta}\ket{\psi}, & r<s, \\
Y_{r,\beta}\vec Z Y_{s,\beta}\ket{\psi}, & r>s, \\
\ket{\psi}, & r=s \wedge q_2=1, \\
-Z_{r,\beta}\ket{\psi}, & r=s \wedge q_2=0.
\end{cases}
\end{align}

Note that the selected operations we need are similar to those in \cite{BabbushSpectra,BabbushSYK}, and we can use a similar approach.
The complexity is linear in $N$, and will therefore be smaller than the complexity of the state preparation.
The technique is shown in \fig{selecth} for $\textsc{select}_1$, and $\textsc{select}_2$ is equivalent. 
If $p < q$ then the $Y$ operation in $Z\ldots ZY$ acts on the same qubit as one of the $Z$s in the $Z\ldots ZX$ operation.
As a result, the $Y$ gets multiplied by $Z$ and becomes $ZY=-iX$.
Therefore the operation implemented is of the form $-iX_{p,\alpha}Z\ldots ZX_{q,\alpha}$.
If $p > q$ then the $X$ operation in $Z\ldots ZX$ acts on the same qubit as one of the $Z$s in the $Z\ldots ZY$ operation.
Thus we have $X$ times $Z$ on that qubit giving $XZ=-iY$.
Therefore the operation implemented is of the form $-iY_pZ\ldots ZY_q$.
If $p=q$ then all the $Z$s cancel leaving only $X_{p,\alpha}Y_{p,\alpha} = iZ_{p,\alpha}$.

Now note that the register $q_1$ is $0$ for $p>q$ and $1$ for $p<q$.
Before and after the ranged operations, we perform an inequality test between $p$ and $q$, controlled on $q_1$, with the result that an extra ancilla is in the state $\ket{1}$ unless $p=q$ and $q_1=0$.
That register is used as a control for the ranged operations, so if $p=q$ and $q_1=0$ then the identity is performed.
We then apply an $S$ gate on this ancilla, with the result that the operations performed are $X_{p,\alpha}Z\ldots ZX_{q,\alpha}$ for $p<q$, $Y_{p,\alpha}Z\ldots ZY_{q,\alpha}$ for $p>q$, and $-Z_{p,\alpha}$ for $p=q$ and $q_1=0$.
For $p=q$ and $q_1=0$ this ancilla is zero, so the phase on the identity is unchanged.
This yields the desired operations with the correct phases, and lastly 
the controlled $Z$ on the $\theta$ register gives the $(-1)^{\theta_1}$ phase factor.

\begin{figure}[tbh]
\centering
\hspace{0mm}
\Qcircuit @R=1em @C=0.75em {
&&&&&   & & & & & && &&
\lstick{|1\rangle} & \qw &\qw          & \qw & \gate{\oplus p=q} &\ctrl{1} &\gate{S}    & \ctrl{1} &\gate{\oplus p=q}       & \qw   &\qw &\\
&&&&&\lstick{\text{control}}&\qw    &\qw                    &\qw        &\ctrl{1}  &\qw && &&&
\qw &\qw                    &\qw    &\qw \qwx & \ctrl{3}   &\ctrl{-1}    &\ctrl{4}       & \qw \qwx       &\ctrl{1}      &\qw &\\
&&&&&\lstick{q_1}           &\qw    &\qw                    &\qw        &\gate{\text{In}_{q_1}}                   \qwx &\qw && &&&
\qw &\qw                    &\qw    &\ctrlo{2}\qwx & \qw        &\qw                &\qw           &\ctrlo{2} \qwx                       &\qw &\qw &\\
&&&&&\lstick{\theta_1}      &\qw    &\qw                    &\qw        &\gate{\text{In}_{\theta_1}}                          \qwx &\qw && &&&
\qw &\qw                    &\qw    &\qw & \qw    &\qw           &\qw    &\qw    &\gate{Z}\qwx &\qw              &\\
&&&&&\lstick{p}             &\qw    &\ustick{\log N-1}\qw   &{/} \qw &\gate{\text{In}_p}                          \qwx &\qw &&=&&&
\qw &\ustick{\log N-1}\qw   &{/}\qw &\gate{\text{In}_p}        & \gate{\text{In}_p}        &\qw       & \qw & \gate{\text{In}_p}  &\qw           &\qw &\\
&&&&&\lstick{q}             &\qw    &\ustick{\log N-1} \qw &{/} \qw &\gate{\text{In}_q}                          \qwx &\qw && &&&
\qw &\ustick{\log N-1}\qw   &{/}\qw &\gate{\text{In}_q} \qwx        & \qw \qwx      & \qw  &\gate{\text{In}_q}     & \gate{\text{In}_q} \qwx          &\qw &\qw &\\
&&&&&\lstick{\alpha}        &\qw    &\qw                   &\qw     &\gate{\text{In}_\alpha}                      \qwx &\qw && &&&
\qw &\qw                    &\qw    &\qw & \gate{\text{In}_\alpha}\qwx &\qw   & \gate{\text{In}_\alpha} \qwx                         &\qw           &\qw             &\qw &\\
&&&&&\lstick{|\psi\rangle}  &\qw    &\ustick{N}      \qw   &{/} \qw &\gate{\textsc{select}_1}        \qwx &\qw && &&&
\qw &\ustick{N}\qw          &{/}\qw &\qw & \gate{\overrightarrow{Z} Y_{p,\alpha}} \qwx &\qw   &\gate{\overrightarrow{Z} X_{q,\alpha}} \qwx &\qw           &\qw    &\qw &\\
}
\caption{\label{fig:selecth} The circuit needed to perform a controlled $\textsc{select}_1$ operation. We have omitted the registers this operation does not act upon for simplicity.
The unitaries labeled as $\protect\overrightarrow{Z} A_j$ apply the operation $Z_0 \cdots Z_{j-1} A_j$ to the target register, depending on the value from the input register, using the technique shown in Figure 9 of \cite{BabbushSpectra}.
This operation can be achieved using an inequality test, followed by a ranged operation via the technique shown in Figure 8 of \cite{BabbushSpectra}.
The controlled $\textsc{select}_2$ operation is completely equivalent except with different control registers.
}
\end{figure}
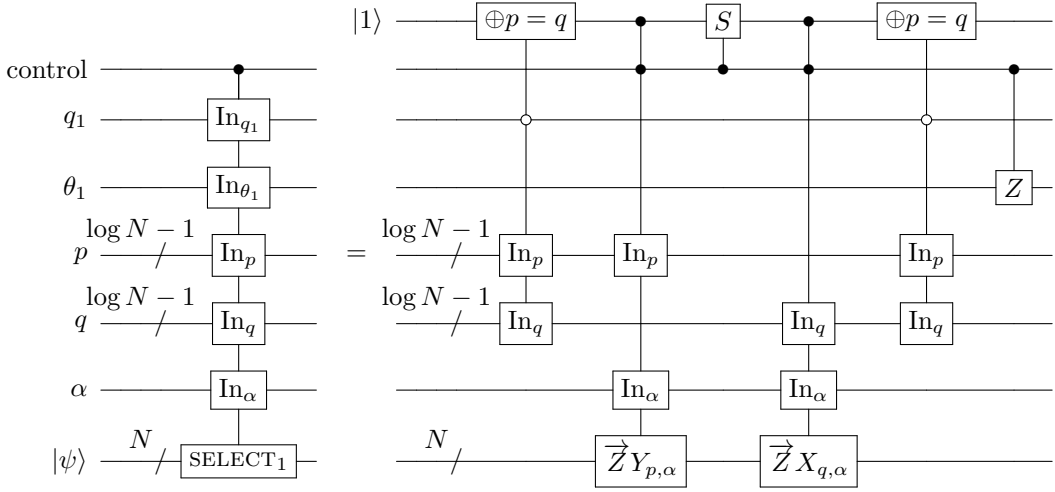

\section{Complexity}
\label{sec:complex}

Let us denote the upper bound on the error required for the eigenvalue estimation by $\Delta E$.  Then following \cite{BabbushSpectra} we find the complexity of the estimation is the cost of each LCU step times $2^m$, where
\begin{equation}
m=\left\lceil\log\left(\frac{\sqrt 2 \pi \lambda}{2\Delta E}\right)\right\rceil .
\end{equation}
Moreover, the error that is allowable for the implementation of each LCU step is
\begin{equation}
    \epsilon = \frac{\sqrt 2 \Delta E}{4\lambda} .
\end{equation}

Some minor costs are as follows.
\begin{enumerate}
\item The cost of the controlled operation in \fig{selecth} is $2N$ Toffoli gates for the two controlled ranged operations, and $2\lceil \log N\rceil$ for the two inequality tests.
These need to be done twice for a total of $4N+4\lceil \log N\rceil$.
\item In the state preparation we initially need to prepare superpositions over $\ell$, $p$, $q$, $r$, $s$, with $\ell\le L$, $N/2>p\ge q$, $N/2>r\ge s$.
This can be achieved by creating an equal superposition over ranges that are powers of two using Hadamards, then flagging success using inequality tests.
The number of Toffolis needed for these inequality tests will correspond to the number of qubits.
Amplitude amplification can give the desired result with amplitude close to 1.
The reflection in the amplitude amplification also needs a number of Toffolis corresponding to the number of qubits, so for $m$ steps of amplitude amplification the number of Toffolis is $3m+1$ times the number of qubits.
This needs to be multiplied by two because there is preparation and inverse preparation at each step.
\item The state preparation needs an inequality test, which has cost $\mu$ for $\mu$ bits of precision in the keep probability, and a cost corresponding to the number of qubits for the controlled swap.
A controlled swap on a pair of qubits can be performed with a single Toffoli and CNOTs.
\item For the state preparation we also do controlled swaps of the $p$ and $q$, as well as $r$ and $s$ registers.
These two controlled swaps cost $2\lceil \log (N/2)\rceil$ Toffolis.
\item Computing the function in \eq{qromreg} requires multiplication by a constant, a regular multiplication, and three additions.
The division by 2 can be achieved by trivially shifting the bits.
The multiplication can be achieved with $2\lceil \log (N/2)\rceil^2$ Toffolis.
\end{enumerate}

In the remainder of this section we quantify the costs of the QROM needed for the state preparation, which is the main contributor to the complexity, then give the total cost.

\subsection{{\micro} orbitals}

The prominence of the {\micro} paper \cite{Reiher2017}, which was the first work to rigorously estimate the T complexity of any quantum algorithm for chemistry makes it an important benchmark. Unfortunately, {\caltech} \cite{Li2019} later argued that there were substantial problems with the orbitals chosen for the {\micro} paper. For the reasons discussed in the paper by {\caltech}, we believe that future papers should compare against this work using only the {\caltech} integrals. But in order to more easily compare with past work, here we analyze the complexity of simulating both {\micro} and {\caltech} FeMoco active spaces. Note that at 152 spin-orbitals the {\caltech} active space is also significantly larger than the 108 spin-orbital {\micro} active space.

Our approach is to choose $L$ by observing the effect of truncation on two different efficient classical correlated approximate methods for molecular electronic structure: Configuration Interaction at the singles and doubles level (CISD) and M\o ller-Plesset perturbation theory to second order (MP2). For a review of both methods, see \cite{Helgaker2002}.

\begin{figure}[b]
\begin{minipage}[t]{.47\textwidth}
\centering
\includegraphics[width=\linewidth]{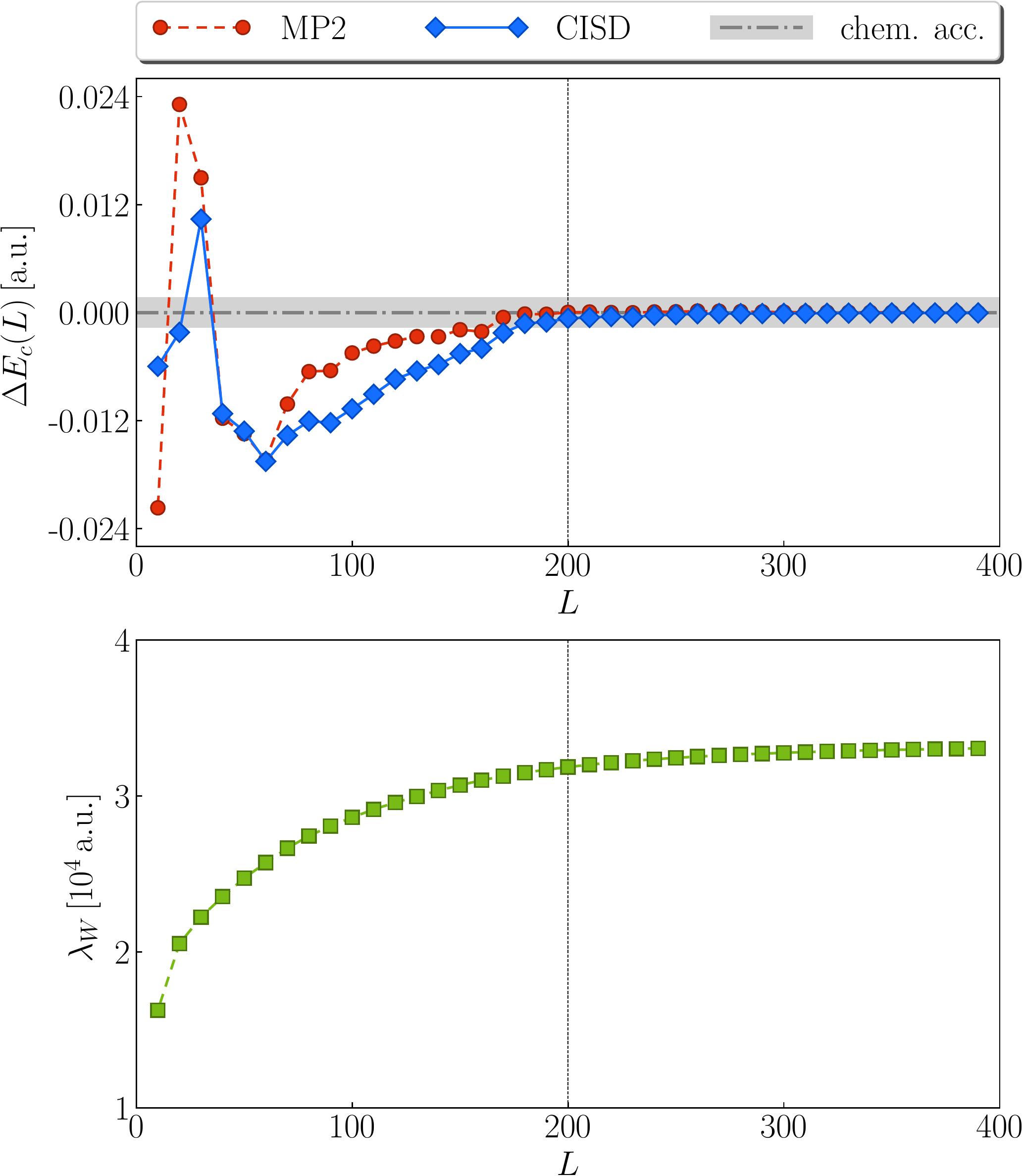}
{\small \phantom{-}~(a)~108 spin-orbital active space from {\micro}.}
\label{fig:microsoft}
\end{minipage}\hspace{8mm}
\begin{minipage}[t]{.47\textwidth}
\centering
\includegraphics[width=\linewidth]{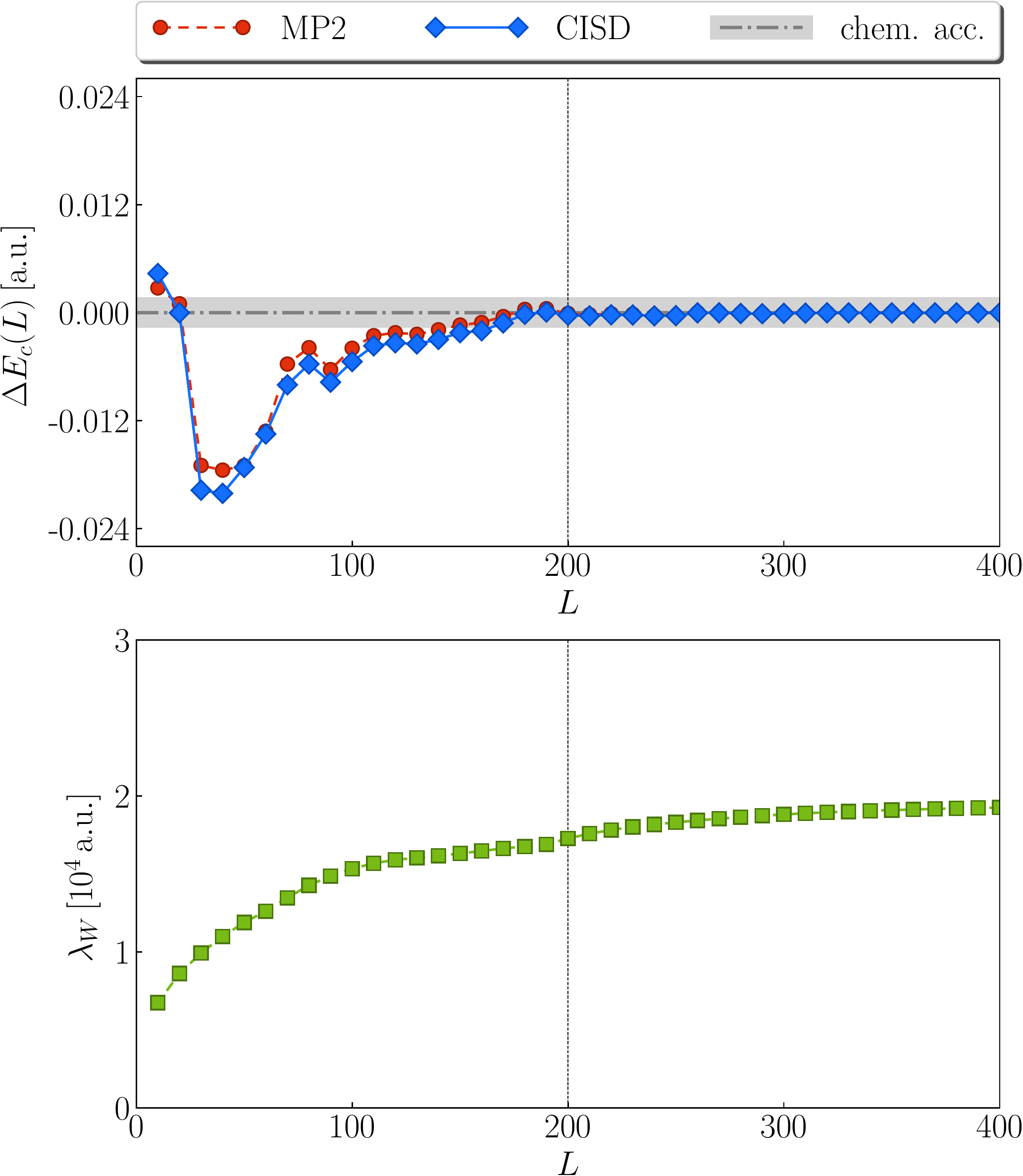}
{\small \phantom{-}~(b)~152 spin-orbital active space from {\caltech}.}
\label{fig:caltech}
\end{minipage}
\caption{\label{fig:four} Top: difference $\Delta E_c$ between truncated and untruncated correlation energy, for the FeMoco cluster with MP2 and CISD methods (red, blue). The grey shaded region represents chemical accuracy; thus, for both active spaces we expect  $L = 200$ is sufficient for our purposes. Bottom: $\lambda_W$ as a function of $L$. For {\micro} \cite{Reiher2017}, $\lambda_T = 1{,}490 \, {\rm  a.u.}$, $\lambda_V = 8{,}373 \, {\rm a.u.}$ and the maximum value of $\lambda_W$ is $34{,}552\,{\rm a.u.}$; thus, we take $\lambda = 36{,}042\,{\rm a.u}$. For {\caltech} \cite{Li2019}, $\lambda_T = 3{,}446 \, {\rm  a.u.}$, $\lambda_V = 4{,}168 \, {\rm a.u.}$ and the maximum value of $\lambda_W$ is $20{,}746\,{\rm a.u.}$; thus, we take $\lambda = 24{,}192\,{\rm a.u.}$}
\end{figure}

We perform CISD/MP2 on the exact Hamiltonian and then perform CISD/MP2 on the truncated Hamiltonian for various truncations and track the discrepancy. In \fig{four} we plot how the energies converge for both the {\micro} \cite{Reiher2017} and {\caltech} \cite{Li2019} integrals. Specifically, we plot the correlation energy (the difference between the mean-field energy and the exact energy) for the reason that the mean-field energy converges much faster than the correlation energy and so it is easier to see the trend this way. In general, CISD and MP2 are very different methods and we would expect truncation to affect them differently. However, both methods appear well converged, for both {\micro} and {\caltech} integrals, by $L=200$. This gives us confidence that the exact ground state energy of the truncated Hamiltonian would also be consistent with the exact ground state energy of the untruncated Hamiltonian by this point. Accordingly, we choose $200$ as the rank of our Coulomb operator. Note that the {\micro} and {\caltech} integrals have very different properties, and so we assume it is a coincidence that both converge around the same value of $L$.

For the integrals of {\micro} with this truncation we obtain $\lambda=36{,}042\,{\rm a.u.}$\ (see \fig{four}).
For both active spaces, we will focus on obtaining the standard ``chemical accuracy'', corresponding to $\Delta E = 0.0016$ a.u.\ \cite{Helgaker2002}.
For the integrals of {\micro}, that gives $\epsilon\approx 1.7\times 10^{-8}$.
The output size for the QROM is approximately $\log(N^2/\epsilon)$ for $N=108$, giving about $40$.
More specifically, the output size for the probabilities in the QROM is given by Eq.\ (36) in \cite{BabbushSpectra}.
In that equation, only the first term is significant, giving
\begin{equation}
\mu = \left\lceil \log\left(\frac{2\sqrt{2}\lambda}{\Delta E}\right)\right\rceil.
\end{equation}
That would give $\mu=26$ bits, except we have three steps of state preparation, which means the number of qubits for the probabilities needs to be increased by 2 to $\mu=28$.
With $N=108$ we need two registers of size $\lceil\log(N/2)\rceil = 6$, as well as two single qubit registers for the $\theta$ values, for a total of $\wid=42$ qubits for steps 3 and 4.
As we take $L=200$,
for step 1 (preparing the superposition over the $\ell$ register), only 8 qubits are needed for $\ell$, for a total of $36$ qubits output.
We will use the QROAM for steps 3 and 4, as these steps have the dominant complexity.

\subsubsection{Dirty ancillae}
If we are attempting to minimize the number of qubits used, it is convenient to combine steps 1 and 3.
That is, we use a state preparation over $\ell$, $p$, and $q$ simultaneously, and output alt values for these three indices.
Then there are only two steps of state preparation, and we can reduce the number of qubits for the keep probabilities to $\mu=27$.
That means the state preparation has an output size of $\wid_1=8+12+2+27=49$.
There will be $\wid_2=41$ qubits used for the output in step 4, because there is one fewer qubit for the keep probability.

If we use dirty qubits, then in the first state preparation we can use the $N=108$ system registers as well as the $\wid_2=41$ ancilla registers that will be used as output in the next step.
We can therefore take $\chunk=4$, which uses $(\chunk-1)\wid_1=147$ qubits, and fits within that size.
Similarly, for the second state preparation, we are able to use the output registers from the first state preparation as dirty qubits.
The cost of the QROAM compute would be $2\lceil \dm/4\rceil + 12\wid$.
For the uncompute, we can take $\chunk=128$, giving cost $2\lceil\dm/\chunk\rceil+4\chunk=2\lceil\dm/128\rceil+512$.

Taking $L=200$ gives $\dm_1=(L+1)(N^2/8+N/4)=298{,}485$ for the first preparation, and Toffoli cost
\begin{equation}
2\lceil\dm_1/4\rceil + 12\wid_1+2\lceil\dm_1/128\rceil+512= 155{,}008.
\end{equation}
For the second state preparation $\dm_2=L(N^2/8+N/4)=297{,}000$, giving Toffoli cost
\begin{equation}
2\lceil\dm_2/4\rceil + 12\wid_2+2\lceil\dm_2/128\rceil+512= 154{,}146.
\end{equation}
The total is 309,154.
The minor costs result in another 1,534 Toffolis, for a total of 310,688 (see \app{mincost}).
We find that the number of qubits for the phase estimation is
\begin{equation}
m = \left\lceil\log\left(\frac{\sqrt{2} \pi \lambda}{2\Delta E} \right)\right\rceil = 26,
\end{equation}
so we obtain an overall complexity (in terms of Toffolis)
\begin{equation}
2^m \times 310688 \approx 2.1\times 10^{13}.
\end{equation}
There is a total of 378
logical qubits needed (see \app{mincost}).

\subsubsection{Large ancilla count}
Alternatively, we can use a large number of ancilla qubits in an attempt to minimize the Toffoli count.
In the compute step it is optimal to take $\chunk=64$, and in the uncompute step it is optimal to take $\chunk=512$.
The combined complexity of compute and uncompute for each step is then
\begin{equation}
\lceil\dm/64\rceil+63\wid+\lceil\dm/512\rceil+512.
\end{equation}
In this case it is better to use steps 3 and 4 as described above, with a separate state preparation for $\ell$ in step 1.
Since there are three steps of state preparation, we should take $\wid=42$.
The Toffoli complexity for step 1 is only 200 using normal QROM.
With the $\dm$ values given above, we obtain a complexity 8,405 for step 3 and 8,379 for step 4.
The minor costs are increased to 1,594.
That gives an overall Toffoli complexity of
\begin{equation}
2^m \times 18578 \approx 1.2\times 10^{12}.
\end{equation}
Altogether there are 3,024 qubits used (see \app{mincost}).

\subsection{{\caltech} orbitals}
An alternative active space for FeMoco was advocated for in \cite{Li2019}. This work argued that the active space Hamiltonian from {\micro} did not properly capture the electronic structure of FeMoco and was classically solvable. {\caltech} introduced an alternative Hamiltonian for the FeMoco active space with $N=152$ spin-orbitals.
There it is found that $\lambda_T = 3{,}446\,{\rm  a.u.}$ and $\lambda_W=20{,}746\,{\rm a.u.}$, for a total of $\lambda = 24{,}192\,{\rm a.u}$.
The smaller value of $\lambda$ means that the number of qubits for the probabilities should be $\mu=27$ regardless of whether we merge steps 1 and 3 or not.
Since $N$ is larger than before, we now need one additional qubit for each of the orbital numbers, for a total of $\wid=43$ qubits.

\subsubsection{Dirty ancillae}
Merging steps 1 and 3, the output size is $\wid_1=51$ for the first state preparation, then $\wid_2=43$ for the second state preparation.
This time taking $\chunk=4$ would use $(\chunk-1)\wid_1=153$ qubits for the first state preparation and
$(\chunk-1)\wid_2=129$ for the second, both of which are small enough to use other qubits as dirty qubits.
We again can take $L=200$ which results in $d_1=(L+1)(N^2/8+N/4)=588{,}126$ for the first preparation and $d_2=L(N^2/8+N/4)=585{,}200$ for the second preparation (step 4).
Using $\chunk=128$ for the uncompute again yields a cost for the first preparation of
\begin{equation}
2\lceil\dm_1/4\rceil + 12\wid_1+2\lceil\dm_1/128\rceil+512= 304{,}378.
\end{equation}
The cost of the second preparation is approximately
\begin{equation}
2\lceil\dm_2/4\rceil + 12\wid_2+2\lceil\dm_2/128\rceil+512= 302{,}772.
\end{equation}
The total of the minor costs is 1,818 for a total of 608,968.
This time we find that $\log(\sqrt 2 \pi \lambda)/(2\Delta E)$ is very slightly larger than 25
\begin{equation}
\log\left(\frac{\sqrt{2} \pi \lambda}{2\Delta E} \right)\approx 25.0015.
\end{equation}
It would be unreasonably inefficient to round up to $m=26$.
Instead we can allow very slightly less error in other parts of the algorithm (which does not affect the complexity significantly because the algorithm depends on that error logarithmically) and take $m=25$.
Then we get a total cost
\begin{equation}
2^m\times 608968\approx 2.0\times 10^{13}.
\end{equation}
The total number of logical qubits is 437.

\subsubsection{Large ancilla count}
If we use a large number of ancilla qubits, we should again take the $\chunk$ sizes as $64$ and $512$ in the compute and uncompute steps, respectively.
We again perform steps 1 and 3 separately for this approach.
The $\wid$ will be $43$ for both steps 3 and 4.
Then using $\lceil\dm/64\rceil+63\wid+\lceil\dm/512\rceil+512$ gives Toffoli costs $13{,}560$ and $13{,}508$ for steps 3 and 4, respectively.
This time the additional step of preparation needs 200 Toffolis for QROM.
The minor costs are increased to 1,872, for
a total of 29,140.
That gives an overall complexity in terms of Toffolis
\begin{equation}
2^m\times 29140 \approx 9.8\times 10^{11}.
\end{equation}
The total number of qubits is 3,143.

In summary, the Toffoli costs are as given in \tab{costs}.
In this table we have given approximate formulae including only the leading terms, and taken $\chunk=4$ for the QROM compute circuits.

\section{Exploiting sparsity in the Coulomb operator}
\label{sec:sparsity}

Next we provide a strategy for further reducing constant factors when qubitizing the non-factorized quantum chemistry Hamiltonians.
It is also possible to apply this strategy to the factorized form, but the result is worse for the case of FeMoco, so we will not address it here.
This strategy will likely reduce the T complexity in practical situations, but only by constant factors. We focus on the form of the Coulomb operator
\begin{equation}
    V = \sum_{\alpha,\beta\in \{\uparrow, \downarrow\}} \sum_{p,q,r,s=1}^{N/2} V_{pqrs}a^\dagger_{p,\alpha} a_{q,\alpha} a^\dagger_{r,\beta} a_{s,\beta},
\end{equation}
which has the truncated form
\begin{equation}
\label{eq:V_truncation}
    \widetilde{V}^{(c)}  = \sum_{\alpha,\beta\in \{\uparrow, \downarrow\}} \sum_{p,q,r,s=1}^{N/2} \widetilde{V}_{pqrs}^{(c)} a^\dagger_{p,\alpha} a_{q,\alpha} a^\dagger_{r,\beta} a_{s,\beta},
    \qquad \qquad
    \widetilde{V}_{pqrs}^{(c)} =
    \begin{cases}
        V_{pqrs}, & \left |V_{pqrs} \right | \geq c,\\
               0, & \left |V_{pqrs} \right | < c .
    \end{cases}
\end{equation}
The purpose of this truncation is to induce sparsity in the operators by removing near-zero elements. The idea of reducing quantum simulation costs by exploiting sparsity in the Coulomb operator was first explored in \cite{McClean2014}, but in the context of Trotter based methods. The idea of the approach here is to choose the value of $c$ to be as large as possible while still leaving classical correlated approximations such as CISD within chemical accuracy (essentially the same procedure we used for choosing $L$ shown in \fig{four}). This is shown in \fig{e2}.

\begin{figure}
\begin{minipage}[t]{.47\textwidth}
\centering
\includegraphics[width=\linewidth]{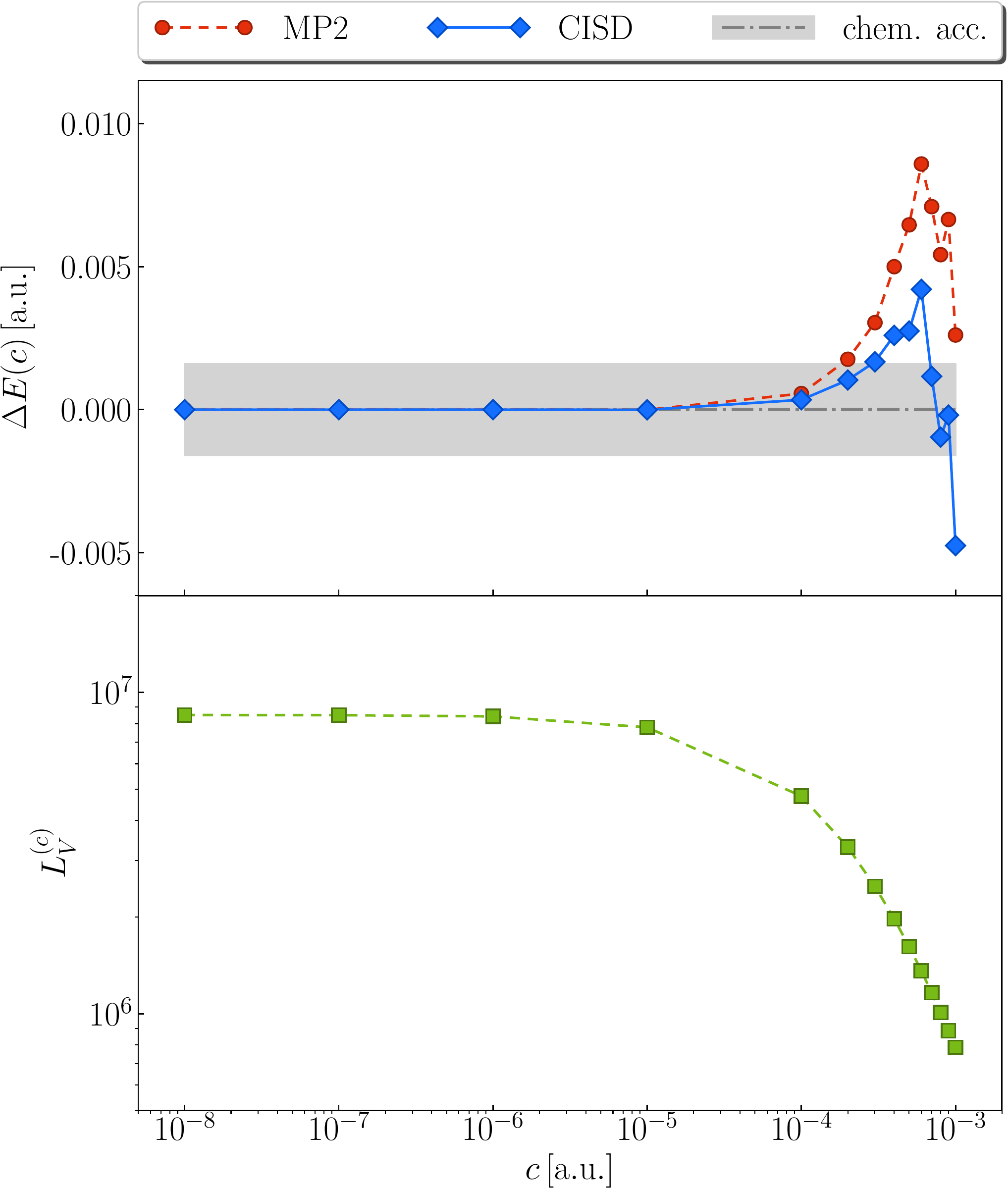}
{{\small \phantom{-}~(a)~108 spin-orbital active space from {\micro}.}}
\label{fig:microsoft_e2}
\end{minipage}\hspace{8mm}
\begin{minipage}[t]{.47\textwidth}
\centering
\includegraphics[width=\linewidth]{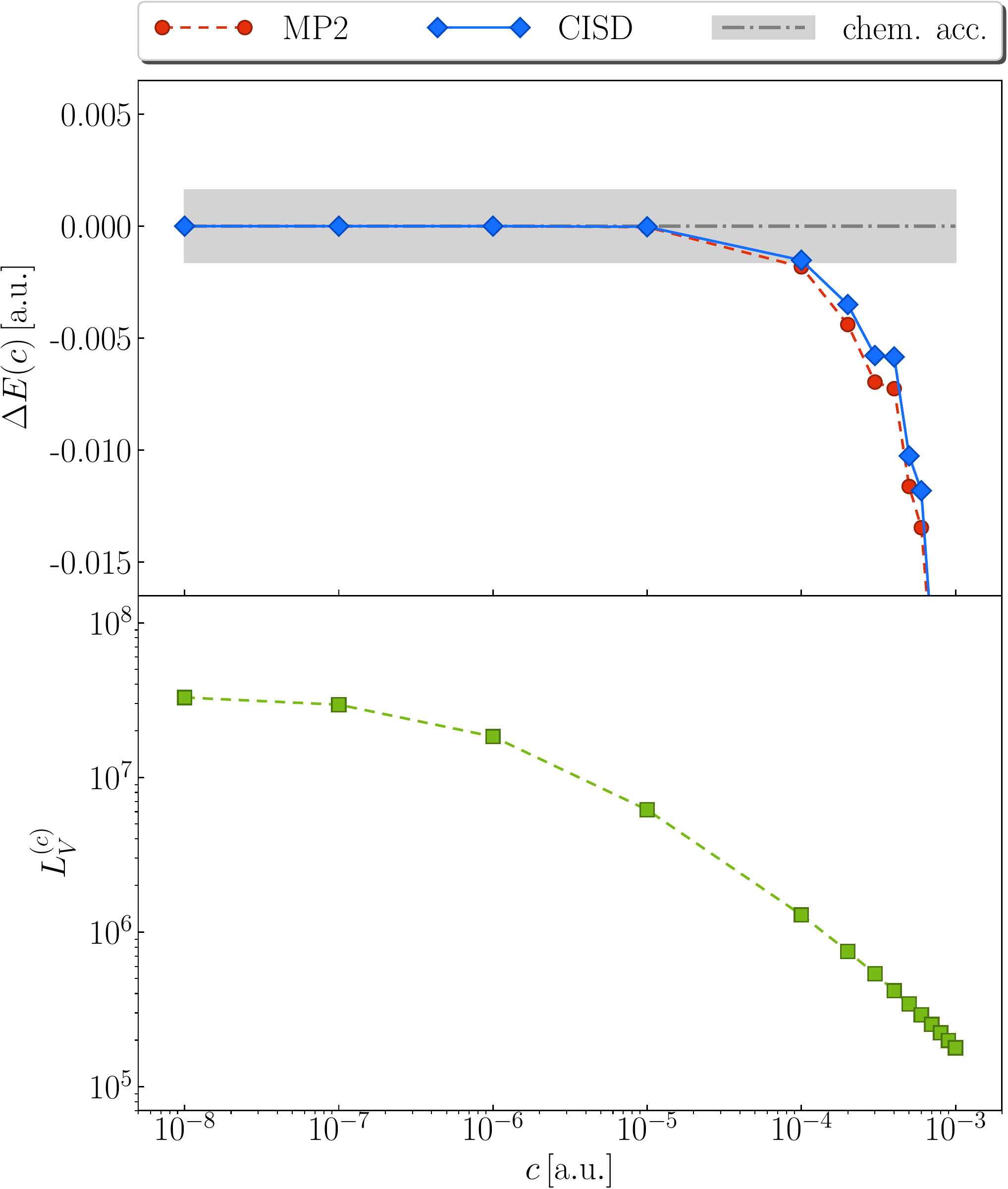}
{{\small \phantom{-}~(b)~152 spin-orbital active space from {\caltech}.}}
\label{fig:caltech_e2}
\end{minipage}
\caption{\label{fig:e2} We show the result of performing the truncation of \eq{V_truncation}. Top: we show difference $\Delta E(c)$ between truncated and untruncated correlation energy, for the FeMoco cluster with MP2 and CISD methods (red, blue) as a function of the $c$ parameter in \eq{V_truncation}. The grey shaded region represents chemical accuracy. Bottom: $L_V^{(c)}$ as a function of $c$. For {\micro} \cite{Reiher2017}, we can safely truncate at $c = 0.0002\, {\rm a.u.}$, which corresponds to $L_V^{(c)} = 3{,}300{,}568$. For {\caltech} \cite{Li2019}, we can safely truncate at $c = 0.0001 \, {\rm a.u.}$, which corresponds to $L_V^{(c)} = 1{,}291{,}648$.}
\end{figure}

We will define $L_V^{(c)}$ as the number of nonzero values of $\widetilde{V}_{pqrs}^{(c)}$. In general, we expect that $L_V^{(c)} = {\cal O}(N^4)$; which is to say that this truncation should not asymptotically change the sparsity of the operators. However, in practice we do expect to find that $L_V^{(c)} < N^4 / 16$; which is to say that we do expect there to be additional sparsity in these operators. While the matter of exactly how much sparsity exists is highly system dependent, we might sometimes desire an algorithm that exploits this sparsity, even if it is highly unstructured.

It is possible to perform a state preparation that has cost dependent on the number of nonzero elements, at the cost of a slightly larger number of ancillae.
Consider the state preparation of \cite{BabbushSpectra}, which has the following steps.
\begin{enumerate}
\item Create an equal superposition over the system registers
\begin{equation}
\frac 1{\sqrt{d}} \sum_{j=1}^d \ket{j}.
\end{equation}
\item Use QROM indexed on the system registers to output alt values and keep values
\begin{equation}
\frac 1{\sqrt{d}} \sum_{j=1}^d \ket{j}\ket{{\rm alt}_j}\ket{{\rm keep}_j}.
\end{equation}
\item Use another ancilla in an equal superposition over $2^\mu$ values where $\mu$ is the number of digits for keep.
Perform an inequality test between this register and the keep register, and swap the contents of the first two registers (the index and alternative index registers) based on the result of this inequality test.
\end{enumerate}
The total number of ancillae used by this state preparation is $2\lceil\log d\rceil+ 2\mu+1$.

In order to perform the preparation in the sparse case, instead of iterating over the index register, we can output the index register in the same way as the alternate index register.
That is, we have a register which iterates over the number of nonzero amplitudes in the state we aim to prepare.
Denoting that number $d$, our steps are as follows.
\begin{enumerate}
\item Create an equal superposition over the register indexing the nonzero entries.
\begin{equation}
\frac 1{\sqrt{d}} \sum_{j=1}^d \ket{j}.
\end{equation}
\item Use QROM indexed on this first register to output ind, alt, and keep values
\begin{equation}
\frac 1{\sqrt{d}} \sum_{j=1}^d \ket{j}\ket{{\rm ind}_j}\ket{{\rm alt}_j}\ket{{\rm keep}_j}.
\end{equation}
\item Use another ancilla in an equal superposition over $2^\mu$ values, and perform an inequality test between this register and the keep register.
Based on the result of this inequality test, swap the contents of the ind and alt registers.
\end{enumerate}
The desired state will then be produced in the second register where we output ind.
The value ${\rm ind}_j$ is simply the index for the $j$th nonzero amplitude.
The correctness of this state preparation routine follows immediately from the correctness of the routine in \cite{BabbushSpectra},
because after the QROM lookup the state in the registers excluding the first is equivalent to that in \cite{BabbushSpectra}.
The Toffoli complexity of this modified state preparation procedure for sparse states depends on the number of nonzero entries rather than the dimension.

For our application, the non-factorized form of the Hamiltonian with truncation of the Coulomb operator can be expressed using the Jordan-Wigner representation.
Using \eq{full_hamiltonian} and the symmetries $T_{pq}=T_{qp}$, $V_{pqrs}=V_{qprs}=V_{pqsr}$, the Hamiltonian can be written as
\begin{align}
	H &= \frac 12 \sum_{\sigma \in \{\uparrow, \downarrow\}}\sum_{p,q=1}^{N/2} T_{pq} (a^\dagger_{p,\sigma} a_{q,\sigma}+a^\dagger_{q,\sigma} a_{p,\sigma}) \nn & \quad + \frac 14 \sum_{\alpha,\beta\in \{\uparrow, \downarrow\}} \sum_{p,q,r,s=1}^{N/2} V_{pqrs}(a^\dagger_{p,\alpha} a_{q,\alpha}+a^\dagger_{q,\alpha} a_{p,\alpha})
	(a^\dagger_{r,\beta} a_{s,\beta}+a^\dagger_{s,\beta} a_{r,\beta}).
\end{align}
Then using the Jordan-Wigner representation via \eq{jw} and truncating $V$, the Hamiltonian can be represented as
\begin{equation}\label{eq:nonfacham}
H \mapsto  \sum_{\sigma \in \{\uparrow, \downarrow\}}\sum_{p,q=1}^{N/2} T_{pq}Q_{pq\sigma}
+ \sum_{\alpha,\beta\in \{\uparrow, \downarrow\}} \sum_{p,q,r,s=1}^{N/2} \widetilde{V}_{pqrs}^{(c)} Q_{pq\alpha} Q_{rs\beta} ,
\end{equation}
where
\begin{equation}
Q_{pq\sigma}=\begin{cases}
    X_{p,\sigma} \vec Z X_{q,\sigma}, & p< q, \\
    Y_{p,\sigma} \vec Z Y_{q,\sigma}, & p> q, \\
 \frac 12\left(\openone - Z_{p,\sigma}\right) , & p=q.
\end{cases}
\end{equation}
Here we have again used the symmetries $T_{pq}=T_{qp}$, $V_{pqrs}=V_{qprs}=V_{pqsr}$.
Using $p<q$ versus $p>q$ to distinguish between $X_{p,\sigma} \vec Z X_{q,\sigma}$ and $Y_{p,\sigma} \vec Z Y_{q,\sigma}$ is a convenient way to perform the controlled operations as described in \fig{selecth}.
This is of the form of a linear combination of unitaries, and the state we need to prepare is of the form
\begin{align}
&\ket{0}\ket{+}\ket{0}\ket{0}\sum_{\sigma \in \{\uparrow, \downarrow\}}\sum_{p,q=1}^{N/2} \sqrt{\frac{|T_{pq}|}{\lambda}} \ket{\theta_{pq00}^{T}}\ket{p,q,\sigma}\ket{0}\nn
    &\quad + \ket{1}\ket{+}\ket{+}\ket{+} \!\!\!\!\!\! \sum_{\alpha,\beta\in \{\uparrow, \downarrow\}}\sum_{p,q,r,s=1}^{N/2} \!\! \sqrt{\frac{|\widetilde{V}_{pqrs}^{(c)}|}{\lambda}} \ket{\theta_{pqrs}^{V}}\ket{p,q,\alpha}\ket{r,s,\beta}.\label{eq:spprep}
\end{align}
The first register is just a convenient replacement for the $\ket{\ell}$ register, and is used to distinguish between the one and two body terms.
The second register is used to distinguish between $\openone$ and $Z_{p,\alpha}$ for $p=q$, and the third register is used to distinguish between $\openone$ and $Z_{r,\beta}$ for $r=s$.
The second, third, and fourth registers will be used to generate the symmetries of the state.
The fifth register is used to give the appropriate sign of the $T_{pq}$ and $\widetilde{V}_{pqrs}^{(c)}$ weightings, which are now combined instead of having two separate registers as before.

We can again use symmetry to reduce the number of coefficients that need to be prepared.
For $V_{pqrs}$ there is symmetry between exchanging $p$ and $q$, between exchanging $r$ and $s$, and between exchanging the $p,q$ and $r,s$ pairs, as described in \eq{symmetries}.
For $T_{pq}$ there is symmetry between $p$ and $q$.
For this reason we will initially only prepare amplitudes for $p\le q$, $r\le s$, and $p\le r$ for $V$, and $p\le q$ for $T$.

To simplify the description, we will introduce some notation,
\begin{equation}
\zeta_{pq} := \begin{cases}
\sqrt{2}, & p<q, \\
1, & p=q, \\
0, & p>q.
\end{cases}
\end{equation}
We will also allow $\zeta$ to be used with four subscripts, which means
\begin{equation}
\zeta_{pqrs} := \begin{cases}
\sqrt{2}, & pq<rs, \\
1, & pq=rs, \\
0, & pq>rs,
\end{cases}
\end{equation}
where the notation $pq<rs$ indicates that either $p<r$ or $p=r$ and $q<s$, and $pq=rs$ indicates that $p=r$ and $q=s$.
That is, if these are the composite indices for the matrix $W$, then $pq<rs$ indicates the upper triangle, and $pq=rs$ indicates the diagonal.

In terms of this notation, the state produced will initially be
\begin{align}
&    \ket{0}\ket{+}\ket{0}\ket{0}\sum_{\sigma \in \{\uparrow, \downarrow\}}\sum_{p,q=1}^{N/2} \zeta_{pq}\sqrt{\frac{|T_{pq}|}{\lambda}}\ket{\theta_{pq}^{T}}\ket{p,q,\sigma}\ket{0} \nonumber \\
&    +\ket{1}\ket{+}\ket{+}\ket{+} \!\!\! \sum_{\alpha,\beta\in \{\uparrow, \downarrow\}}\sum_{p,q,r,s=1}^{N/2}\zeta_{pq}\zeta_{rs}\zeta_{pqrs}
    \sqrt{\frac{|\widetilde{V}_{pqrs}^{(c)}|}{\lambda}}\ket{\theta_{pqrs}^{V}}\ket{p,q,\alpha}\ket{r,s,\beta}.
\end{align}
Using $\zeta$ ensures that the number of nonzero terms is about $1/2$ as much for $T$, and about $1/8$ as much for $V$.
These sparse entries can be prepared by the technique described above for sparse state preparation.
For $T$ the only terms prepared here are those where $p\le q$, and for $V$ only where $p\le q$, $r\le s$, and $pq\le rs$.
Next we perform three steps.
\begin{enumerate}
    \item Swap the $p,q$ and $r,s$ registers controlled on the state of the fourth register.
    \item Swap the $p$ and $q$ controlled on the state of the second register.
    \item Swap the $r$ and $s$ controlled on the state of the third register.
\end{enumerate}

To show the effect of this, we will first show it for the first component of the state, with $0$ in the first register and weightings dependent on $T$.
The first and third controlled swaps have no effect because third and fourth registers are zero.
The second controlled swap gives us
\begin{align}\label{eq:cswap}
&    \ket{0}\ket{0}\ket{0}\ket{0}\sum_{\sigma \in \{\uparrow, \downarrow\}}\sum_{p,q=1}^{N/2} \zeta_{pq}\sqrt{\frac{|T_{pq}|}{2\lambda}}\ket{\theta_{pq}^{T}}\ket{p,q,\sigma}\ket{0}\nn
&\quad +\ket{0}\ket{1}\ket{0}\ket{0}\sum_{\sigma \in \{\uparrow, \downarrow\}}\sum_{p,q=1}^{N/2} \zeta_{pq}\sqrt{\frac{|T_{pq}|}{2\lambda}}\ket{\theta_{pq}^{T}}\ket{q,p,\sigma}\ket{0} \nonumber \\ 
 &  =  \ket{0}\ket{0}\ket{0}\ket{0}\sum_{\sigma \in \{\uparrow, \downarrow\}}\sum_{p,q=1}^{N/2} \zeta_{pq}\sqrt{\frac{|T_{pq}|}{2\lambda}}\ket{\theta_{pq}^{T}}\ket{p,q,\sigma}\ket{0}\nn
&\quad +\ket{0}\ket{1}\ket{0}\ket{0}\sum_{\sigma \in \{\uparrow, \downarrow\}}\sum_{p,q=1}^{N/2} \zeta_{qp}\sqrt{\frac{|T_{pq}|}{2\lambda}}\ket{\theta_{pq}^{T}}\ket{p,q,\sigma}\ket{0} \nonumber \\
 & = \ket{0}\sum_{\sigma \in \{\uparrow, \downarrow\}}\sum_{p,q=1}^{N/2} \sqrt{\frac{|T_{pq}|}{2\lambda}}\left(\zeta_{pq}\ket{0}+\zeta_{qp}\ket{1}\right)\ket{0}\ket{0}\ket{\theta_{pq}^{T}}\ket{p,q,\sigma}\ket{0}
 \nonumber \\
 & = \ket{0}\sum_{\sigma \in \{\uparrow, \downarrow\}}\sum_{p,q=1}^{N/2} \sqrt{\frac{|T_{pq}|}{\lambda}}\ket{\kappa_{pq}}\ket{\theta_{pq}^{T}}\ket{0}\ket{0}\ket{p,q,\sigma}\ket{0},
\end{align}
where we have defined the state labelling
\begin{equation}
\kappa_{pq} := \begin{cases}
0, & p<q, \\
1, & p>q, \\
+, & p=q.
\end{cases}
\end{equation}

Using similar notation we can describe the effect of the controlled swaps between $p,q$ and $r,s$ in a compact way.
The effect of this controlled swap on the second component of the state will be
\begin{equation}
\ket{1} \sum_{\alpha,\beta\in \{\uparrow, \downarrow\}}\sum_{p,q,r,s=1}^{N/2}\zeta_{pq}\zeta_{rs}
    \sqrt{\frac{|\widetilde{V}_{pqrs}^{(c)}|}{\lambda}}\ket{+}\ket{+}\ket{\kappa_{pqrs}}\ket{\theta_{pqrs}^{V}}\ket{p,q,\alpha}\ket{r,s,\beta},
\end{equation}
where
\begin{equation}
\kappa_{pqrs} := \begin{cases}
0, & pq<rs, \\
1, & pq>rs, \\
+, & pq=rs.
\end{cases}
\end{equation}
The reasoning is exactly the same as for the $T$ component of the state.
The net effect is that we remove the $\zeta_{pqrs}$ and the control register is now in the state $\ket{\kappa_{pqrs}}$.
The controlled swaps with $p,q$ and $r,s$ have similar effect, giving the state as in \eq{spprep}, except the second, third and fourth registers are not in $\ket{+}$ states.
Nevertheless, for $p=q$ the second register is in an equal superposition of $\ket{0}$ and $\ket{1}$, so we can we use it to select between $\openone$ and $Z_{p,\sigma}$ for the controlled operations.
Similarly for $r=s$ the third register is in an equal superposition of $\ket{0}$ and $\ket{1}$.
These features are sufficient to give the linear combination of unitaries required.

Hence we can prepare the desired state with about $1/8$ as many nonzero entries using the sparse preparation procedure, and controlled swaps to generate the entries that are identical due to symmetries.
For the controlled operations that need to be performed, note that the form of the Hamiltonian without factorization has $a^\dagger_{p,\sigma} a_{q,\sigma}$ and $a^\dagger_{r,\sigma} a_{s,\sigma}$ operations that are mapped to the Pauli operators in exactly the same way as for the factorized form.
This means that these controlled operations can be performed using exactly the same $\textsc{select}_1$ and $\textsc{select}_2$ operations as for the factorized form.

\section{Complexity for sparse preparation}
\label{sec:sparse_complex}

\subsection{{\micro} orbitals}
Next we use the results of \sec{sparsity} to determine the complexity in the case of using a sparse preparation with the non-factorized Hamiltonian.
For the integrals of {\micro} 
we have $\lambda_T = 1{,}490 \, {\rm  a.u.}$ and $\lambda_V = 8{,}373 \, {\rm a.u.}$ for a total of $\lambda=9{,}863 \, {\rm a.u}$.
The number of qubits for the phase estimation is
\begin{equation}
m = \left\lceil\log\left(\frac{\sqrt{2} \pi \lambda}{2\Delta E} \right)\right\rceil = 24,
\end{equation}
and the output size for the keep probabilities for the QROAM is
\begin{equation}
\mu = \left\lceil \log\left(\frac{2\sqrt{2}\lambda}{\Delta E}\right)\right\rceil =25.
\end{equation}
There are eight registers of size $\lceil \log  (N/2) \rceil$, because the sparse preparation scheme needs to output ind values and alt values of $p,q,r,s$.
There are also $2$ qubits needed for the two output values of $\theta$ (one for the ind and one for the alt values of $p,q,r,s$), as well as $2$ qubits used for ind and alt values of the first register which distinguishes between $T$ and $V$.
As a result the QROAM output size is $\wid = \mu+8\lceil \log  (N/2) \rceil+4$.
The number of orbitals is $N=108$, giving $\lceil \log  (N/2) \rceil=6$ and therefore $\wid=77$.

Here we just consider the use of a large number of ancillae to minimize the Toffoli cost.
The cost of the preparation is then
\begin{equation}
\lceil\dm/\chunk_1\rceil+\wid(\chunk_1-1)
\end{equation}
and of the inverse preparation is
\begin{equation}
\lceil\dm/\chunk_2\rceil+\chunk_2 .
\end{equation}
Using a truncation threshold of $0.0002\,{\rm a.u.}$, we find that $3{,}300{,}568$ nonzero values of $\widetilde{V}_{pqrs}^{(c)}$ are needed.
From the symmetries, the number of unique nonzero values is about 435,023.
Adding to that the $N^2/8+N/4=1485$ unique terms for $T$, we get 436,508.
Then the optimal values are $k_1=2^6$ and $k_2=2^9$, giving a preparation cost of 11,672 and an inverse preparation cost of 1,365, for a total of 13,037.
There are 746 Toffolis for the minor costs (see \app{mincost}) giving a total of 13,783.
That gives an overall complexity in terms of Toffolis
\begin{equation}
2^m\times 13783 \approx 2.3\times 10^{11}.
\end{equation}
The total number of logical qubits needed is 5,103.

\subsection{{\caltech} orbitals}
For the integrals of {\caltech} we have $\lambda_T = 3{,}446 \, {\rm  a.u.}$, $\lambda_V = 4{,}168 \, {\rm a.u.}$ so
$\lambda=7{,}614\,{\rm a.u.}$, and $N=152$.
That gives $m=24$, $\mu=24$, and $\lceil \log  (N/2) \rceil=7$.
That gives the output size for the QROAM of $M=84$.
Using a truncation threshold of $0.0001\,{\rm a.u.}$, we find that 1,291,648 nonzero values are needed.
From the symmetries, the number of unique nonzero values is about 176,572.
Adding to that the $N^2/8+N/4=2926$ unique terms for $T$, we get 179,498.
Using $k_1=2^5$ and $k_2=2^9$ gives a preparation cost of $8{,}214$ and an inverse preparation cost of $863$.
(The optimal value of $k_1$ is actually $2^6$, but the improvement is only slight and it almost doubles the number of ancilla needed.)
The minor costs have a total of 918.
That gives an overall complexity in terms of Toffolis
\begin{equation}
2^m\times 9995 \approx 1.7\times 10^{11}.
\end{equation}
The number of logical qubits used is 2,904.

The Toffoli performance of this algorithm is the best of the alternatives we have so far considered.
As this is the most promising for implementation, we consider a further improvement by better allocating the allowable error between the different parts of the algorithm.
In this case
\begin{equation}
\log\left(\frac{\sqrt{2} \pi \lambda}{2\Delta E} \right) \approx 23.33.
\end{equation}
In \cite{BabbushSpectra} the error is given in Eq.~(23) which is equivalent to
\begin{equation}
\Delta E^2 \approx \lambda^2 \left[\left(\frac{\pi}{2^{m+1}}\right)^2+(\epsilon_{\textsc{prep}}+\pi\epsilon_{\textsc{QFT}})^2\right].
\end{equation}
Here $\epsilon_{\textsc{prep}}$ is the error allowed in the state preparation and $\epsilon_{\textsc{QFT}}$ is the gate synthesis error allowed in the inverse quantum Fourier transform.
The formula for $m$ is obtained by equally allocating the allowable (squared) error between $\left({\pi}/{2^{m+1}}\right)^2$ and $(\epsilon_{\textsc{prep}}+\pi\epsilon_{\textsc{QFT}})^2$.
If we instead allow approximately 26\% more error from the phase estimation, then it is possible to use $m=23$ (which reduces the number of logical qubits to 2,903).
The allowable error in $\epsilon_{\textsc{QFT}}$ can be reduced to compensate, which has a negligible gate cost because it is not multiplied by $2^m$.
As a result, the complexity in terms of Toffolis then becomes
\begin{equation}
2^m\times 9995 \approx 8.4\times 10^{10}.
\end{equation}
This is a full order of magnitude improvement over the low rank factorization method, and \emph{three} orders of magnitude improvement over the approach of \cite{Reiher2017}.

\section{Discussion}
\label{sec:conc}

The cost of performing Toffoli gates may be estimated as follows.
The efficient CCZ factory from \cite{Gidney2018pub} has a rectangular footprint of $12 \times 6$ logical qubit patches.
The paper specifies a code distance of $d=31$ which, in the rotated surface code, means each patch covers $2 \cdot 31^2 \approx 2{,}000$ physical qubits.
The factory outputs a CCZ state every $5.5d$ cycles which, assuming a surface code cycle time of 1 microsecond, is once every 170 microseconds.
Thus, the spacetime volume of the factory is $2{,}000 \cdot (12 \times 6)$ physical qubits times $170$ microseconds which equals roughly $24$ qubitseconds.
Every Toffoli we perform requires at least this much spacetime volume. With the same overhead one can use these techniques to produce two magic states.

\begin{table*}[t]
\small
    \begin{tabular}{ c | c | c | c}
    & leading order scaling &
    {\micro} \cite{Reiher2017} &
   {\caltech} \cite{Li2019} \\    \hline \hline
Toffolis for dirty ancilla algorithm &
    $\sqrt 2 \pi \lambda  NM L / (4 \Delta E)$ & $2.1\times 10^{13}$ & $2.0\times 10^{13}$\\
logical qubits for dirty ancilla algorithm & $N + \log(L^2N^8)+4\mu+m$ & 378 & 437 \\
    \hline
    Toffolis for many ancilla algorithm & $\pi \lambda N\sqrt{L\wid}/\Delta E$ & $1.2\times 10^{12}$ & $9.8\times 10^{11}$ \\
    logical qubits for many ancilla algorithm &
     $N + N\sqrt{L\wid/8}$ & 3,024 & 3,143 \\
         \hline
Toffolis for sparse algorithm & NA & $2.3\times 10^{11}$ & $8.4\times 10^{10}$\\
logical qubits for sparse algorithm & NA & 5,103 & 2,903 \\
    \hline
  \end{tabular}
  \caption{\label{tab:costs} The leading order Toffoli and spatial complexities for the algorithms using dirty ancilla or allowing a large ancilla count.
  The formulae given are approximate, with the counts determined in a more exact way explained in the text. Here $N$ is the number of spin-orbitals, $L$ is the rank of the Coulomb operator, $\lambda$ is the 1-norm of the Hamiltonian as defined in \eq{lambdas}, $\Delta E$ is the target precision for phase estimation, $\wid$ is the output size for the QROAM, $\mu$ is the number of qubits used for ``keep'' probabilities, and $m$ is the number of qubits used for the phase estimation.}
\end{table*}

The lowest Toffoli count that we report in \tab{costs} is $8 \times 10^{10}$.
Combined with the $24$ qubitsecond spacetime volume for distilling a CCZ state, the spacetime volume of the algorithm is about three megaqubitweeks.
Three megaqubitweeks of spacetime volume means that if we use ``only" three million physical qubits then the computation must run for at least a week; if we want the computation to finish in a day, we need at least 23 million qubits.
We want computations to finish in less than a day, and we don't want to use 23 million qubits. This implies that, when attempting to move our algorithm into the regime of practical computations, we should focus on optimizations that reduce spacetime volume (such as exploiting symmetries in the Hamiltonian to reduce the size of QROM reads) instead of optimizations that convert spacetime volume (such as performing parallel distillation of states, which reduces time at the cost of space).

For the purpose of simulating arbitrary basis chemistry Hamiltonians, our approach is the best scaling and also the most practical shown to date. Despite us using a larger and more accurate active space, we have certainly improved over the $10^{14}$ T gates required by \cite{Reiher2017}. Based on the numbers above, a lower bound on the spacetime volume of the \cite{Reiher2017} algorithm is roughly two gigaqubitweeks of distillation, which is about seven hundred times more than our approach. This ignores storage, routing, and Clifford operations, but this level of comparison seems appropriate considering that the large spacetime volume of state distillation is still likely to dominate the overall cost.

While our many ancilla algorithm has ${\cal O}(N^{3/2})$ spatial complexity and would require a few thousand logical qubits to simulate FeMoco, at three megaqubitweeks of spacetime volume required just for state distillation, our algorithm is bottlenecked by the Toffoli complexity, and not by the logical qubit costs. Thus, despite the increased spatial costs, we regard the many qubit algorithm as more practical than the dirty ancilla algorithm. However, recent work \cite{Litinski2019} has suggested that by distinguishing between qubits that are, and are not, used for error detection, one can use a lower code distance for most of the logical qubits required for magic state distillation. While explicit factories realizing these improvements have not yet been laid out, these improvements could significantly reduce our estimates of the required spacetime volume for distillation, perhaps to less than a megaqubitweek.

One might wonder if there are other techniques that could be used to reduce the Toffoli count at the expense of extra ancilla. In our algorithm, the Toffolis are coming from our use of QROM. Unfortunately, a lower bound proven in \cite{Lowpreparation} suggests that no further tradeoffs of this type will be possible that asymptotically reduce the Toffoli count of the QROM we are using. A more promising approach might be to utilize additional structure in order to reduce the effective number of coefficients that must be read from QROM. We presented one such strategy based on leveraging sparsity in the Coulomb operator, and another based on leveraging the rank deficiency of the Coulomb operator.

In some cases, sparsity in the Hamiltonian arises due to locality of interactions (especially for large systems) \cite{McClean2014}, but sparsity also arises for reasons having to do with the symmetry point groups of molecules in real space and the symmetry point groups of molecular orbitals in an active space (this would likely be the origin of sparsity in the FeMoco active space, for instance). It might also be possible to exploit these symmetries (which are embedded in the coefficient tensors) using techniques from group theory. Likewise, it might be possible to further exploit the low rankness of the Coulomb operator by using the second tensor factorization discussed in \cite{Motta2018}.

 Another avenue for reducing the total cost would be to reduce the size of the $\lambda$ parameter. For example, a technique described in Section V of \cite{Rubin2018} may help reduce the $\lambda$ parameter by exploiting $n$-representability conditions. Yet another approach might be to incorporate mean-field background subtraction \cite{Al-Saidi2006}. One could also explore options for selecting the active space differently. For instance, by using different orbitals that are more local, or more symmetric, one could potentially induce more sparsity in the Hamiltonian. Or, one could try to select active space orbitals with a goal of reducing $\lambda$. In both of these contexts, the use of pseudopotentials might help \cite{Vanderbilt1990}.

Of course, it might be the case that the best path forward uses a different algorithm or representation entirely. For instance, it is possible that Trotterization based on the Trotter step of \cite{Motta2018} would provide a significant improvement over the Trotter results in \cite{Reiher2017}. In practice, Trotter errors (which determine the number of Trotter steps required) depend sensitively on the structure of the Trotter step, and this has not yet been analyzed for the approach of \cite{Motta2018}. However, this seems like a promising direction. In that context one might also wonder if using higher-order Trotter formulae could further bring down the costs, as observed for simulating simpler models in \cite{Childs2017}.

Given the high overhead that appear to be required for simulating FeMoco in a molecular orbital basis, we might wonder how practical it would be to perform the simulation in the plane wave basis. No approach based on second quantization is sensible here because one would need perhaps $N=10^6$ plane waves to simulate a FeMoco active space to sufficiently high accuracy. Nevertheless, the first quantized algorithm of \cite{BabbushContinuum} continues to look competitive. While the constant factors still need to be worked out, that approach scales as $\widetilde{\cal O}(N^{1/3}\eta^{8/3}/\Delta E)$ where the number of electrons is $\eta$. If $N=10^6$, $\eta=54$ (as in the {\micro} active space) and $\Delta E = 0.0016 \, {\rm a.u.}$, the quantity $N^{1/3}\eta^{8/3}/\Delta E$ is roughly $10^9$. Thus, if further symmetries can be exploited in this simulation to further reduce costs (given that there will be other overheads), we expect this might be a viable way forward; however, more work is clearly needed.

\acknowledgements

The authors thank Garnet Kin-Lic Chan, Austin Fowler, Yuval Sanders, and Hartmut Neven for helpful discussions. DWB is funded by Australian Research Council Discovery Projects DP160102426 and DP190102633. MM was supported by the United States Department of Energy via the grant {DE-SC}0019374 awarded to Garnet Kin-Lic Chan.

\bibliographystyle{apsrev4-1-fixed}

\appendix
\section{Cost of computing table lookups assisted by dirty ancillae}
\label{app:lookup}

In \cite{Lowpreparation} it is explained how to perform an efficient table lookup (a QROAM read) assisted by dirty ancillae.
Here ``dirty'' ancillae are used to mean ancillae that need not be initialized to some known state before the procedure, and are returned to their initial state at the end.
We provide an equivalent technique, and prove the following result.
\begin{theorem}\label{thdirty}
Given a function $f: \mathbb{Z}_\dm \to \mathbb{Z}_2^\wid$ and $k$ a power of 2 satisfying $1<k<d$,
it is possible to apply the transformation
\begin{equation}
\sum_{j=1}^\dm \psi_j \ket{j} \ket{0} \mapsto \sum_{j=1}^\dm \psi_j \ket{j} \ket{f(j)}
\end{equation}
using $2\lceil\dm/\chunk\rceil + 4\wid(\chunk-1)$ Toffoli gates, $(\chunk-1)\wid$ dirty ancillae and $\lceil\log( \dm/\chunk)\rceil$ clean ancillae.
\end{theorem}

\begin{proof}
The technique proceeds as follows (see \fig{lookup}).
\begin{itemize}
    \item Allocate a register $r_0$ with $\wid$ clean qubits in the $|+\rangle$ state. This register will ultimately store the output.
    \item Borrow registers $r_1, \ldots, r_{\chunk-1}$ each of size $\wid$ containing dirty ancillae.
    \item Let $l$ be the superposed integer value of the bottom $\log \chunk$ qubits in the address register. Using a series of $\wid (k-1)$ controlled swaps, permute the registers $r_0, ..., r_{\chunk-1}$ such that $r_0$ ends up where $r_l$ was. The other registers can be permuted in any order. Call this swapping procedure $S$.
    \item Let $h$ be the superposed integer value of the top $\lceil\log( \dm/\chunk)\rceil$ qubits of the address register. Perform a table lookup with address $h$ targeting the $r_0, ..., r_{\chunk-1}$ registers. The data for register $r_l$ at address $h$ is equal to the data from the original table at address $h \cdot \chunk + l$. In effect, this is reading many possible outputs at once. Call this lookup process $T$.
    \item Perform the inverse of the swapping procedure $S$.
    \item Because the qubits in $r_0$ were in the $|+\rangle$ state, they were not affected by $T$ (which only targeted them with controlled bit flips).
    Apply Hadamard operations to $r_0$ so that it will be affected by the next $T$, which will also uncompute the dirt XOR'ed into the other registers.
    \item Perform $S$ then $T$ then $S^{-1}$ again.
    \item $r_0$ is now storing the output. The other registers are restored. Return the borrowed $r_1,\dots, r_{\chunk-1}$ registers.
\end{itemize}

The swapping subroutine $S$ has a Toffoli count of $\wid (\chunk - 1)$.
The table lookup subroutine $T$ has a Toffoli count of $\lceil\dm/\chunk\rceil$.
We compute/uncompute $S$ four times and perform $T$ twice.
Therefore the total Toffoli count is $2\lceil\dm/\chunk\rceil +4\wid(\chunk-1)$.
The space cost of the procedure is
$(\chunk-1)\wid$ dirty ancillae for workspace and $\lceil\log( \dm/\chunk)\rceil$ clean ancillae hidden in the implementation of $T$.
\end{proof}

Note that the transformation also uses $\wid$ qubits to store the output.
The value of $\chunk$ that minimizes the Toffoli count is approximately $\sqrt{2\dm/\wid}$.
In practice the number of available dirty qubits often bounds $k$ to be a much smaller value.
The value of $k$ must be greater than 2 in order to have a Toffoli count lower than a standard lookup.

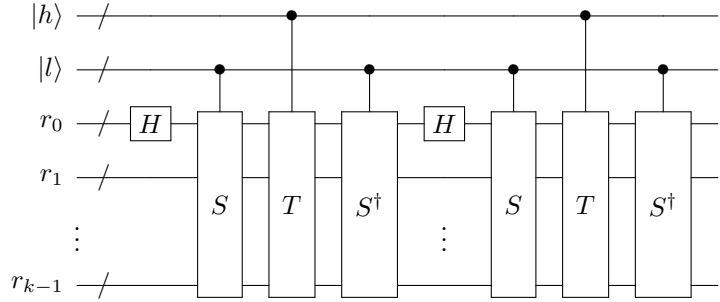
\begin{figure}[h]
\centerline{
\Qcircuit @C=1em @R=0.75em @!R {
\lstick{\ket{h}} & \qw {/} & \qw & \qw & \ctrl{2} & \qw & \qw & \qw& \ctrl{2} & \qw  & \qw \\
\lstick{\ket{l}} & \qw {/} & \qw & \ctrl{1} & \qw & \ctrl{1} &  \qw & \ctrl{1}& \qw  & \ctrl{1} & \qw \\
\lstick{r_0} & \qw {/} & \gate{H} & \multigate{3}{S} & \multigate{3}{T} & \multigate{3}{S^\dagger} &  \gate{H} & \multigate{3}{S} & \multigate{3}{T} & \multigate{3}{S^\dagger} & \qw  \\
\lstick{r_1} & \qw {/} & \qw & \ghost{S} & \ghost{T} & \ghost{S^\dagger} &  \qw & \ghost{S} & \ghost{T} & \ghost{S^\dagger} & \qw \\
\vdots  & & & & & & \vdots \\
\lstick{r_{k-1}} & \qw {/} & \qw & \ghost{S} & \ghost{T} & \ghost{S^\dagger} &  \qw & \ghost{S} & \ghost{T} & \ghost{S^\dagger} & \qw \\
}}\caption{\label{fig:lookup} The sequence of operations used for computing table lookups with dirty ancillae. The output register is $r_0$, and the registers $r_1$ to $r_{k-1}$ are the dirty ancillae.
Note that the scheme in \cite{Lowpreparation} (see Fig.\ 1(d)) uses ``\textsc{Swap}'' for the operation we here call $S$, and ``\textsc{Select}'' for the operation we here call $T$.
The sequence of operations we use here is different than in \cite{Lowpreparation}, resulting in a reduction in the number of registers needed.}
\end{figure}
\FloatBarrier

\section{Cost of computing table lookups assisted by clean ancillae}
\label{app:clean-lookup}

When clean ancillae are available, the optimization from \app{lookup} can be performed more efficiently. The result is as follows.
\begin{theorem}
Given a function $f: \mathbb{Z}_\dm \to \mathbb{Z}_2^\wid$ and $k$ a power of 2 satisfying $1<k<d$,
it is possible to apply the transformation
\begin{equation}
\sum_{j=1}^\dm \psi_j \ket{j} \ket{0} \mapsto \sum_{j=1}^\dm \psi_j \ket{j} \ket{f(j)}
\end{equation}
using $\lceil\dm/\chunk\rceil +\wid(\chunk-1)$ Toffoli gates and $(\chunk-1)\wid+\lceil\log(\dm/\chunk)\rceil$ clean ancillae.
\end{theorem}

\begin{proof}
The steps are as follows.
\begin{itemize}
    \item Allocate registers $r_0, \ldots, r_{\chunk-1}$ each with $\wid$ clean qubits initialized to the $|0\rangle$ state.
    \item Let $h$ be the superposed integer value of the top $\lceil\log(\dm/\chunk)\rceil$ qubits of the address register. Perform a table lookup with address $h$ targeting the $r_0, \dots, r_{\chunk-1}$ registers. The data for register $r_l$ at address $h$ is equal to the data from the original table at address $h \cdot \chunk + l$. In effect, this is reading many possible outputs at once.
    \item Let $l$ be the superposed integer value of the bottom $\log \chunk$ qubits in the address register. Using a series of $\wid\chunk$ controlled swaps, permute the registers $r_0, \dots, r_{\chunk-1}$ such that $r_l$ ends up where $r_0$ was. The other registers can be permuted in any order.
    \item $r_0$ is now storing the output.
    \item Note that every computational basis value of the address register results in a specific computational basis value for registers $r_0, ..., r_{k-1}$ at this point. We have performed the equivalent of table lookup targeting $r_0, ..., r_{k-1}$.
    We can use the uncomputation strategy from \app{unlookup} to uncompute this effective lookup.
    The first thing done by that strategy is to measure all output qubits in the $X$ basis.
    We will not be using any of the qubits from registers $r_1, ..., r_{\chunk -1}$ until that point so they can be measured now instead of later.
    This frees the clean ancillae for other uses.
    Keep the measurement results so they can be used by the uncomputation process.
\end{itemize}

The swapping subroutine of this algorithm has a Toffoli count of $\wid(\chunk - 1)$.
The table lookup subroutine has a Toffoli count of $\lceil\dm/\chunk\rceil$.
We perform each exactly once, therefore the total Toffoli count is $\lceil\dm/\chunk\rceil + \wid(\chunk-1)$.
The space cost of the procedure is
$(\chunk-1)\wid$ clean ancillae for workspace and $\lceil\log(\dm/\chunk)\rceil$ clean ancillae hidden in the implementation of the table lookup subroutine.
\end{proof}

Again the transformation uses $\wid$ qubits to store the output.
The value of $\chunk$ that minimizes the Toffoli count is approximately $\sqrt{\dm/\wid}$.
In practice the number of available clean qubits may bound $k$ to be a much smaller value.
\FloatBarrier

\section{Efficient uncomputation of table lookups using measurement based uncomputation}
\label{app:unlookup}

Next we consider the cost of uncomputing a table lookup, reversing the procedure described in the preceding two appendices.
The result is as follows.
\begin{theorem}
Given a function $f: \mathbb{Z}_\dm \to \mathbb{Z}_2^\wid$ and $k$ a power of 2 satisfying $1<k<d$,
it is possible to apply the transformation
\begin{equation}
\sum_{j=1}^\dm \psi_j \ket{j} \ket{f(j)} \mapsto \sum_{j=1}^\dm \psi_j \ket{j} \ket{0}
\end{equation}
using $\lceil\dm/\chunk\rceil +\chunk$ Toffoli gates and $\chunk+\lceil\log(\dm/\chunk)\rceil$ clean ancillae, or alternatively using $2\lceil\dm/\chunk\rceil +4\chunk$ Toffoli gates, $\chunk-1$ dirty ancillae and $\lceil\log(\dm/\chunk)\rceil+1$ clean ancillae.
\end{theorem}

\begin{proof}
Whenever a quantum circuit uncomputes a qubit $q$ by performing a series of unitary operations that result in the qubit being in the $|0\rangle$ state, the circuit can be optimized by application of the deferred measurement principle.
For example, suppose the last operation involving $q$ is a CNOT targeting $q$.
Then the circuit can be optimized by measuring $q$ in the $X$ basis before the CNOT, then replacing CNOT with a $Z$ gate conditioned on the measurement result.
This is an optimization because, in the surface code, $Z$ gates require no spacetime volume.

This general pattern of ``take an $X$-axis interaction that clears a qubit and use an $X$ basis measurement and a classically-conditioned phase fixup operation  instead" applies to many constructions, including table lookups.
Let's consider the nature of the phase fixup task we must perform when we eagerly measure the output qubits of a table lookup in the $X$ basis.

We can think of each entry in the table as corresponding to a multi-control multi-target CNOT operation.
There is one control (or anti-control) for each address qubit, and one target (or skip) for each output qubit.
Because we will be measuring the output qubits, it is helpful to flip our perspective and think of the output qubits as the controls and the address qubits as the targets.
For example, suppose the entry at address 2 of the table is the bit string $10011000$.
This means that the table lookup must toggle the qubits at offset 0, 3, and 4 of the target register conditioned on the address register storing 2.
Or, equivalently, this means that the table lookup must negate the amplitude of the $|2\rangle$ state of the address register, conditioned on the $X_0 \cdot X_3 \cdot X_4$ observable of the qubits in the target register.
Because we measured the output qubits in the $X$ basis, we know the value of the $X_0 \cdot X_3 \cdot X_4$ observable and thus know whether or not we need to negate the amplitude of the $|2\rangle$ state of the address register.
In fact, using the observables we measured, we can figure out for every state $|j\rangle$ of the address register whether or not $|j\rangle$ needs to be negated.
Performing the necessary negations uncomputes the table lookup.

Let $S$ be the set of all address register states that need to have their amplitude negated.
This will be some arbitrary subset of the states from $|0\rangle$ to $|\dm-1\rangle$.
We now transform the task of negating all the states in $S$ into the task of performing a table lookup of size $\dm/2$ with output size $2$.
Let $q$ be the least significant qubit of the address register, and $u$ be a clean ancilla qubit in the $|1\rangle$ state.
Apply a CNOT from $q$ onto $u$, then apply a Hadamard transform to each.
Now define a ``fixup table" $F$ with entries $F_j$ each specifying two output bits defined by $S$:

\begin{equation}
F_j =
\begin{cases}
    00 & (2j \notin S) \land (2j+1 \notin S) \\
    01 & (2j \in S) \land (2j+1 \notin S) \\
    10 & (2j \notin S) \land (2j+1 \in S) \\
    11 & (2j \in S) \land (2j+1 \in S) \\
\end{cases}
\end{equation}

After preparing $q$ and $u$, perform a table lookup from $F$ onto $q$ and $u$.
The address of the lookup is all of the qubits of the address register, except for the least significant qubit.
Because of how we defined $F$ and how we prepared $q$ and $u$, this negates the phase of all states from $S$.
This uncomputes the table lookup, and we finish the uncomputation by uncomputing the preparation of $q$ and $u$.
See \fig{unlookup} for a quantum circuit showing an overview of the process.

This technique implements the uncomputation of a table lookup with address size $\dm$ and output size $\wid$ in terms of computing a table lookup with address size $\dm/2$ and output size $2$.
Reducing the address and output size increase the effectiveness of the techniques explained in \app{clean-lookup} and \app{lookup}, since a larger value of $k$ can be used.

By combining this technique and \app{clean-lookup}, if $\chunk$ additional clean ancillae are available, the Toffoli cost of uncomputing a table lookup with address size $\dm$ and output size $\wid$ can be reduced to $\lceil \dm/\chunk\rceil + \chunk$.
The procedure to use is shown in \fig{unlookupc}, where $r_0$ is initially $\ket{0}$, which is flipped to $\ket{1}$ by the $X$ gate.
The controlled swap $S$ shifts that $\ket{1}$ to position $l$, giving a one-hot unary encoding of the value in the bottom $\log k$ qubits in the address register.
The Hadamards and $T$ then yield the correct phase factor.
The reverse controlled swap then erases the unary encoding.
The first controlled swap has cost $k-1$, with the cost reduced as compared to the table lookup because the outputs are single qubits instead of $\wid$-qubit registers.
The $T$ has cost $\lceil \dm/\chunk\rceil$, and the final controlled swap may be performed with Clifford gates.
The reason is that a unary encoding may be erased using measurements and Clifford gates.
There are $k$ ancillae used for the unary encoding, and as before there are $\lceil \log(\dm/\chunk)\rceil$ ancillae needed for the implementation of $T$.

Combining this technique with \app{lookup} instead, if $\chunk$ additional dirty ancillae are available, the Toffoli cost of uncomputing a table lookup with address size $\dm$ and output size $\wid$ can be reduced to $2\lceil \dm/\chunk\rceil + 4\chunk$.
The procedure to use is shown in \fig{unlookupd}, which is a slight modification over that for the table lookup.
As before $r_0$ is a clean register, and $r_1$ to $r_{k-1}$ can be dirty registers, though now they need just be qubits.
The first $S$, $T$, and $S^\dagger$ have no effect on the target qubit $r_0$.
Then the $Z$ gate changes the state of this qubit to $(\ket{0}-\ket{1})/\sqrt 2$, so the next $T$ yields the correct phase factor.
At the end the $Z$ and $H$ return the state of $r_0$ to $\ket{0}$.
There are $k-1$ dirty ancillae used, one clean ancilla for $r_0$, and $\lceil \log(\dm/\chunk)\rceil$ clean ancillae needed for the implementation of $T$.

In the case of clean ancillae, the Toffoli cost is minimized for $\chunk\approx \sqrt{\dm}$, where the Toffoli count is $2 \sqrt{\dm}$.
For dirty ancillae, the cost is minimized at $\chunk\approx\sqrt{\dm/2}$, where the Toffoli count is $\sqrt{32 \dm}$.\end{proof}

To explain in more detail how to erase the unary register using Clifford gates, see \fig{qrom} to \fig{qromclif}.
\fig{qrom} shows how to map binary to unary using controlled swaps.
The sequence of controlled swaps is the same as in \cite{Lowpreparation}.
If we first flip the top qubit of the unary register to 1, the controlled swap network moves the 1 to the correct position.
To uncompute the unary register, we do this process in reverse. We expand each controlled swap into a CNOT-Toffoli-CNOT construction, as shown in \fig{qrominv}.
The construction is considerably simplified because for the second CNOT in each group the control is known to be in the $\ket{0}$ state, so the CNOT can be omitted entirely.
Furthermore, when the output of a Toffoli is known to be zero we can use the uncomputation trick instead, leaving us with a circuit using no non-Clifford gates, as shown in \fig{qromclif}.

\newcommand{\cgate}[1]{*+<.6em>{#1} \POS ="i","i"+UR;"i"+UL **\dir{-};"i"+DL **\dir{-};"i"+DR **\dir{-};"i"+UR **\dir{-},"i" \cw}
\newcommand{\igate}[1]{*+<.6em>{#1} \POS ="i","i"+UR;"i"+UL **\dir{-};"i"+DL **\dir{-};"i"+DR **\dir{-};"i"+UR **\dir{-},"i"}
\newcommand{\imultigate}[2]{*+<1em,.9em>{\hphantom{#2}} \POS [0,0]="i",[0,0].[#1,0]="e",!C *{#2},"e"+UR;"e"+UL **\dir{-};"e"+DL **\dir{-};"e"+DR **\dir{-};"e"+UR **\dir{-},"i"}
\newcommand{\ighost}[1]{*+<1em,.9em>{\hphantom{#1}}}
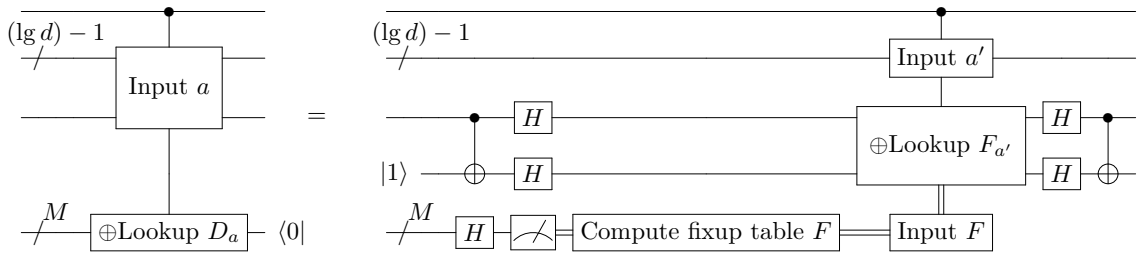
\begin{figure}[h]
\resizebox{\textwidth}{!}{
\Qcircuit @R=1em @C=0.7em {
  &\qw    &                  \qw&\qw &                               \ctrl{1}      &\qw                       &&& &&&&&\qw    &                  \qw&\qw &     \qw&     \qw&                                     \qw&                    \ctrl{1}                       &     \qw&     \qw&\qw&\\
  &\qw {/}&\ustick{(\lg \dm)-1}\qw&\qw &       \multigate{1}{\text{Input }a}         &\qw                       &&& &&&&&\qw {/}&\ustick{(\lg \dm)-1}\qw&\qw &     \qw&     \qw&                                     \qw&                   \gate{\text{Input }a^\prime}    &     \qw&     \qw&\qw&\\
  &\qw    &                  \qw&\qw &              \ghost{\text{Input }a}         &\qw                       &&&=&&&&&\qw    &                  \qw&\qw &\ctrl{1}&\gate{H}&                                     \qw&\multigate{1}{\oplus\text{Lookup }F_{a^\prime}}\qwx&\gate{H}&\ctrl{1}&\qw&\\
  &       &                     &    &                                         \qwx&                          &&& &&&&&       &\lstick{|1\rangle}   &\qw &   \targ&\gate{H}&                                     \qw&       \ghost{\oplus\text{Lookup }F_{a^\prime}}    &\gate{H}&   \targ&\qw&\\
  &\qw {/}&\ustick{\wid}        \qw&\qw &        \gate{\oplus \text{Lookup }D_{a}}\qwx& \rstick{\langle0|} \qw   &&& &&&&&\qw {/}&\ustick{\wid}        \qw&\qw &\gate{H}&\meter  &\cgate{\text{Compute fixup table }F}    &                         \cgate{\text{Input }F}\cwx&        &        &   &\\
}
}
    \caption{
        Uncomputing a table lookup using eager measurement and phase fixups.
        Reduces the effective address size from $\dm$ to $\dm/2$ and the effective output size from $\wid$ to $2$.
    }
        \label{fig:unlookup}
\end{figure}

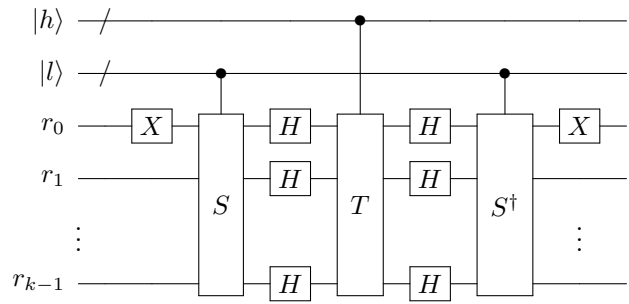
\begin{figure}[htb]
\centerline{
\Qcircuit @C=1em @R=0.7em @!R {
\lstick{\ket{h}} & \qw {/} & \qw & \qw & \qw & \ctrl{2} & \qw & \qw & \qw & \qw \\
\lstick{\ket{l}} & \qw {/} & \qw & \ctrl{1} & \qw & \qw & \qw & \ctrl{1} &  \qw & \qw \\
\lstick{r_0} & \qw & \gate{X}  & \multigate{3}{S} & \gate{H} & \multigate{3}{T} & \gate{H} & \multigate{3}{S^\dagger} & \gate{X} & \qw  \\
\lstick{r_1} & \qw & \qw & \ghost{S} & \gate{H} & \ghost{T} & \gate{H} & \ghost{S^\dagger} &  \qw & \qw \\
\vdots  & & & & &  & & & \vdots \\
\lstick{r_{k-1}} & \qw & \qw & \ghost{S} & \gate{H} & \ghost{T} & \gate{H} & \ghost{S^\dagger} &  \qw & \qw \\
}}\caption{\label{fig:unlookupc} The sequence of operations used for applying the phase fixup with clean ancillae. The registers $r_0$ to $r_{k-1}$ are single qubits, initially set to $\ket{0}$.}
\end{figure}

\begin{figure}[htb]
\centerline{
\Qcircuit @C=1em @R=0.7em @!R {
\lstick{\ket{h}} & \qw {/} & \qw & \qw & \ctrl{2} & \qw & \qw & \qw& \ctrl{2} & \qw  & \qw & \qw & \qw \\
\lstick{\ket{l}} & \qw {/} & \qw & \ctrl{1} & \qw & \ctrl{1} &  \qw & \ctrl{1}& \qw  & \ctrl{1} & \qw & \qw & \qw \\
\lstick{r_0} & \qw & \gate{H} & \multigate{3}{S} & \multigate{3}{T} & \multigate{3}{S^\dagger} &  \gate{Z} & \multigate{3}{S} & \multigate{3}{T} & \multigate{3}{S^\dagger} &  \gate{Z} &  \gate{H} & \qw  \\
\lstick{r_1} & \qw & \qw & \ghost{S} & \ghost{T} & \ghost{S^\dagger} &  \qw & \ghost{S} & \ghost{T} & \ghost{S^\dagger} & \qw & \qw & \qw \\
\vdots  & & & & & & \vdots \\
\lstick{r_{k-1}} & \qw & \qw & \ghost{S} & \ghost{T} & \ghost{S^\dagger} &  \qw & \ghost{S} & \ghost{T} & \ghost{S^\dagger} & \qw & \qw & \qw \\
}}\caption{\label{fig:unlookupd} The sequence of operations used for applying the phase fixup with dirty ancillae. The register $r_0$ is initially set to $\ket{0}$, whereas $r_1$ to $r_{k-1}$ are dirty qubits.}
\end{figure}
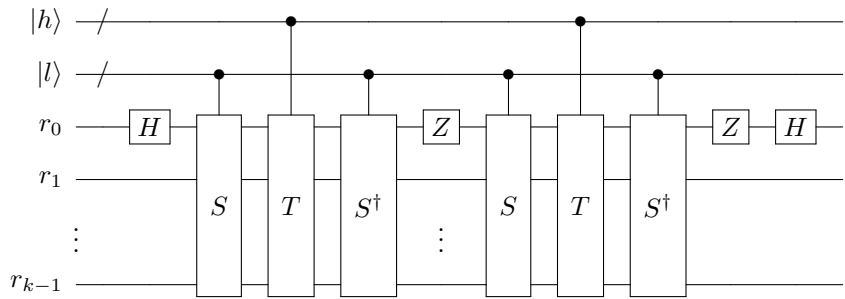

\begin{figure}[htb]
\centerline{
\Qcircuit @C=1em @R=.1em @!R {
& \qw & \ctrl{4} & \qw & \qw & \qw & \qw& \qw & \qw & \qw \\
& \qw & \qw & \ctrl{4} & \ctrl{5} &  \qw & \qw& \qw & \qw & \qw  \\
& \qw & \qw & \qw & \qw &  \ctrl{5} & \ctrl{6} & \ctrl{7} & \ctrl{8} & \qw   \\
& \gate{X} & \qswap & \qswap & \qw & \qswap & \qw & \qw & \qw & \qw  \\
& \qw & \qswap & \qw & \qswap &\qw & \qswap &  \qw & \qw & \qw  \\
& \qw & \qw & \qswap & \qw & \qw & \qw & \qswap & \qw & \qw \\
& \qw & \qw & \qw & \qswap & \qw & \qw &\qw & \qswap & \qw \\
& \qw & \qw & \qw & \qw & \qswap & \qw & \qw & \qw & \qw  \\
& \qw & \qw & \qw & \qw & \qw & \qswap & \qw & \qw & \qw  \\
& \qw & \qw & \qw & \qw & \qw & \qw & \qswap & \qw & \qw  \\
& \qw & \qw & \qw & \qw & \qw & \qw & \qw & \qswap & \qw 
}}\caption{\label{fig:qrom} This circuit maps a binary register in the top three qubits to a one-hot unary register in the bottom 8 qubits.}
\end{figure}
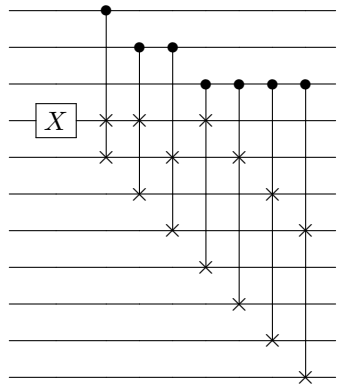

\FloatBarrier

\begin{figure}[htb]
\centerline{
\Qcircuit @C=1em @R=.4em @!R {
& \qw & \qw & \qw & \qw & \qw & \qw & \qw & \qw & \qw & \qw & \qw & \qw & \qw & \qw & \qw & \qw & \qw & \qw & \qw & \ctrl{3} & \qw & \qw & \qw \\
& \qw & \qw & \qw & \qw &  \qw & \qw & \qw & \qw & \qw & \qw & \qw & \qw & \qw & \ctrl{3} & \qw & \qw & \ctrl{2} & \qw & \qw & \qw & \qw & \qw & \qw  \\
& \qw & \ctrl{4} & \qw & \qw & \ctrl{3} & \qw & \qw & \ctrl{2} & \qw & \qw & \ctrl{2} & \qw & \qw & \qw & \qw & \qw & \qw & \qw & \qw & \qw & \qw & \qw & \qw  \\
& \qw & \qw & \qw & \qw & \qw & \qw & \qw & \qw & \qw & \targ & \ctrl{4} & \targ & \qw & \qw & \qw & \targ & \ctrl{2} & \targ & \targ & \ctrl{1} & \targ & \targ & \qw \\
& \qw & \qw & \qw & \qw & \qw & \qw & \targ & \ctrl{4} & \targ & \qw & \qw & \qw & \targ & \ctrl{2} & \targ & \qw & \qw & \qw & \ctrl{-1} & \targ & \ctrl{-1} & \qw & \qw \\
& \qw & \qw & \qw & \targ & \ctrl{4} & \targ & \qw & \qw & \qw & \qw & \qw & \qw & \qw & \qw & \qw & \ctrl{-2} & \targ & \ctrl{-2} & \qw & \qw & \qw & \qw & \qw \\
& \targ & \ctrl{4} & \targ & \qw & \qw & \qw & \qw & \qw & \qw & \qw & \qw & \qw & \ctrl{-2} & \targ & \ctrl{-2} & \qw & \qw & \qw & \qw & \qw & \qw & \qw & \qw \\
& \qw & \qw & \qw & \qw & \qw & \qw & \qw & \qw & \qw & \ctrl{-4} & \targ & \ctrl{-4} & \qw & \qw & \qw & \qw & \qw & \qw & \qw & \qw & \qw & \qw & \qw \\
& \qw & \qw & \qw & \qw & \qw & \qw & \ctrl{-4} & \targ & \ctrl{-4} & \qw & \qw & \qw & \qw & \qw & \qw & \qw & \qw & \qw & \qw & \qw & \qw & \qw & \qw \\
& \qw & \qw & \qw & \ctrl{-4} & \targ & \ctrl{-4} & \qw & \qw & \qw & \qw & \qw & \qw & \qw & \qw & \qw & \qw & \qw & \qw & \qw & \qw & \qw & \qw & \qw \\
& \ctrl{-4} & \targ & \ctrl{-4} & \qw & \qw & \qw & \qw & \qw & \qw & \qw & \qw & \qw & \qw & \qw & \qw & \qw & \qw & \qw & \qw & \qw & \qw & \qw & \qw
}}\caption{\label{fig:qrominv} The inverse of the circuit for mapping binary to unary, with the controlled swaps implemented using CNOTs and Toffolis.}
\end{figure}
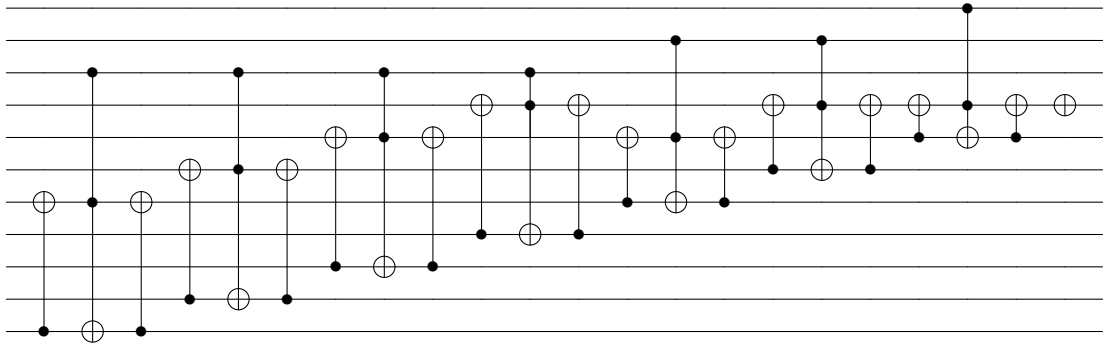

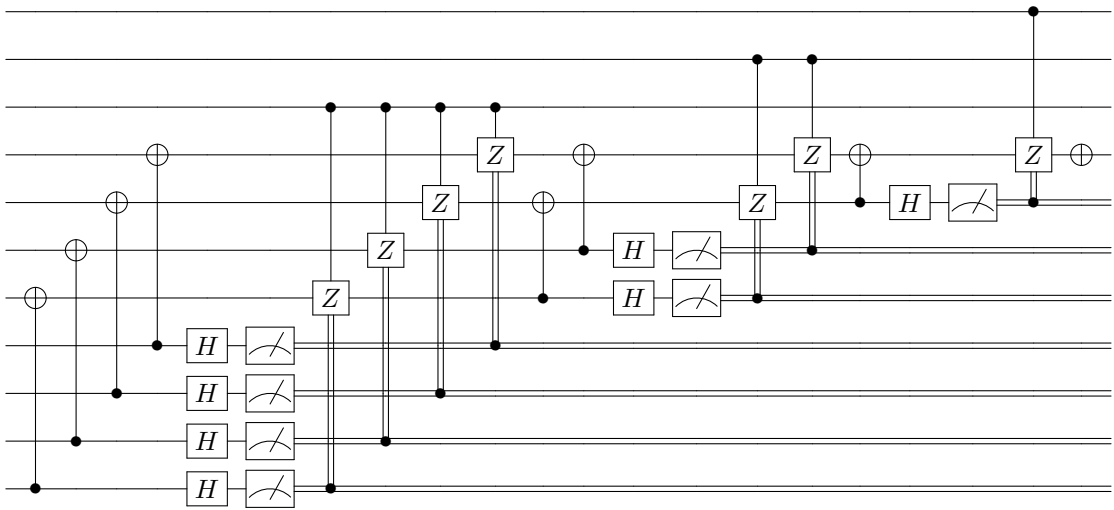
\begin{figure}[htb]
\centerline{
\Qcircuit @C=0.7em @R=.4em @!R {
& \qw & \qw & \qw & \qw & \qw & \qw & \qw & \qw & \qw & \qw & \qw & \qw & \qw & \qw & \qw & \qw & \qw & \qw & \qw & \ctrl{3} & \qw & \qw  \\
& \qw & \qw & \qw & \qw &  \qw & \qw & \qw & \qw & \qw & \qw & \qw & \qw & \qw & \qw & \ctrl{3} & \ctrl{2} & \qw & \qw & \qw & \qw & \qw & \qw    \\
& \qw & \qw & \qw & \qw & \qw & \qw & \ctrl{4} & \ctrl{3} & \ctrl{2} & \ctrl{1} & \qw & \qw & \qw & \qw & \qw & \qw & \qw & \qw & \qw & \qw & \qw & \qw   \\
& \qw & \qw & \qw & \targ & \qw & \qw & \qw & \qw & \qw & \gate{Z} & \qw & \targ & \qw & \qw & \qw & \gate{Z} & \targ & \qw & \qw & \gate{Z} & \targ & \qw \\
& \qw & \qw & \targ & \qw & \qw & \qw & \qw & \qw & \gate{Z} & \qw \cwx & \targ & \qw & \qw & \qw & \gate{Z} & \qw \cwx & \ctrl{-1} & \gate{H} & \meter & \control \cw \cwx & \cw & \cw \\
& \qw & \targ & \qw & \qw & \qw & \qw & \qw & \gate{Z} & \qw \cwx & \qw \cwx & \qw & \ctrl{-2} & \gate{H} & \meter &  \cw \cwx & \control \cw \cwx & \cw & \cw & \cw & \cw & \cw & \cw  \\
& \targ & \qw & \qw & \qw & \qw & \qw & \gate{Z} & \qw \cwx & \qw \cwx & \qw \cwx & \ctrl{-2} & \qw & \gate{H} & \meter & \control \cw \cwx & \cw & \cw & \cw & \cw & \cw & \cw & \cw \\
& \qw & \qw & \qw & \ctrl{-4} & \gate{H} & \meter & \cw \cwx & \cw \cwx & \cw \cwx & \control \cw \cwx & \cw & \cw & \cw & \cw & \cw & \cw & \cw & \cw & \cw & \cw & \cw & \cw \\
& \qw & \qw & \ctrl{-4} & \qw & \gate{H} & \meter & \cw \cwx & \cw \cwx & \control \cw \cwx & \cw & \cw & \cw & \cw & \cw & \cw & \cw & \cw & \cw & \cw & \cw & \cw & \cw \\
& \qw & \ctrl{-4} & \qw & \qw & \gate{H} & \meter & \cw \cwx & \control \cw \cwx & \cw & \cw & \cw & \cw & \cw & \cw & \cw & \cw & \cw & \cw & \cw & \cw & \cw & \cw  \\
& \ctrl{-4} & \qw & \qw & \qw & \gate{H} & \meter & \control \cw \cwx & \cw & \cw & \cw & \cw & \cw & \cw & \cw & \cw & \cw & \cw & \cw & \cw & \cw & \cw & \cw
}}\caption{\label{fig:qromclif} The circuit in \fig{qrominv} with the Toffolis replaced with measurements and controlled operations, and the second CNOT in each group omitted.}
\end{figure}

\section{The scaling $\lambda$ in general contexts}
\label{app:lambda_general}

Here we briefly discuss how we expect $\lambda$ might scale for more general systems, beyond the FeMoco system considered in the rest of this work. In the plane wave basis, one can obtain a clean bound of $\lambda = {\cal O}(N^2)$ \cite{BabbushLow}. One might roughly expect similar scaling in a more general context, based on intuition about electrons interacting pairwise which would imply the spectral norm of the Hamiltonian should go roughly quadratically in $N$. However, for arbitrary basis sets the relationships between $\lambda$, basis size, basis type, molecular structure, etc.\ are very complex and difficult to rigorously bound. Here we use numerics to provide insight into how $\lambda$ scales.

In \figx{hydrogens}{a} we show how $\lambda$ scales for a chain of Hydrogen atoms when basis resolution is fixed and the system grows towards the thermodynamic (large system size) limit (e.g.~a chain with many Hydrogen atoms). We focus on an atomic chain because chains approach their thermodynamic limit much faster than other configurations of atoms. We see in all of our numerics that $\lambda_T \leq \lambda_V \leq \lambda_W$. For the hydrogen chains we see that $\lambda_V = {\cal O}(N^{2.2})$ and $\lambda_W = {\cal O}(N^{2.5})$. It is interesting that $\lambda_W$ has slightly asymptotically worse scaling than $\lambda_V$; however, the difference is far less than the factor of $N$ that we save by performing the low rank truncation (resulting in us scaling like $\lambda_W$ instead of like $\lambda_V$).

Slightly different behavior occurs when we fix the system (in this case the ${\rm H}_4$ system) and grow towards the continuum limit (the limit of an arbitrarily large basis), shown in \figx{hydrogens}{b}. Here, we see that $\lambda_V = {\cal O}(N^{2.7})$ and $\lambda_W = {\cal O}(N^3)$. Overall the scaling of $\lambda$ is a bit worse but $\lambda_W / \lambda_V$ still scales as roughly ${\cal O}(N^{0.3})$.

In both of these numerical experiments we find that $\lambda$ is scaling worse than $\Omega(N^{1.5})$, which is the condition we require for our $\widetilde{\cal O}(N^{3/2}\lambda)$ scaling to be better than the $\widetilde{\cal O}(\lambda^2)$ scaling of \cite{Campbell2019}. Still, it is clear that more research is required to fully understand the relationship between $\lambda$ and molecular structure, and how the size of $\lambda$ might be reduced.

\begin{figure}[h]
\begin{minipage}[t]{.47\textwidth}
\centering
\includegraphics[width=\linewidth]{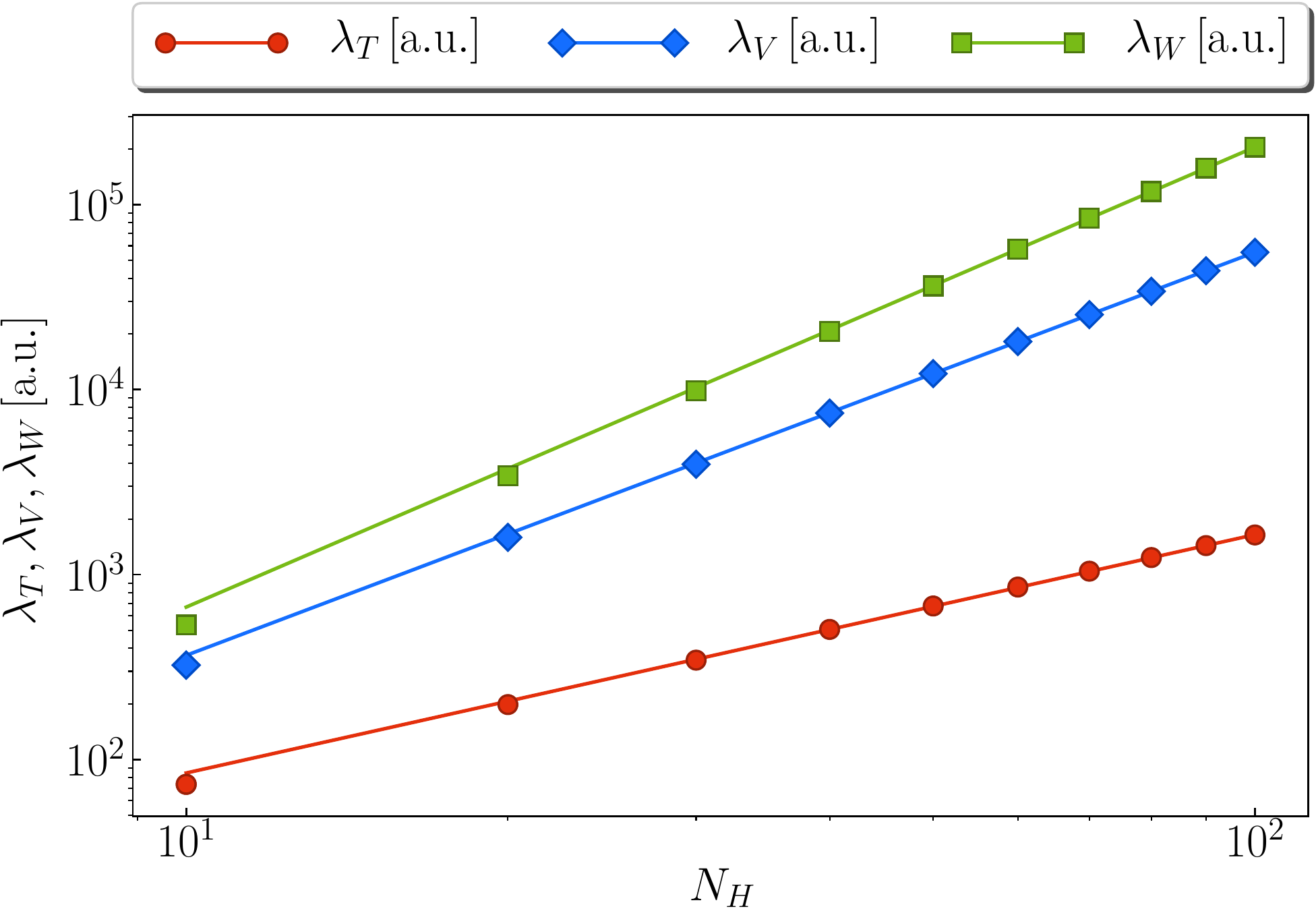}
{\phantom{-}~{\small (a)~$\lambda$ values as function of Hydrogen chain size.}}
\label{fig:h_chain}
\end{minipage}\hspace{8mm}
\begin{minipage}[t]{.47\textwidth}
\centering
\includegraphics[width=\linewidth]{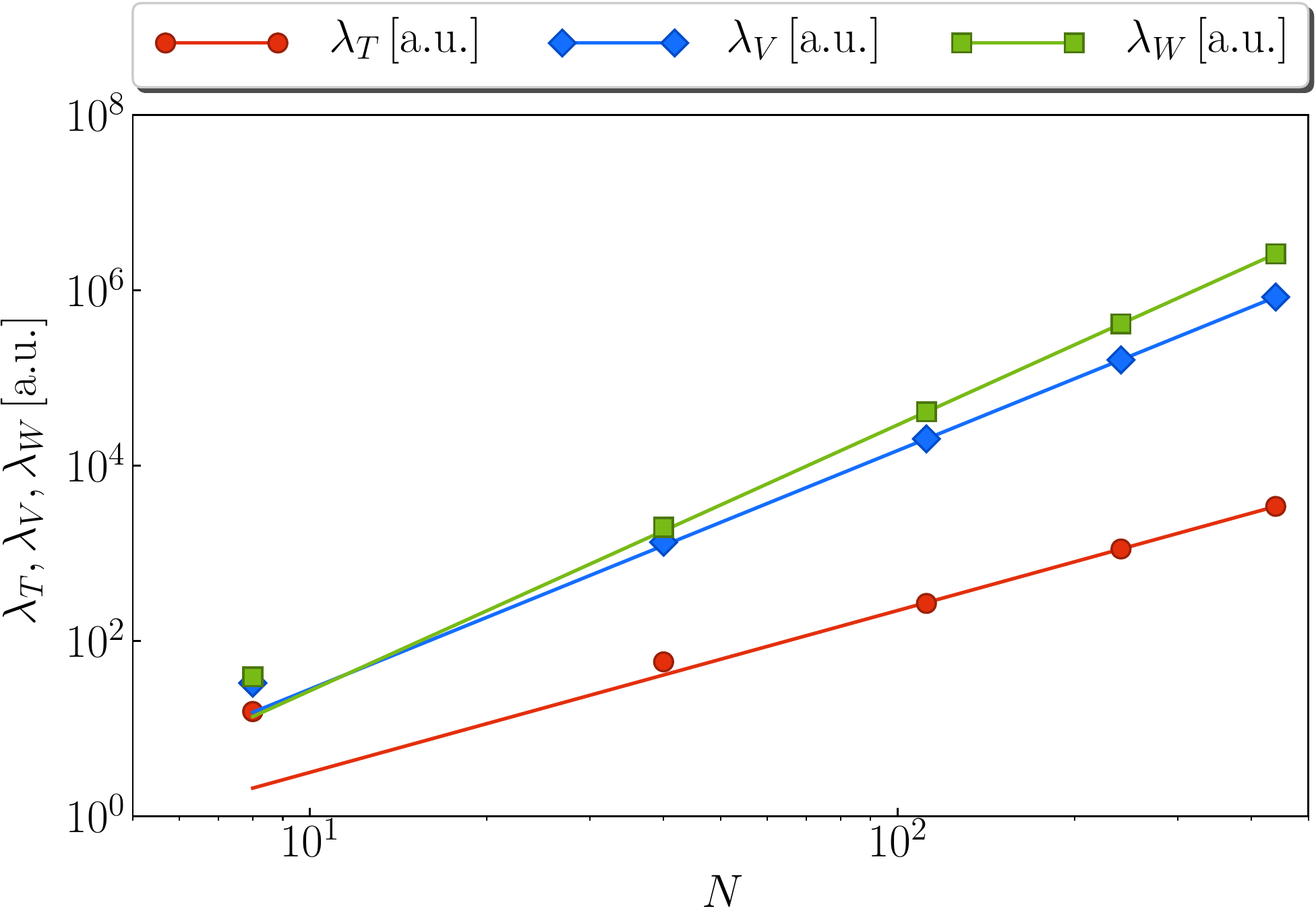}
{\phantom{-}~{\small (b)~$\lambda$ values as function of basis size for ${\rm H}_4$.}}
\label{fig:h4}
\end{minipage}
\caption{\label{fig:hydrogens} The scaling of $\lambda_T$, $\lambda_V$ and $\lambda_W$ for simple systems growing towards continuum and thermodynamic limits. In (a) we see a linear chain of Hydrogen atoms represented in the STO-6G basis, as a function of the number of Hydrogens, $N_H$, in the chain. The atoms have a spacing of 1.4 Bohr radii. For this system the number of spin-orbitals, $N$, is twice the number of Hydrogens, $N_H$. The linear regressions suggest $\lambda_T = {\cal O}(N_H^{1.3})$, $\lambda_V = {\cal O}(N_H^{2.2})$ and $\lambda_W = {\cal O}(N_H^{2.5})$. In (b) we see a plaquette of four Hydrogen atoms (${\rm H}_4$) with a spacing of 2 Bohr radii on each side, as function of the number of spin-orbitals (qubits, $N$) used to represent the system. The points here correspond to progressively larger basis sets between STO-6G and cc-pV5Z (the first point seems to have a slightly different trend, likely because STO-6G uses very different primitive functions than the correlation consistent basis sets). The linear regressions suggest $\lambda_T = {\cal O}(N^{1.9})$, $\lambda_V = {\cal O}(N^{2.7})$ and $\lambda_W = {\cal O}(N^{3.0})$.}
\end{figure}
\FloatBarrier

\section{Detailed costings}\label{app:mincost}
Here we give the details for the minor Toffoli costs and the numbers of logical qubits used.
\subsection{{\micro} orbitals}
For the case where the low rank factorization approach is used, and the small number of dirty qubits is used, the minor costs are as follows.
\begin{enumerate}
    \item $4(N+\lceil \log N\rceil)=460$ Toffolis for the controlled operations.
    \item For the initial equal superposition state preparation, we have 8 qubits for $\ell$, and 6 qubits for each of the $p,q,r,s$ registers for a total of 32.
    To obtain a final amplitude close to one it is convenient to also prepare an ancilla in an equal superposition of 15 out of 16 basis states (on 4 qubits).
    Then the number of qubits is increased to 36, and two steps of amplitude amplification takes 454 Toffolis (see \app{prep}). The final amplitude is close enough to one that it does not affect the complexity at the precision we are giving.
    \item The inequality tests and controlled swaps for the state preparation need to be done twice for each of the two preparations, because there is a cost to prepare and a cost to unprepare. We are preparing 21 qubits in the first step, and 13 qubits in the second step.
    The total cost is therefore $2(2\mu+21+13)=176$.
    \item The controlled swaps used for the symmetries have a cost of $4\lceil \log (N/2)\rceil=24$ Toffolis, taking into account preparation and inverse preparation.
    \item A careful accounting of the number of Toffolis needed to evaluate $s$ in \eq{qromreg} gives 105 (see \app{arithmetic}). We incur this cost four times, because we have two state preparations and two inverse preparations.  The total is therefore $4\times 105=420$.
\end{enumerate}
The total of these costs is 1,534.

The numbers of qubits used are as follows.
\begin{enumerate}
\item The system is represented on $N=108$ qubits.
\item We are preparing a state with a number of qubits
$\lceil \log L\rceil +6+4\lceil\log(N/2)\rceil=38$.
In this expression the first term takes account of $\ell$, the $+6$ takes account of $q_1,q_2,\theta_{pq}^{(\ell)},\theta_{rs}^{(\ell)},\alpha$, and $\beta$, and the last term takes account of the $p,q,r,s$ registers.
\item There are 4 ancillae used in preparing the equal superposition state, plus a flag qubit for success.
\item Computing \eq{qromreg} (for the register to iterate over for the QROAM) twice needs $19\times 2=38$ qubits for the output.
\item There are $49+41=90$ qubits used as output in the two steps of the state preparation.
\item The QROAM uses $\lceil\log(d/\chunk)\rceil=17$ clean qubits for the preparation.
A smaller number of clean qubits are used for the inverse preparation, and we may reuse the same qubits.

\item Each state preparation needs $\mu=27$ qubits to store a superposition state to perform the inequality comparison, plus a flag qubit for success.
\item The number of qubits needed for the phase estimation $m = 26$.
\end{enumerate}
Adding all these qubit counts together gives 378.

Next we consider the case where the large number of ancilla qubits is used.
The minor costs are the same as before, except now there are three inequality tests and controlled swaps needed for the state preparation, and $\mu$ is increased to $28$.
In preparing the superposition over $\ell$, the size of the registers acted upon by the controlled swap is $\lceil \log L\rceil=8$.
In preparing the superposition over $p$ and $q$, there are two registers of size $\lceil\log(N/2)\rceil=6$, as well as the qubit register storing $\theta_{pq}^{(\ell)}$, for a total of 13.
The number of qubits in preparing the superposition over $r$ and $s$ is the same.
Therefore the Toffoli cost of the inequality tests and controlled swaps is increased from 176 to $6\mu+2(8+13+13)=236$.
The minor costs therefore are increased to 1,594.

Evaluating the number of qubits used is similar to the case for the dirty ancillae considered above.
\begin{enumerate}
\item The system is represented on $N=108$ qubits.
\item We are preparing a state with a number of qubits
$\lceil \log L\rceil +6+4\lceil\log(N/2)\rceil=38$.
\item There are 4 ancillae used in preparing the equal superposition state, plus a flag qubit for success.
\item Computing \eq{qromreg} twice needs $19\times 2=38$ qubits for the output.
\item There are $8+42+42=92$ qubits used as output in the three steps of the state preparation.
\item The QROAM uses $(\chunk-1)\wid+\lceil\log(d/\chunk)\rceil= 2{,}659$ qubits.
These are erased, so we need only count the cost of those used in step 4 (since those may reuse the qubits used for step 3).
\item Each state preparation needs $\mu=28$ qubits to store a superposition state to perform the inequality comparison, as well as a single qubit to store the result of the inequality comparison.
We need not count the cost of the qubits used in step 4, because we may reuse qubits from the QROAM.
\item The number of qubits needed for the phase estimation $m = 26$.
\end{enumerate}
Altogether there are 3,024 qubits used.

Third for the {\micro} orbitals, we consider the costs for the sparse state preparation approach.
In this case the minor costs are as follows.
\begin{enumerate}
    \item $4(N+\lceil \log N\rceil)=460$ Toffolis for the controlled operations.
    \item The preparation of the equal superposition state is just for a single register of 19 bits. It is more efficient to use an 6-qubit ancilla and prepare a superposition over 19 basis states in that ancilla.
Then the amplitude amplification takes one step, and is on 25 bits so has a cost of 160.
    \item In this case there is a \emph{single} state preparation for 26 qubits, so the cost for the preparation plus inverse preparation is $2(\mu + 26) = 102$.
    \item This time there is a single state preparation, but three controlled swaps used to generate the symmetries.
    Two of these controlled swaps are on registers of size $\lceil\log (N/2)\rceil$, and one is on a register of twice the size, so the cost is the same as before $4\lceil \log (N/2)\rceil=24$.
\end{enumerate}
This time there is no additional cost from arithmetic, because we are not using \eq{qromreg}.
The total of the minor costs is therefore 746.

The numbers of qubits used are as follows.
\begin{enumerate}
\item The $N=108$ system registers.
\item The state being prepared has size $7+4\lceil\log(N/2)\rceil=31$ qubits.
Here the 7 includes the qubit distinguishing between $T$ and $V$, the three qubits used for the symmetries, and $\theta$, $\alpha$, and $\beta$.
\item Preparation of the equal superposition state uses 6 ancilla qubits, plus there is a flag qubit for success.
\item The 19 qubits that we iterate over for the QROAM.
\item The qubits needed for the QROAM $\chunk_1\wid$, which includes the output $\wid$ registers and another $(\chunk_1-1)\wid$ working output registers.
We subtract $2+4\lceil\log(N/2)\rceil=26$ qubits that are in the main output register that we have already accounted for.
Here $\chunk_1=64$ and $M=77$ gives 4,902.
\item The QROAM uses another $\lceil\log(d/\chunk_1)\rceil=12$ clean qubits.
\item Each state preparation needs $\mu=25$ qubits to store a superposition state to perform the inequality comparison, as well as a flag qubit.
These do not add to the cost because we can reuse qubits used in the QROAM.
\item The number of qubits needed for the phase estimation $m = 24$.
\end{enumerate}
The total number of logical qubits needed is then 5,103.

\subsection{{\caltech} orbitals}
For the {\caltech} orbitals using the small number of dirty ancillae,
the minor costs are as follows.
\begin{enumerate}
    \item $4(N+\lceil \log N\rceil)=640$ Toffolis for the controlled operations.
    \item This time there are 36 qubits, because the $p,q,r,s$ registers need 7 qubits each.
    In this case it is more efficient to perform amplitude amplification separately to prepare the equal superpositions over $\ell,p,q$ and $r,s$ separately.
    Preparing an ancilla with 11 out of 16 basis states (4 qubits) in the first preparation and 17 out of 32 basis states (5 qubits) in the second, we can achieve amplitude close to one with two steps of amplitude amplification. The total Toffoli cost is 534 (see \app{prep}).
    \item We are preparing 23 qubits in the first step, and 15 qubits in the second step.
    The total cost for the inequality tests and controlled swaps is $2(2\mu+23+15)=184$.
    \item The controlled swaps used for the symmetries have a cost of $4\lceil \log (N/2)\rceil=28$ Toffolis.
    \item A careful accounting of the cost of evaluating \eq{qromreg} gives 108 (see \app{arithmetic}), and this is incurred four times for a total of $4\times 108=432$.
\end{enumerate}
The total of these costs is 1,818.

This time the qubits used are as follows.
\begin{enumerate}
\item The system is on $N=152$ qubits.
\item The number of qubits in the prepared state is $\lceil\log L\rceil+6+4\lceil\log(N/2)\rceil=42$.
\item There are 9 qubits used for preparing the equal superposition state, plus one to flag success.
\item Computing \eq{qromreg} twice needs $20\times 2=40$ qubits for the output.
\item The two QROMs output $51+43=94$ qubits.
\item The QROAM uses $\lceil\log(d/\chunk)\rceil=18$ clean qubits for the preparation.
A smaller number of clean qubits are used for the inverse preparation, and we may reuse the same qubits as before.
\item The two size $\mu=27$ registers for the inequality tests, plus the two qubits flagging the results.
\item The number of qubits for the phase estimation is $m=25$.
\end{enumerate}
The total number of logical qubits is 437.

Next we consider the {\caltech} orbitals with a large number of clean ancillae.
This time the minor costs are the same, except the cost of the inequality tests and controlled swaps is increased to $6\mu+2(8+15+15)=238$ (instead of 184).
Here the number of qubits that must be swapped in preparing the superposition over $\ell$ is $\lceil \log L \rceil=8$.
The number of qubits for preparing the superposition over $p$ and $q$ is 15, because $p$ and $q$ have registers of $\lceil\log(N/2)\rceil=7$ qubits each, and there is the qubit storing $\theta_{pq}^{(\ell)}$.
The preparation over $r$ and $s$ also needs swaps with 15 qubits.
Therefore the minor costs are increased to 1,872.

The numbers of logical qubits used are as follows.
\begin{enumerate}
\item The system is represented on $N=152$ qubits.
\item We are preparing a state with a number of qubits
$\lceil \log L\rceil +6+4\lceil\log(N/2)\rceil=42$.
\item There are 9 ancillae used in preparing the equal superposition state, plus a flag qubit for success.
\item Computing \eq{qromreg} twice needs $20\times 2=40$ qubits for the output.
\item There are $8+43+43=94$ qubits used as output in the three steps of the state preparation.
\item The QROAM uses $(\chunk-1)\wid+\lceil\log(d/\chunk)\rceil= 2{,}723$ qubits.
These are erased, so we need only count the cost of those used in step 4 (since those may reuse the qubits used for step 3).
\item Each state preparation needs $\mu=27$ qubits to store a superposition state to perform the inequality comparison, and a flag qubit.
We need not count the cost of the qubits used in step 4, because we may reuse qubits from the QROAM.
\item The number of qubits needed for the phase estimation $m = 26$.
\end{enumerate}
The total number of qubits is therefore 3,143.

Lastly, we consider the costs with the {\caltech} orbitals with the sparse state preparation approach.
The minor Toffoli costs are as follows.
\begin{enumerate}
\item $4(N+\lceil \log N\rceil)=640$ Toffolis for the controlled operations.
\item The equal superposition state to prepare is on 18 qubits, and to obtain final amplitude close to one we use 3 ancilla qubits and prepare a superposition over 3 of the 8 basis states. The single step of amplitude amplification takes 142 Toffolis.
\item The cost for the preparation plus inverse preparation on 30 qubits is $2(\mu + 30) = 108$.
\item The controlled swaps used to generate the symmetries have a cost of $4\lceil \log (N/2)\rceil=28$.
\end{enumerate}
These minor costs have a total of 918.

For the number of qubits, we have the following.
\begin{enumerate}
\item The $N=152$ system registers.
\item The state being prepared has size $7+4\lceil\log(N/2)\rceil=35$ qubits.
\item Preparation of the equal superposition state uses 3 ancilla qubits, plus there is a flag qubit for success.
\item The 18 qubits that we iterate over for the QROAM.
\item The qubits needed for the QROAM $\chunk_1\wid$, which includes the output $\wid$ registers and another $(\chunk_1-1)\wid$ working output registers.
We need to subtract $2+4\lceil\log(N/2)\rceil=30$ qubits that are part of the state being prepared that we have already accounted for.
Here $\chunk=32$ and $M=84$ giving 2,658.
\item The QROAM uses another $\lceil\log(d/\chunk_1)\rceil=13$ clean qubits.
\item Each state preparation needs $\mu=24$ qubits to store a superposition state to perform the inequality comparison.
These do not add to the cost because we can reuse qubits used in the QROAM.
\item The number of qubits needed for the phase estimation $m = 24$.
\end{enumerate}
These give a total of 2,904.
\FloatBarrier

\section{Preparation of equal superposition states}
\label{app:prep}
Here we explain in more detail how the number of Toffolis to prepare equal superposition states were determined.
In the case of low rank factorization,
we aim to prepare an equal superposition state
\begin{equation}
\frac 1{(N^2/8+N/4)\sqrt{L+1}}    \sum_{\ell=0}^L \sum_{p=0}^{N/2-1} \sum_{q=0}^{p} \sum_{r=0}^{N/2-1} \sum_{s=0}^{r} \ket{\ell,p,q,r,s}.
\end{equation}
Here we take all variables to start at zero, because that is how they would be encoded in practice.
(In the body of the paper we took $p,q,r,s$ to start from 1 for simplicity.)
A way to prepare this state is to initially use Hadmards to prepare equal superpositions over all registers, over larger ranges that are powers of 2.
Then we perform inequality tests to check that $\ell\le L$, $p\ge q$, $p<N/2$, $r\ge s$, $r< N/2$ are satisfied.
The amplitude for success can then be brought close to 1 by amplitude amplification.

With $N=108$ (the {\micro} orbitals), there are 6 qubits for each of $p,q,r,s$, and 8 qubits for $\ell$, for a total of 32. That means the state flagged by success is
\begin{equation}
\frac 1{\sqrt{2^{32}}}    \sum_{\ell=0}^L \sum_{p=0}^{N/2-1} \sum_{q=0}^{p} \sum_{r=0}^{N/2-1} \sum_{s=0}^{r} \ket{\ell,p,q,r,s}.
\end{equation}
This has amplitude of
\begin{equation}
   \sin\phi= \frac{(N^2/8+N/4)\sqrt{L+1}}{2^{16}}\approx 0.321 .
\end{equation}
If we were to use $\step=2$ steps of amplitude amplification, then the amplitude would be
\begin{equation}
    \sin((2\step+2)\phi) \approx 0.998 .
\end{equation}
Although that is close to 1, we can do considerably better by introducing another 4 qubits, and preparing an equal superposition of 15 basis states out of 16 for those.
Then the initial amplitude is
\begin{equation}
   \sin\phi= \frac{(N^2/8+N/4)\sqrt{15(L+1)}}{2^{18}}\approx 0.311 .
\end{equation}
Then the amplitude after two steps of amplitude amplification is
\begin{equation}
    \sin((2\step +2)\phi) \approx 0.99994 .
\end{equation}

For the Toffoli cost, there are four facts we need to keep in mind (see \app{arith}).
\begin{enumerate}
    \item To perform an inequality test between two variables (numbers encoded in quantum registers) with an equal number of qubits, the number of Toffolis required is equal to the number of qubits.
    \item To perform an inequality test between a variable and a constant given classically (like $N/2$) requires a number of Toffolis equal to the number of qubits \emph{minus} one.
    \item If there is an inequality test between a variable and a constant that is a multiple of a power of 2, then the number of Toffolis is decreased by that power.
    \item A reflection about zero on a register takes a number of Toffolis two less than the number of qubits.
    This is because it is equivalent to a multiply controlled Toffoli, with one of the qubits as target.
\end{enumerate}

At the end we will flag on success for the inequality tests for the state we aim to prepare.
This is because there is not perfect amplitude for success, and we aim to eliminate the error from the cases where the inequality tests are not satisfied.
However, we do \emph{not} need to perform the final inequality test for the ancilla, because that was just to adjust the amplitude.

The costs in the state preparation are as follows.
\begin{enumerate}
\item The inequality tests between variables, $p\ge q$ and $r\ge s$, take 6 Toffolis each, for a total of 12.
\item The inequality tests $p<N/2$ and $r< N/2$ would take 5 Toffolis, except we can save a Toffoli on each because the constant is $N/2=54=2\times 27$.
These two therefore cost 8 Toffolis.
\item The inequality test $\ell\le L$ is with a constant on 8 qubits so has a cost of 7.
We could also save some Toffolis if we were to choose $L+1=200$ instead of 201, but we will not do that for consistency with the rest of the paper.
\item The inequality test with 15 for the ancilla has a cost of 3.
\item There is a total of $8+6\times 4+4=36$ qubits, so reflection about the entire state has a cost of 34.
\item There are 6 inequality tests, so reflection about success has a cost of 4 Toffolis.
\item At the end we wish to flag on success of the 5 inequality tests for the state, which has a cost of 4.
\end{enumerate}

The total cost of the inequality tests for the state is 27.
For $\step$ steps of amplitude amplification, these inequality tests are performed $2\step+1$ times, for a cost of $27(2\step+1)$.
The inequality test on the ancilla is performed $2\step$ times, for a cost of $6\step$.
The two reflections have a combined cost of $34+4=38$.  They are performed once for each step of amplitude amplification for a cost of $38\step$.
Together with the final check to produce the flag qubit, the cost is
$30(2\step+1)+6\step+38\step+4$.
With two steps of amplitude amplification the cost is 227.
Because we need preparation and inverse preparation for each step in the LCU approach, we have a cost of 454.

With the {\caltech} orbitals $N=152$, and we will find it convenient to prepare the superpositions over $\ell,p,q$ and $r,s$ separately.
Each of the $p,q,r,s$ registers now has 7 qubits, and $\ell$ still has 8 qubits.
For the preparation over $\ell,p,q$, we prepare an equal superposition over all registers, then perform inequality tests to give a state flagged on success
\begin{equation}
\frac 1{\sqrt{2^{22}}}    \sum_{\ell=0}^L \sum_{p=0}^{N/2-1} \sum_{q=0}^{p} \ket{\ell,p,q}
\end{equation}
This has amplitude of
\begin{equation}
   \sin\phi= \frac{\sqrt{(N^2/8+N/4)(L+1)}}{2^{11}}\approx 0.374 .
\end{equation}
In this case we will also prepare a superposition over 11 out of 16 basis states on another 4 qubits.
Then the initial amplitude is
\begin{equation}
   \sin\phi= \frac{\sqrt{11(N^2/8+N/4)(L+1)}}{2^{13}}\approx 0.310 .
\end{equation}
The amplitude after $\step=2$ steps of amplitude amplification is
\begin{equation}
    \sin((2\step+2)\phi) \approx 0.99997 .
\end{equation}
The costs are as follows.
\begin{enumerate}
\item The inequality test $p\ge q$ has a cost of 7.
\item In this case $N/2=152/2=76 = 4\times 19$.
This means we can save two Toffolis, and the cost for the inequality test $p<N/2$ is 4.
\item The cost of the inequality test $\ell\le L$ is 7.
\item The inequality test on the ancilla has a cost of 3.
\item The reflection about zero for the $8+2\times 7+4=26$ qubit state has cost of 24.
\item The reflection for the output qubits from the 4 inequality tests has cost 2.
\item Flagging success of all inequality tests for the state has cost 2.
\end{enumerate}
There is cost $7+4+7=18$ for the inequality tests for the state, 3 for the ancilla inequality test, and $24+2=26$ for the reflections, giving total cost $18(2\step+1)+6\step+26\step + 1=155$.
Because we need preparation and inverse preparation that gives a cost of 310.

Next, for the preparation of $r,s$ the inequality tests give a state flagged on success of
\begin{equation}
\frac 1{\sqrt{2^{22}}}  \sum_{r=0}^{N/2-1} \sum_{s=0}^{r} \ket{\ell,r,s}
\end{equation}
with amplitude
\begin{equation}
   \sin\phi= \frac{\sqrt{N^2/8+N/4}}{2^7}\approx 0.423 .
\end{equation}
We use another 5 qubits, and prepare a superposition of 17 basis states out of 32 for those.
That gives amplitude
\begin{equation}
   \sin\phi= \sqrt{\frac{17(N^2/8+N/4)}{2^{19}}}\approx 0.308 .
\end{equation}
The two steps of amplitude amplification give
\begin{equation}
    \sin((2\step+2)\phi) \approx 0.999986 .
\end{equation}
The costs are as follows.
\begin{enumerate}
\item The inequality test $r\ge s$ has a cost of 7.
\item The inequality test $r<N/2$ has cost 4.
\item The inequality test on the ancilla has a cost of 4.
\item The reflection about zero for the $2\times 7+5=19$ qubit state has cost of 17.
\item The reflection for the output qubits from the 3 inequality tests has cost 1.
\item Flagging success of inequality tests has cost 1.
\end{enumerate}
There is cost $7+4=11$ for the inequality tests for the state, 4 for the ancilla, and $19+1=20$ for the reflections, giving total cost $11(2\step+1)+8\step+20\step + 1=112$.
Taking account of the forward and reverse preparation we need 224.
Adding to that the 310 Toffolis for preparing the equal superposition over $\ell,p,q$ we get 534.

Next, we consider the preparation of the equal superposition state over a single register for the sparse preparation.
For $N=108$ for the {\micro} orbitals, we need to prepare an equal superposition over 436,508 states, which can be represented on 19 qubits.
We also prepare an equal superposition over 19 basis states out of 64 for 6 qubits.
Then we have an initial amplitude of
\begin{equation}
   \sin\phi= \sqrt{\frac{19\times 436508}{2^{25}}} \approx 0.497 .
\end{equation}
The amplitude after a \emph{single} step of amplitude amplification is
\begin{equation}
    \sin((2\step+2)\phi) \approx 0.99995 .
\end{equation}
The costs are as follows.
\begin{enumerate}
\item Because $436508=4\times 109127$, we can save 2 Toffolis and the main inequality test takes 16 Toffolis. That is done three times for a cost of 48.
\item The inequality test on the ancilla takes 5 Toffolis, and is done twice for a cost of 10.
\item There is a reflection on all qubits.  There are $19+5=24$, so the cost is 22.
\item There is a reflection on two qubits output from the inequality test, with no cost.
\item We need only flag success of the main inequality test, with no extra cost.
\end{enumerate}
The total is therefore $48+10+22=80$.
The multiplying by 2 for the preparation and inverse preparation gives 160.

In preparing the {\caltech} orbitals with $N=152$, we need a superposition over 179,498 basis states, which needs 18 qubits.
We create a superposition over 3 basis states out of 8 on a 3-qubit ancilla to give an initial amplitude of
\begin{equation}
   \sin\phi= \sqrt{\frac{3\times 179498}{2^{21}}} \approx 0.507 .
\end{equation}
A single step of amplitude amplification gives amplitude
\begin{equation}
    \sin((2\step+2)\phi) \approx 0.9997 .
\end{equation}
The costs are as follows.
\begin{enumerate}
\item Since 179,498 is a multiple of 2, we need 16 Toffolis for that inequality test. That is done three times for a cost of 48.
\item The inequality test on the ancilla takes 2 Toffolis, and is done twice for a cost of 4.
\item There is a reflection on all qubits.  There are $18+3=21$, so the cost is 19.
\item There is a reflection on two qubits output from the inequality test, with no cost.
\item We need only flag success of the main inequality test, with no extra cost.
\end{enumerate}
The total of these costs is 71.
Multiplying by 2 to account for preparation and inverse preparation gives us 142.
\FloatBarrier

\section{Complexity of computing $s$}
\label{app:arithmetic}
Here we determine the complexity of computing $s$ in \eq{qromreg}, which is
\begin{equation}
s = \ell (N^2/8+N/4) + p(p+1)/2 + q .
\end{equation}
The strategy to compute this function efficiently is as follows.
\begin{enumerate}
\item Copy $p$ into the output register.
\item For qubit $k$ encoding $p$, use it to control addition of $2^kp$ to the output register.
After this the value in the output register is $p(p+1)$.
\item Discard the least significant qubit in the output register. This qubit must be zero, and the effect of discarding it is to divide by 2, giving $p(p+1)/2$. 
\item Add $q$ to the output register.
\item For bit $k$ of $N^2/8+N/4$, if it is nonzero add $2^k\ell$ to the output register.
\end{enumerate}

For the controlled addition we can control copying of the register to be added to a new register,
then control addition of that ancilla register to the output register.
Normally, if there are $n$ qubits then that will result in $n$ Toffolis in addition to the Toffolis for the addition.
The ancilla register can be erased without further non-Clifford gates by measurement in the $X$ basis and classically controlled Clifford gates \cite{GidneyAdder}.
In our case, we can save one Toffoli because we are controlling copying $p$ by a qubit from $p$.
Copying that qubit of $p$ to the ancilla controlled on itself is just achieved by a CNOT gate.
The cost of the additions just corresponds to the number of qubits that need be acted upon minus one.

In the case of the multiplication by a constant, the additions we have a sequence of classically controlled additions.
For example, for $N=108$, where $N^2/8+N/4$ is equal to $1485=1024+256+128+64+8+4+1$.
That means, to add $\ell (N^2/8+N/4)$, we add $\ell$, then $4\ell$, then $8\ell$, and so forth.
When we are adding a multiple of a power of 2, then we save a number of Toffolis corresponding to that power, because we do not need to act on the less-significant qubits.
See \app{arith} for an in-depth discussion of the exact costs of addition and subtraction.
In general we aim to save Toffolis by adding numbers that are multiples of larger powers of 2 towards the end.
This is because the number may be larger, but we save Toffolis because we do not need to act on the less-significant qubits.
We can improve the complexity slightly be adding $q$ earlier in the computation than indicated above, but it has a less clear interpretation because the $q$ is added during the multiplication.

First let us consider $N=108$, so $p$ and $q$ are encoded in $n=6$ qubits and have maximum possible values of 53.
The sequence of operations and their costs are then as described below.
In the following we use $p_0$, $p_1$, and so forth to denote the successive qubits encoding $p$.
The expression $p(p+1)/2$ can be written in terms of these values as
\begin{align}
    &p(p+1)/2\nn
    &=p/2+p_0 p/2 + p_1 p + 2 p_2 p + 2^2 p_3 p + 2^3 p_4 p + 2^4 p_5 p \nn
    &=[p_0+2p_1+2^2p_2+2^3p_3+2^4p_4+2^5p_5]/2+[p_0+2p_0(2p_1+2^2p_2+2^3p_3+2^4p_4+2^5p_5)]/2 \nn &\quad +2p_1(2p_1+2(2^2p_2+2^3p_3+2^4p_4+2^5p_5)]/2+2^2p_2[2^2p_2+2(2^3p_3+2^4p_4+2^5p_5)]/2\nn &\quad +2^3p_3[2^3p_3+2(2^4p_4+2^5p_5)]/2+2^4p_4[2^4p_4+2(2^5p_5)]/2+2^{10}p_5/2 \label{eq:secline} \\
    &=[p_1+2p_2+2^2p_3+2^3p_4+2^4p_5]+[p_0+2p_0p_1+2^2p_0p_2+2^3p_0p_3+2^4p_0p_4+2^5p_0p_5]\nn &\quad +2[p_1+2^2p_1p_2+2^3p_1p_3+2^4p_1p_4+2^5p_1p_5]+2^3[p_2+2^2p_2p_3+2^3p_2p_4+2^4p_2p_5]\nn &\quad +2^5[p_3+2^2p_3p_4+2^3p_3p_5]+2^7[p_4+2^2p_4p_5]+2^9p_5 \, .\label{eq:bitform}
\end{align}
In \eq{secline} we have combined the terms like $p_0p_1$, $p_1p_0$ and given a factor of 2.
\begin{enumerate}
\item Copy $p_1$ to $p_5$ into the output ancilla, which has no Toffoli cost. That gives the first term in square brackets in \eq{bitform}.
\item Controlled on $p_0$, copy $p$ into the working ancilla, with cost $n-1=5$. That gives the second term in square brackets in \eq{bitform}.
\item Add the working ancilla to the output ancilla, with cost $n=6$, then erase the working ancilla with Clifford gates.
\item Add $q$ to the output.  The maximum value in the output is now $(3/2)p+q$ which has a maximum of $132$, so has 8 bits.  The cost of the addition is therefore 7.
\item Controlled on $p_1$, copy $p_1$ to $p_5$ into working ancilla, with cost $4$. We use one zeroed qubit between $p_1$ and $p_1p_2$, because the third term in brackets in \eq{bitform} has $p_1$ then $2^2p_1p_2$.
\item Add the working ancilla [corresponding to the third term in brackets in \eq{bitform}] times 2 to the output.  The maximum value in the output is now $227$, and needs 8 bits, and the quantity being added is a multiple of $2$, so the cost is 6.
\item Controlled on $p_2$, copy $p_2$ to $p_5$ into working ancilla, with cost $3$.
\item Add the working ancilla [corresponding to the fourth term in brackets in \eq{bitform}] times $2^3$ to the output.  The maximum value in the output is now $381$, and needs 9 bits. Since we are adding a multiple of $2^3$ the cost is 5.
\item Controlled on $p_3$, copy $p_3$ to $p_5$ into the working ancilla, with cost $2$.
\item Add the working ancilla [corresponding to the fifth term in brackets in \eq{bitform}] times $2^5$ to the output.  The maximum value in the output is now $669$, and needs 10 bits. Since we are adding a multiple of $2^5$ the cost is 4.
\item Controlled on $p_4$, copy $p_4$ and $p_5$ into working ancilla, with cost $1$.
\item Add the working ancilla [corresponding to the sixth term in brackets in \eq{bitform}] times $2^7$ to the output.  The maximum value in the output is now $972$, and needs 10 bits. Since we are adding a multiple of $2^7$ the cost is 2.
\item Copy $p_5$ into working ancilla, with no Toffoli cost.
\item Add the working ancilla times $2^9$ to the output.  The maximum value in the output is now $N^2/8+N/4-1=1484$, and needs 11 bits. Since we are adding a multiple of $2^9$ the cost is 1.
\item Add $\ell$.  After adding the maximum value is $1484+L=1684$, which needs 11 bits.  The cost is 10.
\item Add $4\ell$.  After that maximum value is $1684+4L=2484$ which needs 12 bits, and we need to act on $12-2=10$ bits, so the cost is 9.
\item Add $8\ell$.  The maximum value is $2484+8L=4084$ with 12 bits.  We act on $12-3=9$ bits, so the cost is 8.
\item Add $64\ell$ for a maximum value of $4084+64L=16884$ with 15 bits.  We act on $15-6=9$ bits, for a cost of 8.
\item Add $128\ell$ for a maximum of $16884 + 128L =42484$ and 16 bits.  We act on $16-7=9$ bits for a cost of 8.
\item Add $256\ell$ for a maximum of $42484+256L=93684$ and 17 bits.  We act on $17-8=9$ bits for a cost of 8.
\item Add $1024\ell$ for a maximum of $93684+1024L=298484$ and 19 bits.  We act on $19-10=9$ bits for a cost of 8.
\end{enumerate}
Adding all these costs together gives a total Toffoli cost of $5+6+7+4+6+3+5+2+4+1+2+1+10+9+8\times 5=105$.

Next consider $N=152$, so $p$ and $q$ are encoded in $n=7$ qubits and have maximum possible values of 75.
Then we have $N^2/8+N/4=2926$.
In this case it is more efficient to use the fact that $2926=2048 + 1024 - 128 - 16 - 2$ and perform some subtractions.
The expression $p(p+1)/2$ can be written in terms of $p_j$ as
\begin{align}
    p(p+1)/2&=[p_1+2p_2+2^2p_3+2^3p_4+2^4p_5+2^5p_6]\nn
    &\quad +[p_0+2p_0p_1+2^2p_0p_2+2^3p_0p_3+2^4p_0p_4+2^5p_0p_5+2^6p_0p_6]\nn &\quad +2[p_1+2^2p_1p_2+2^3p_1p_3+2^4p_1p_4+2^5p_1p_5+2^6p_1p_6]\nn & \quad +2^3[p_2+2^2p_2p_3+2^3p_2p_4+2^4p_2p_5+2^5p_2p_6]\nn &\quad +2^5[p_3+2^2p_3p_4+2^3p_3p_5+2^4p_3p_6]+2^7[p_4+2^2p_4p_5+2^3p_4p_6]\nn &\quad +2^9[p_5+2^2p_5p_6]+2^{11}p_6 \, .\label{eq:bitform2}
\end{align}
Then the sequence of elementary steps we need to perform is as follows.
\begin{enumerate}
\item Copy $p_1$ to $p_6$ into the output ancilla, which has no Toffoli cost. That gives the first term in square brackets in \eq{bitform2}.
\item Controlled on $p_0$, copy $p$ into the working ancilla, with cost $n-1=6$. That gives the second term in square brackets in \eq{bitform2}.
\item Add the working ancilla to the output ancilla. The maximum value is $113$, which needs 7 bits so the cost is $6$.
\item Add $q$ to the output.  The maximum value in the output is now $(3/2)p+q$ which has a maximum of $187$, so has 8 bits.  The cost of the addition is therefore 7.
\item Controlled on $p_1$, copy $p_1$ to $p_6$ into working ancilla, with cost $5$.
\item Add the working ancilla [corresponding to the third term in brackets in \eq{bitform2}] times 2 to the output.  The maximum value in the output is now $333$, and needs 9 bits. Since we are adding a multiple of 2 the cost is 7.
\item Controlled on $p_2$, copy $p_2$ to $p_6$ into working ancilla, with cost $4$.
\item Add the working ancilla [corresponding to the fourth term in brackets in \eq{bitform2}] times $2^3$ to the output.  The maximum value in the output is now $583$, and needs 10 bits. Since we are adding a multiple of $2^3$ the cost is 6.
\item Controlled on $p_3$, copy $p_3$ to $p_6$ into the working ancilla, with cost $3$.
\item Add the working ancilla [corresponding to the fifth term in brackets in \eq{bitform2}] times $2^5$ to the output.  The maximum value in the output is now $939$, and needs 10 bits. Since we are adding a multiple of $2^5$ the cost is 4.
\item Controlled on $p_4$, copy $p_4$ to $p_6$ into working ancilla, with cost $2$.
\item Add the working ancilla [corresponding to the fifth term in brackets in \eq{bitform2}] times $2^7$ to the output.  The maximum value in the output is now $1579$, and needs 11 bits. Since we are adding a multiple of $2^7$ the cost is 3.
\item Controlled on $p_5$, copy $p_5$ and $p_6$ into working ancilla, with cost $1$.
\item Add the working ancilla [corresponding to the fifth term in brackets in \eq{bitform2}] times $2^9$ to the output.  The maximum value in the output is now $2091$, and needs 12 bits. Since we are adding a multiple of $2^9$ the cost is 2.
\item Copy $p_6$ into working ancilla, with no Toffoli cost.
\item Add the working ancilla times $2^{11}$ to the output.  The maximum value in the output is now $N^2/8+N/4-1=2925$, and needs 12 bits. Since we are adding a multiple of $2^{11}$ the addition can be performed with Cliffords.
\item Add $1024\ell$.  After adding the maximum value is $2925+1024L=207725$, which needs 18 bits.  Since we added a multiple of $2^{10}$ the cost is 7.
\item Subtract $128\ell$.  After that maximum value is $207725-128L=182125$, and we need to act on $18-7=11$ bits, so the cost is 10.
\item Subtract $16\ell$.  The maximum value is $182125-16L=178925$, and we act on $17-4=13$ bits, so the cost is 12.
\item Subtract $2\ell$ for a maximum value of $178925-2L=178525$, and we act on $17-1=16$ bits, for a cost of 15.
\item Add $2048\ell$ for a maximum of $178525+2048L=588125$ and 20 bits.  We act on $20-11=9$ bits for a cost of 8.
\end{enumerate}
Adding all these costs together gives a total Toffoli cost of $6+6+7+5+7+4+6+3+4+2+3+1+2+7+10+12+15+8=108$.
\FloatBarrier

\section{Costs of addition, subtraction and inequality tests}
\label{app:arith}
Here we describe the exact costs for addition and subtraction and inequality testing.
In \fig{adder}, the addition circuit from \cite{GidneyAdder} is reproduced.
This circuit is for addition of 5-qubit variables modulo $2^5$.
In the following we will refer to numbers given in quantum registers as ``variables'', and numbers given classically as ``constants''.
It can be seen that the cost of this circuit is 4 Toffolis (those on the right can be performed with measurements and Clifford gates).
More generally, if there is modular addition (with the modulus $2^n$) on $n$-qubit variables, the Toffoli cost is $n-1$.
The same circuit can be used for non-modular addition if it is known that the carry qubit would be zero (in this case the number would be less than $2^5$).
Then the cost for addition of $n$-qubit variables without a carry qubit is therefore $n-1$.

\begin{figure}\centerline{
\Qcircuit @C=0.9em @R=.6em {
\lstick{i_0} & \ctrl{2} & \qw & \qw& \qw & \qw & \qw& \qw & \qw & \qw& \qw & \qw & \qw & \qw & \qw& \qw &  \qw & \qw& \qw & \qw & \qw & \ctrl{2} & \ctrl{1} & \rstick{i_0} \qw \\
\lstick{t_0} & \control \qw & \qw & \qw& \qw &  \qw & \qw& \qw & \qw & \qw & \qw & \qw& \qw &  \qw & \qw& \qw & \qw & \qw &  \qw & \qw& \qw & \control \qw & \targ & \rstick{(t+i)_0} \qw \\
&  & \ctrl{2} & \qw& \ctrl{3} &  \qw & \qw & \qw& \qw & \qw & \qw &  \qw & \qw & \qw& \qw & \qw & \qw & \qw & \ctrl{3} & \qw & \ctrl{1} & \qw \\
\lstick{i_1} & \qw & \targ & \ctrl{2}& \qw &  \qw & \qw & \qw& \qw & \qw & \qw  & \qw & \qw & \qw & \qw & \qw & \qw & \qw & \qw & \ctrl{2} & \targ & \qw & \ctrl{1} & \rstick{i_1} \qw  \\
\lstick{t_1} & \qw & \targ & \control\qw& \qw &  \qw & \qw & \qw& \qw & \qw & \qw & \qw & \qw & \qw & \qw & \qw & \qw & \qw & \qw & \control \qw & \qw & \qw & \targ & \rstick{(t+i)_1} \qw \\
&  & & & \targ &  \ctrl{2} & \qw& \ctrl{3} & \qw &  \qw & \qw & \qw& \qw & \qw & \qw & \ctrl{3} & \qw & \ctrl{1} & \targ & \qw \\
\lstick{i_2} & \qw & \qw & \qw&\qw & \targ &  \ctrl{2} & \qw & \qw& \qw & \qw & \qw & \qw & \qw & \qw & \qw & \ctrl{2} & \targ & \qw & \qw & \qw & \qw & \ctrl{1} & \rstick{i_2} \qw \\
\lstick{t_2} & \qw & \qw & \qw& \qw &\targ &\control  \qw & \qw & \qw& \qw & \qw & \qw & \qw & \qw & \qw & \qw & \control\qw & \qw & \qw & \qw & \qw & \qw & \targ & \rstick{(t+i)_2} \qw \\
& & &  & & & & \targ &  \ctrl{2} & \qw& \ctrl{3} & \qw &  \ctrl{3} & \qw & \ctrl{1} & \targ & \qw  \\
\lstick{i_3} & \qw & \qw & \qw& \qw & \qw & \qw&\qw & \targ &  \ctrl{2} & \qw &  \qw& \qw & \ctrl{2} & \targ & \qw & \qw & \qw & \qw & \qw & \qw & \qw & \ctrl{1} & \rstick{i_3} \qw \\
\lstick{t_3} & \qw & \qw & \qw& \qw & \qw & \qw& \qw &\targ &  \control \qw & \qw &  \qw &  \qw & \control \qw & \qw & \qw & \qw & \qw & \qw & \qw & \qw & \qw & \targ & \rstick{(t+i)_3} \qw \\
& & & & & &  & & & & \targ &  \ctrl{2} & \targ & \qw &  \\
\lstick{i_4} & \qw & \qw & \qw & \qw & \qw & \qw& \qw & \qw & \qw&\qw & \qw &  \qw & \qw & \qw &  \qw & \qw & \qw &  \qw & \qw & \qw & \qw & \ctrl{1} & \rstick{i_4} \qw \\
\lstick{t_4} & \qw & \qw & \qw & \qw & \qw & \qw& \qw & \qw & \qw& \qw &\targ &  \qw & \qw & \qw &  \qw & \qw & \qw & \qw & \qw & \qw & \qw & \targ & \rstick{(t+i)_4} \qw
}}\caption{\label{fig:adder}A circuit to perform addition on 5 qubits modulo $2^5$. This circuit is also sufficient for non-modular addition if we know there will not be overflow.}
\end{figure}
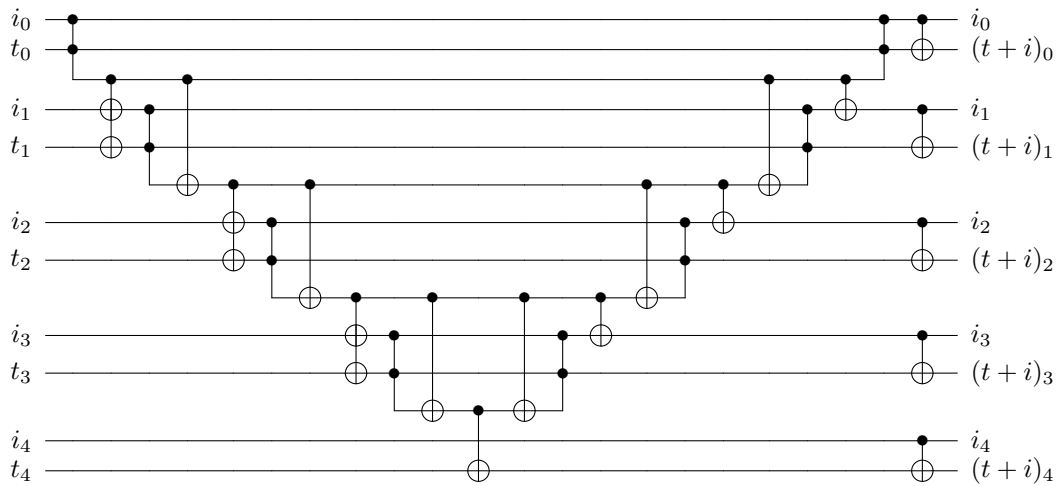

Alternatively, if we wish to perform non-modular addition, and compute a carry qubit, then the circuit will be as in \fig{addcarry}.
That shows addition of 4-qubit variables, with a Toffoli cost of 4.
More generally, the addition of $n$-qubit variables with a carry qubit has cost $n$.
In either case (with or without the carry qubit) we can make further savings if one of the numbers being added is known to be a multiple of a power of 2.
Say $i$ is a multiple of 2, so the final digit $i_0$ is 0.
Then the first carry qubit (the third line in \fig{adder} and \fig{addcarry}) is zero, and the CNOTs it controls have no effect and may be omitted.
As a result, we may perform the circuit as before except starting from $i_1$ and $t_1$ instead of $i_0$ and $t_0$.
This means that we save a Toffoli.
More generally, if there are $k$ trailing zeros for one of the numbers (it is a multiple of $2^k$), then we may save $k$ Toffolis.

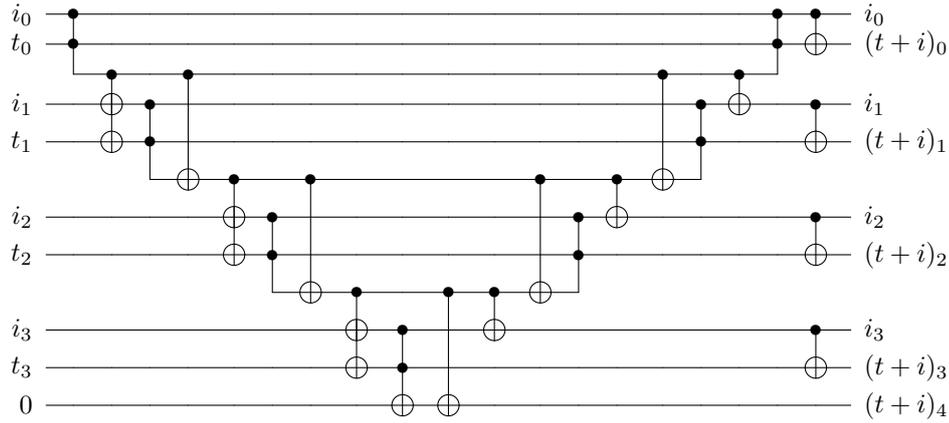
\begin{figure}\centerline{
\Qcircuit @C=0.9em @R=.6em {
\lstick{i_0} & \ctrl{2} & \qw & \qw& \qw & \qw & \qw& \qw & \qw  & \qw & \qw & \qw& \qw &  \qw & \qw& \qw & \qw & \qw & \ctrl{2} & \ctrl{1} & \rstick{i_0} \qw \\
\lstick{t_0} & \control \qw & \qw & \qw& \qw &  \qw & \qw& \qw  & \qw & \qw &  \qw & \qw& \qw & \qw & \qw &  \qw & \qw& \qw & \control \qw & \targ & \rstick{(t+i)_0} \qw \\
& & \ctrl{2} & \qw& \ctrl{3} &  \qw & \qw & \qw& \qw & \qw & \qw & \qw& \qw & \qw & \qw & \ctrl{3} & \qw & \ctrl{1} & \qw \\
\lstick{i_1} & \qw & \targ & \ctrl{2}& \qw &  \qw & \qw & \qw& \qw   & \qw & \qw & \qw & \qw & \qw & \qw & \qw & \ctrl{2} & \targ & \qw & \ctrl{1} & \rstick{i_1} \qw  \\
\lstick{t_1} & \qw & \targ & \control\qw& \qw &  \qw & \qw & \qw& \qw  & \qw & \qw & \qw & \qw & \qw & \qw & \qw & \control \qw & \qw & \qw & \targ & \rstick{(t+i)_1} \qw \\
 & & & & \targ &  \ctrl{2} & \qw& \ctrl{3} & \qw &  \qw & \qw & \qw & \ctrl{3} & \qw & \ctrl{1} & \targ & \qw \\
\lstick{i_2} & \qw & \qw & \qw&\qw & \targ &  \ctrl{2} & \qw & \qw& \qw & \qw & \qw & \qw & \ctrl{2} & \targ & \qw & \qw & \qw & \qw & \ctrl{1} & \rstick{i_2} \qw \\
\lstick{t_2} & \qw & \qw & \qw& \qw &\targ &\control  \qw & \qw & \qw& \qw  & \qw & \qw & \qw & \control\qw & \qw & \qw & \qw & \qw & \qw & \targ & \rstick{(t+i)_2} \qw \\
& & & & & & & \targ &  \ctrl{2} & \qw& \ctrl{3} & \ctrl{1} & \targ & \qw  \\
\lstick{i_3} & \qw & \qw & \qw& \qw & \qw & \qw&\qw & \targ &  \ctrl{2} & \qw & \targ & \qw & \qw & \qw & \qw & \qw & \qw & \qw & \ctrl{1} & \rstick{i_3} \qw \\
\lstick{t_3} & \qw & \qw & \qw& \qw & \qw & \qw& \qw &\targ &  \control \qw & \qw &  \qw &  \qw & \qw & \qw & \qw & \qw & \qw & \qw & \targ & \rstick{(t+i)_3} \qw \\
\lstick{0} & \qw & \qw & \qw & \qw & \qw & \qw & \qw & \qw & \targ  & \targ & \qw & \qw & \qw & \qw & \qw & \qw & \qw & \qw & \qw & \rstick{(t+i)_4} \qw
}}\caption{\label{fig:addcarry}A circuit to perform addition on 4 qubits, with a single carry output bit.}
\end{figure}

Next consider the case where one of the numbers is given classically, so is a constant.
We can use the same circuit, and use NOT gates to prepare $i$ in the desired state corresponding to our classically given number.
In practice we could make further simplifications to reduce the number of gates, but in most cases they do not reduce the number of Toffolis.
One which does is to note the value of $i_0$.
If $i_0=1$, then the first Toffoli can be replaced with a CNOT, saving a single Toffoli.
That means that, for adding a constant, we always have a Toffoli cost 1 less than adding a variable.
Moreover, if $i_0=0$, then we can make the same simplification as described above for variables.
That means that if the constant is a multiple of $2^k$ then we can save $k$ Toffolis.
This saving is in addition to the saving of 1 because we are adding a constant. 

\begin{figure}\centerline{
\Qcircuit @C=0.9em @R=.6em {
\lstick{i_0}  & \ctrl{1}& \ctrl{2}        & \qw     & \qw& \qw & \qw & \qw& \qw & \qw & \qw& \qw & \qw & \qw & \qw & \qw& \qw &  \qw & \qw& \qw & \qw & \qw & \ctrl{2} & \rstick{i_0} \qw \\
\lstick{t_0}  & \targ   & \control \qw & \qw     & \qw& \qw &  \qw & \qw& \qw & \qw & \qw & \qw & \qw& \qw &  \qw & \qw& \qw & \qw & \qw &  \qw & \qw& \qw & \control \qw  & \rstick{(t-i)_0} \qw \\
                  &           &                   & \ctrl{1} & \qw& \ctrl{3} &  \qw & \qw & \qw& \qw & \qw & \qw &  \qw & \qw & \qw& \qw & \qw & \qw & \qw & \ctrl{3} & \qw & \ctrl{2} & \qw \\
\lstick{i_1}  & \ctrl{1}& \qw            & \targ & \ctrl{2}      & \qw &  \qw & \qw & \qw& \qw & \qw & \qw  & \qw & \qw & \qw & \qw & \qw & \qw & \qw & \qw & \ctrl{2}      & \targ & \qw  & \rstick{i_1} \qw  \\
\lstick{t_1} & \targ    & \qw            & \qw   & \control\qw& \qw &  \qw & \qw & \qw& \qw & \qw & \qw & \qw & \qw & \qw & \qw & \qw & \qw & \qw & \qw & \control \qw & \targ & \qw & \rstick{(t-i)_1} \qw \\
&  & & & & \targ &  \ctrl{1} & \qw& \ctrl{3} & \qw &  \qw & \qw & \qw& \qw & \qw & \qw & \ctrl{3} & \qw & \ctrl{2} & \targ & \qw \\
\lstick{i_2} & \ctrl{1} & \qw & \qw & \qw&\qw & \targ &  \ctrl{2} & \qw & \qw& \qw & \qw & \qw & \qw & \qw & \qw & \qw & \ctrl{2} & \targ & \qw & \qw & \qw & \qw  & \rstick{i_2} \qw \\
\lstick{t_2} & \targ & \qw & \qw & \qw& \qw &\qw &\control  \qw & \qw & \qw& \qw & \qw & \qw & \qw & \qw & \qw & \qw & \control\qw & \targ & \qw & \qw & \qw & \qw  & \rstick{(t-i)_2} \qw \\
& & & & & & & & \targ &  \ctrl{1} & \qw& \ctrl{3} & \qw &  \ctrl{3} & \qw & \ctrl{2} & \targ & \qw  \\
\lstick{i_3} & \ctrl{1} & \qw & \qw & \qw& \qw & \qw & \qw&\qw & \targ &  \ctrl{2} & \qw &  \qw& \qw & \ctrl{2} & \targ & \qw & \qw & \qw & \qw & \qw & \qw & \qw & \rstick{i_3} \qw \\
\lstick{t_3} & \targ & \qw & \qw & \qw& \qw & \qw & \qw& \qw &\qw &  \control \qw & \qw &  \qw &  \qw & \control \qw & \targ & \qw & \qw & \qw & \qw & \qw & \qw & \qw & \rstick{(t-i)_3} \qw \\
& & & & & & & & & & & \targ &  \ctrl{2} & \targ & \qw &  \\
\lstick{i_4} & \ctrl{1} & \qw & \qw & \qw & \qw & \qw & \qw& \qw & \qw & \qw&\qw & \qw &  \qw & \qw & \qw &  \qw & \qw & \qw &  \qw & \qw & \qw & \qw  & \rstick{i_4} \qw \\
\lstick{t_4} & \targ & \qw & \qw & \qw & \qw & \qw & \qw& \qw & \qw & \qw& \qw &\targ &  \qw & \qw & \qw &  \qw & \qw & \qw & \qw & \qw & \qw & \qw & \rstick{(t-i)_4} \qw
}}\caption{\label{fig:sub}A circuit to perform subtraction on 5 qubits modulo $2^5$.}
\end{figure}
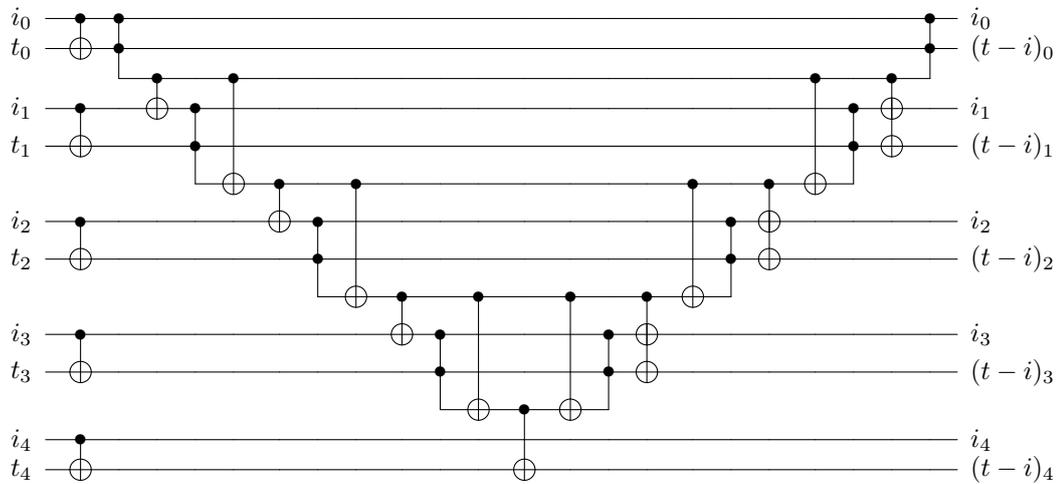

In order to perform subtraction, one can simply reverse the circuit for addition, as shown in \fig{sub}.
This circuit is for modular subtraction of $n$-qubit variables, and can also be used for non-modular subtraction if it is known that $t\ge i$.
The cost to subtract two $n$-qubit numbers is $n-1$.
Again we can make a saving if it is known that $i$ is a multiple of a power of 2.
If $i$ is a multiple of $2^k$, then we can save $k$ Toffolis, via exactly the same reasoning as for addition.
Similarly, if $i$ is a constant, then we can save a further Toffoli.

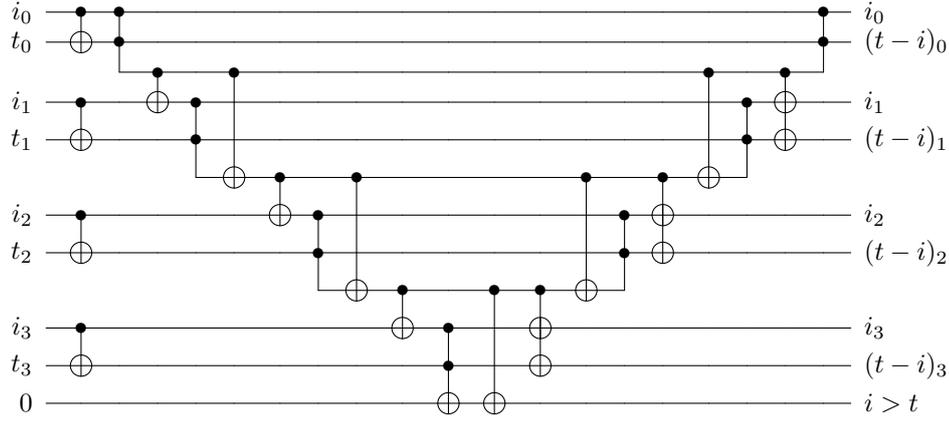
\begin{figure}\centerline{
\Qcircuit @C=0.9em @R=.6em {
\lstick{i_0}  & \ctrl{1}& \ctrl{2}        & \qw     & \qw& \qw & \qw & \qw& \qw & \qw & \qw& \qw & \qw& \qw &  \qw & \qw& \qw & \qw & \qw & \ctrl{2} & \rstick{i_0} \qw \\
\lstick{t_0}  & \targ   & \control \qw & \qw     & \qw& \qw &  \qw & \qw& \qw & \qw & \qw  &  \qw & \qw& \qw & \qw & \qw &  \qw & \qw& \qw & \control \qw  & \rstick{(t-i)_0} \qw \\
                  &           &                   & \ctrl{1} & \qw& \ctrl{3} &  \qw & \qw & \qw& \qw & \qw & \qw& \qw & \qw & \qw & \qw & \ctrl{3} & \qw & \ctrl{2} & \qw \\
\lstick{i_1}  & \ctrl{1}& \qw            & \targ & \ctrl{2}      & \qw &  \qw & \qw & \qw& \qw & \qw & \qw & \qw & \qw & \qw & \qw & \qw & \ctrl{2}      & \targ & \qw  & \rstick{i_1} \qw  \\
\lstick{t_1} & \targ    & \qw            & \qw   & \control\qw& \qw &  \qw & \qw & \qw & \qw & \qw & \qw & \qw & \qw & \qw & \qw & \qw & \control \qw & \targ & \qw & \rstick{(t-i)_1} \qw \\
&  & & & & \targ &  \ctrl{1} & \qw& \ctrl{3} & \qw &  \qw & \qw & \qw& \ctrl{3} & \qw & \ctrl{2} & \targ & \qw \\
\lstick{i_2} & \ctrl{1} & \qw & \qw & \qw&\qw & \targ &  \ctrl{2} & \qw  & \qw & \qw & \qw & \qw & \qw & \ctrl{2} & \targ & \qw & \qw & \qw & \qw  & \rstick{i_2} \qw \\
\lstick{t_2} & \targ & \qw & \qw & \qw& \qw &\qw &\control  \qw & \qw & \qw& \qw  & \qw & \qw & \qw & \control\qw & \targ & \qw & \qw & \qw & \qw  & \rstick{(t-i)_2} \qw \\
& & & & & & & & \targ &  \ctrl{1} & \qw& \ctrl{3} & \ctrl{2} & \targ & \qw  \\
\lstick{i_3} & \ctrl{1} & \qw & \qw & \qw& \qw & \qw & \qw&\qw & \targ &  \ctrl{2} & \qw & \targ & \qw & \qw & \qw & \qw & \qw & \qw & \qw & \rstick{i_3} \qw \\
\lstick{t_3} & \targ & \qw & \qw & \qw& \qw & \qw & \qw& \qw &\qw &  \control \qw & \qw &  \targ &  \qw & \qw & \qw & \qw & \qw & \qw & \qw & \rstick{(t-i)_3} \qw \\
\lstick{0} & \qw & \qw & \qw & \qw & \qw& \qw & \qw & \qw& \qw &\targ &  \targ & \qw &  \qw & \qw & \qw & \qw & \qw & \qw & \qw & \rstick{i>t} \qw
}}\caption{A circuit to perform an inequality test on 4 qubits. The output bit is $(t-i)_4$, and will be 1 if $t<i$. This is because it performs modular subtraction on 5 bits, with $i_4$ and $t_4$ guaranteed to be zero.  The maximum values of $i$ and $t$ are 15.  If we subtract $i$ from $t$ modulo 32, and $t<i$, then the result will be between 17 and 31.  These numbers all have bit 5 equal to 1.}
\end{figure}

As was noted in \cite{Cuccaro}, subtraction can be used for inequality testing as well.
That is, if we perform modular subtraction $t-i$, then the carry qubit will carry the information about whether $i>t$.
In \fig{ineq}, we show the circuit for an inequality test on two 4-qubit variables.
We have taken the circuit in \fig{sub} for two 5-qubit variables, taken $i_5$ and $t_5$ to be zero, and simplified.
There the fourth carry qubit can take the role of the output flagging if $i>t$.
We no longer need the CNOT between that ancilla and the $t_4$ qubit, because that was just to copy out the value of the carry qubit, and we remove the second CNOT and Toffoli with the carry qubit as target, because they had the role of erasing that ancilla in the subtraction circuit.
Note that this circuit still gives the result of the subtraction in other registers.
To restore the $t$ register, we can use the circuit shown in \fig{ineq2}.

Note that with 4-qubit integers, they may represent numbers from 0 to 15, and the result of the subtraction can be from $-15$ to $+15$.
Since the subtraction is performed modulo 32, negative numbers from $-15$ to $-1$ will become 17 to 31, in which case the fifth qubit will be in the state $\ket{1}$.
Since negative numbers result from $t-i<0$, so $t<i$, this carry qubit is flagging the result of the inequality test.
We could also test $t\le i$ by reversing the action of the circuit between $t$ and $i$, giving a qubit which has the result of the inequality test $t>i$.
A NOT gate would give the result of the inequality test $t\le i$.
The inequality test on $n$-qubit variables has cost $n$.
Again, if $i$ is known to be a multiple of $2^k$, then we can save $k$ Toffolis, and if $i$ is a constant then we can save an additional Toffoli.

In summary, the rules for counting costs with $n$-bit integers are as follows.
\begin{enumerate}
\item Addition with a carry qubit has cost $n$.
\item Addition with no carry qubit has cost $n-1$.
\item Subtraction has cost $n-1$.
\item Inequality tests have cost $n$.
\item If $i$ is a multiple of $2^k$, there is a saving of $k$.
\item If $i$ is a constant then there is a saving of 1.
\end{enumerate}
In the case of addition or subtraction, $i$ is the number being added or subtracted from the target register.
For the inequality test we are testing $t<i$.

\begin{figure}[htb]
\centerline{
\Qcircuit @C=0.9em @R=.6em {
\lstick{i_0}  & \ctrl{1}& \ctrl{2}        & \qw     & \qw& \qw & \qw & \qw& \qw & \qw & \qw& \qw & \qw& \qw &  \qw & \qw& \qw & \qw & \qw & \ctrl{2} & \ctrl{1} & \rstick{i_0} \qw \\
\lstick{t_0}  & \targ   & \control \qw & \qw     & \qw& \qw &  \qw & \qw& \qw & \qw & \qw  &  \qw & \qw& \qw & \qw & \qw &  \qw & \qw& \qw & \control \qw  & \targ & \rstick{t_0} \qw \\
                  &           &                   & \ctrl{1} & \qw& \ctrl{3} &  \qw & \qw & \qw& \qw & \qw & \qw& \qw & \qw & \qw & \qw & \ctrl{3} & \qw & \ctrl{1} & \qw \\
\lstick{i_1}  & \ctrl{1}& \qw            & \targ & \ctrl{2}      & \qw &  \qw & \qw & \qw& \qw & \qw & \qw & \qw & \qw & \qw & \qw & \qw & \ctrl{2}      & \targ & \qw  & \ctrl{1} & \rstick{i_1} \qw  \\
\lstick{t_1} & \targ    & \qw            & \qw   & \control\qw& \qw &  \qw & \qw & \qw & \qw & \qw & \qw & \qw & \qw & \qw & \qw & \qw & \control \qw & \qw & \qw & \targ & \rstick{t_1} \qw \\
&  & & & & \targ &  \ctrl{1} & \qw& \ctrl{3} & \qw &  \qw & \qw & \qw& \ctrl{3} & \qw & \ctrl{1} & \targ & \qw \\
\lstick{i_2} & \ctrl{1} & \qw & \qw & \qw&\qw & \targ &  \ctrl{2} & \qw  & \qw & \qw & \qw & \qw & \qw & \ctrl{2} & \targ & \qw & \qw & \qw & \qw  & \ctrl{1} & \rstick{i_2} \qw \\
\lstick{t_2} & \targ & \qw & \qw & \qw& \qw &\qw &\control  \qw & \qw & \qw& \qw  & \qw & \qw & \qw & \control\qw & \qw & \qw & \qw & \qw & \qw  & \targ & \rstick{t_2} \qw \\
& & & & & & & & \targ &  \ctrl{1} & \qw& \ctrl{3} & \ctrl{1} & \targ & \qw  \\
\lstick{i_3} & \ctrl{1} & \qw & \qw & \qw& \qw & \qw & \qw&\qw & \targ &  \ctrl{2} & \qw & \targ & \qw & \qw & \qw & \qw & \qw & \qw & \qw & \ctrl{1} & \rstick{i_3} \qw \\
\lstick{t_3} & \targ & \qw & \qw & \qw& \qw & \qw & \qw& \qw &\qw &  \control \qw & \qw &  \qw &  \qw & \qw & \qw & \qw & \qw & \qw & \qw & \targ & \rstick{t_3} \qw \\
\lstick{0} & \qw & \qw & \qw & \qw & \qw& \qw & \qw & \qw& \qw &\targ &  \targ & \qw &  \qw & \qw & \qw & \qw & \qw & \qw & \qw & \qw & \rstick{t<i} \qw
}}\caption{\label{fig:ineq}A circuit to perform an inequality test on 4 bits. This time the circuit restores the original values in the $t$ registers, so it does not alter this register and just outputs the result of the inequality test.}
\end{figure}
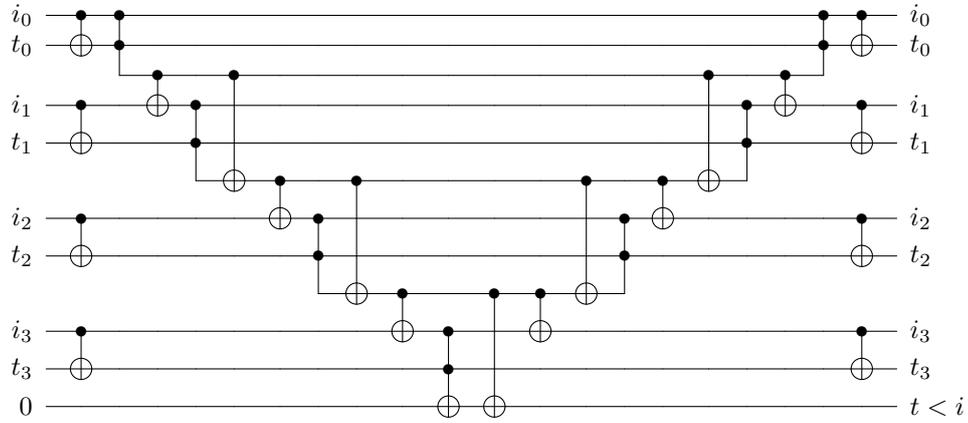

\begin{figure}[htb]
\centerline{
\Qcircuit @C=0.9em @R=.6em {
\lstick{i_0}  & \qw & \qw        & \qw     & \qw& \qw & \qw & \qw& \qw & \qw & \qw& \qw & \qw& \qw &  \qw & \qw& \qw & \qw & \qw & \qw & \qw & \rstick{i_0} \qw \\
\lstick{t_0}  & \targ   & \ctrl{1} & \qw     & \qw& \qw &  \qw & \qw& \qw & \qw & \qw  &  \qw & \qw& \qw & \qw & \qw &  \qw & \qw& \qw & \ctrl{1}  & \targ & \rstick{t_0} \qw \\
                  &           &                   & \ctrl{1} & \qw& \ctrl{3} &  \qw & \qw & \qw& \qw & \qw & \qw& \qw & \qw & \qw & \qw & \ctrl{3} & \qw & \ctrl{1} & \qw \\
\lstick{i_1}  & \ctrl{1}& \qw            & \targ & \ctrl{2}      & \qw &  \qw & \qw & \qw& \qw & \qw & \qw & \qw & \qw & \qw & \qw & \qw & \ctrl{2}      & \targ & \qw  & \ctrl{1} & \rstick{i_1} \qw  \\
\lstick{t_1} & \targ    & \qw            & \qw   & \control\qw& \qw &  \qw & \qw & \qw & \qw & \qw & \qw & \qw & \qw & \qw & \qw & \qw & \control \qw & \qw & \qw & \targ & \rstick{t_1} \qw \\
&  & & & & \targ &  \ctrl{1} & \qw& \ctrl{3} & \qw &  \qw & \qw & \qw& \ctrl{3} & \qw & \ctrl{1} & \targ & \qw \\
\lstick{i_2} & \ctrl{1} & \qw & \qw & \qw&\qw & \targ &  \ctrl{2} & \qw  & \qw & \qw & \qw & \qw & \qw & \ctrl{2} & \targ & \qw & \qw & \qw & \qw  & \ctrl{1} & \rstick{i_2} \qw \\
\lstick{t_2} & \targ & \qw & \qw & \qw& \qw &\qw &\control  \qw & \qw & \qw& \qw  & \qw & \qw & \qw & \control\qw & \qw & \qw & \qw & \qw & \qw  & \targ & \rstick{t_2} \qw \\
& & & & & & & & \targ &  \ctrl{1} & \qw& \ctrl{3} & \ctrl{1} & \targ & \qw  \\
\lstick{i_3} & \ctrl{1} & \qw & \qw & \qw& \qw & \qw & \qw&\qw & \targ &  \ctrl{2} & \qw & \targ & \qw & \qw & \qw & \qw & \qw & \qw & \qw & \ctrl{1} & \rstick{i_3} \qw \\
\lstick{t_3} & \targ & \qw & \qw & \qw& \qw & \qw & \qw& \qw &\qw &  \control \qw & \qw &  \qw &  \qw & \qw & \qw & \qw & \qw & \qw & \qw & \targ & \rstick{t_3} \qw \\
\lstick{0} & \qw & \qw & \qw & \qw & \qw& \qw & \qw & \qw& \qw &\targ &  \targ & \qw &  \qw & \qw & \qw & \qw & \qw & \qw & \qw & \qw & \rstick{t<i} \qw
}}\caption{\label{fig:ineq2}A circuit to perform an inequality test on 4 bits, with $i$ given classically and $i_0$ equal to 1. We can save one Toffoli, and the first ancilla is not needed because we can control directly on $t_0$.}
\end{figure}
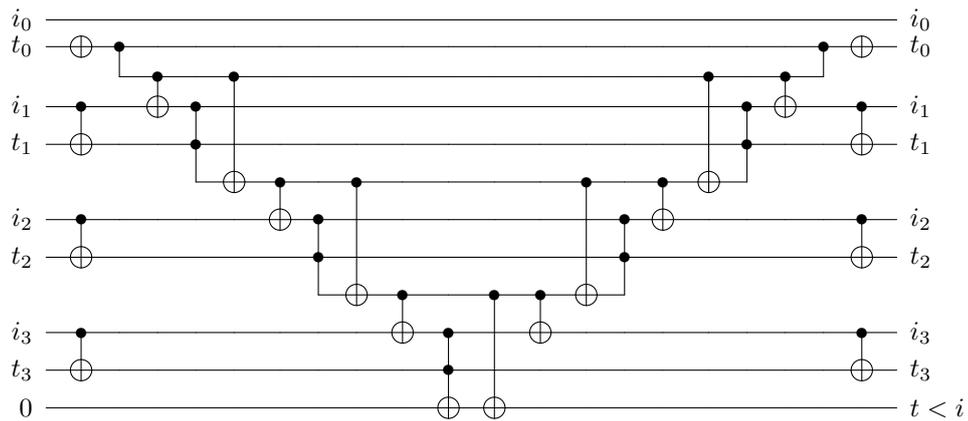

\end{document}